\documentclass[12pt]{article}
\usepackage[margin=1in]{geometry}
\usepackage{graphicx}
\usepackage[dvipsnames]{xcolor}
\usepackage{color}
\usepackage{amssymb}
\usepackage[nointegrals]{wasysym}
\usepackage{amsmath}
\usepackage{amsthm}
\usepackage{bbm}
\usepackage{pifont}

\usepackage{hyperref}
\usepackage{url}

\usepackage{thmtools,thm-restate}

\newtheorem{theorem}{Theorem}
\newtheorem{lemma}[theorem]{Lemma}

\newtheorem{corollary}{Corollary}
\theoremstyle{definition}
\newtheorem{claim}{Claim}
\newtheorem{Defi}{Definition}
\newtheorem{problem}{Problem}
\newtheorem{example}{Example}

\usepackage[authoryear,round,sort]{natbib}

\usepackage{algorithm}
\usepackage{subcaption}
\usepackage{algpseudocode}
\usepackage{booktabs}
\usepackage{tikz}
\usepackage[capitalize]{cleveref}
\crefname{equation}{}{}
\crefname{claim}{Claim}{Claims}
\Crefname{appsec}{Appendix}{Appendices}

\algblockdefx{WithProbability}{EndWithProbability}[1]{\textbf{with probability}
#1}{\textbf{end}}
\renewcommand{\b}[1]{\left[ #1 \right]}
\def\dTV{d_{\mathrm{TV}}}

\def\R{\mathbb{R}}
\def\Zplus{\mathbb{Z}_{\ge 0}}

\def\be{\mathbf{e}}
\def\bs{\mathbf{s}}
\def\br{\mathbf{r}}
\def\bS{\mathbf{S}}

\def\bx{\mathbf{x}}
\def\bX{\mathbb{X}}
\def\by{\mathbf{y}}
\def\bz{\mathbf{z}}
\def\bv{\mathbf{v}}
\def\bV{\mathbf{V}}
\def\mcX{\mc{X}}

\def\bc{\mathbf{c}}
\def\bM{\mathbf{M}}
\def\bS{\mathbf{S}}
\def\bT{\mathbf{T}}

\def\probname{\textbf{MMS}}
\def\sampprobname{\textbf{MMS-Sample}}
\def\relaxation{\tau_{\mathrm{rel}}}
\def\relaxationk{\tau_{\mathrm{rel},k}}
\def\relaxationgamma{\tau_{\mathrm{rel},\gamma}}
\def\mixingtime{\tau_{\mathrm{mix}}}

\def\mixingtimestate{\tau_{\mathrm{mix}}(\varepsilon; P, \bx_0)}

\def\mixinggammalb{t^*_\gamma(\varepsilon)}

\def\numsampk{N_{k}}
\def\numsampgamma{N_{\gamma}}
\def\numsamplbk{\underline{N}_k}
\def\numsampubk{\overline{N}_k}

\def\numsamplbgamma{\underline{N}_\gamma}
\def\numsamplbgammastar{\underline{N}_{\gamma^*}}

\def\numsampubgammaeps{\overline{N}_\gamma(\varepsilon)}
\def\solutiondensity{p_{\gamma}}
\def\solutiondensitystar{p_{\gamma}^*}
\def\Mgamma{M_{\gamma}}
\def\Pgamma{P_{\gamma}}
\def\Qgamma{Q_{\gamma}}
\def\Mreduced{M_{k}}
\def\Preduced{P_{k}}
\def\Preducedtwo{P_{2}}
\def\given{\; | \;}
\def\tpg{\tilde \pi_{\gamma}}

\newcommand{\ALSynMRPct}{2.52}
\newcommand{\ALPUMSMRPct}{3.36}
\newcommand{\ALSynSizeOne}{0.27}
\newcommand{\ALPUMSSizeOne}{0.32}
\newcommand{\ALTVDUnadjustedAll}{0.126}
\newcommand{\ALTVDAdjustedAll}{0.097}

\newcommand{\ALTVDAdjustedAccOne}{0.083}

\newcommand{\ALlzzoTVDUnadjustedAll}{0.091}
\newcommand{\ALlzzoTVDAdjustedAll}{0.077}

\newcommand{\ALlzzoTVDAdjustedAccOne}{0.069}

\newcommand{\ALlzzzzoTVDUnadjustedAll}{0.109}
\newcommand{\ALlzzzzoTVDAdjustedAll}{0.080}

\newcommand{\ALlzzzzoTVDAdjustedAccOne}{0.067}

\newcommand{\NVSynMRPct}{6.27}
\newcommand{\NVPUMSMRPct}{11.59}

\newcommand{\NVTVDUnadjustedAll}{0.229}
\newcommand{\NVTVDAdjustedAll}{0.219}

\newcommand{\NVTVDAdjustedAccOne}{0.205}

\newcommand{\NVlzzoTVDUnadjustedAll}{0.217}
\newcommand{\NVlzzoTVDAdjustedAll}{0.212}

\newcommand{\NVlzzoTVDAdjustedAccOne}{0.198}

\newcommand{\NVlzzzzoTVDUnadjustedAll}{0.222}
\newcommand{\NVlzzzzoTVDAdjustedAll}{0.217}

\newcommand{\NVlzzzzoTVDAdjustedAccOne}{0.201}

\newcommand{\NVmrtwoSynMRPct}{10.59}

\def\Dpums{\mc P}
\def\Drw{\tilde{\Dpums}}

\def\empdist{\hat{\mc Q}}
\def\empdistrw{\tilde{\mc Q}}

\crefname{equation}{}{}
\crefname{claim}{Claim}{Claims}
\newcommand{\incolor}[1]{\textcolor{red}{\mathbf{#1}}}

\algblockdefx{WithProbability}{EndWithProbability}[1]{\textbf{with probability}
#1}{\textbf{end}}
\newcommand{\Plus}{\mathord{\text{\ding{58}}}}
\newcommand{\mc}[1]{\ensuremath{\mathcal{#1}}}
\newcommand\numberthis{\addtocounter{equation}{1}\tag{\theequation}}

\newcommand{\p}[1]{\left( #1 \right)}
\renewcommand{\b}[1]{\left[ #1 \right]}
\DeclareMathOperator*{\argmax}{arg\,max}
\DeclareMathOperator*{\argmin}{arg\,min}

\newcommand{\EE}[2]{\mathbb{E}_{#1} \b{#2}}
\def\dTV{d_{\mathrm{TV}}}

\def\R{\mathbb{R}}
\def\Zplus{\mathbb{Z}_{\ge 0}}

\def\be{\mathbf{e}}
\def\bs{\mathbf{s}}
\def\br{\mathbf{r}}
\def\bS{\mathbf{S}}

\def\bx{\mathbf{x}}
\def\bX{\mathbb{X}}
\def\by{\mathbf{y}}
\def\bz{\mathbf{z}}
\def\bv{\mathbf{v}}
\def\bV{\mathbf{V}}
\def\mcX{\mc {X}}

\def\bc{\mathbf{c}}
\def\bM{\mathbf{M}}
\def\bS{\mathbf{S}}
\def\bT{\mathbf{T}}

\def\probname{\textbf{MMS}}
\def\sampprobname{\textbf{MMS-Sample}}

\def\relaxation{\tau_{\mathrm{rel}}}
\def\relaxationk{\tau_{\mathrm{rel},k}}
\def\relaxationgamma{\tau_{\mathrm{rel},\gamma}}
\def\mixingtime{\tau_{\mathrm{mix}}}

\def\mixingtimestate{\tau_{\mathrm{mix}}(\varepsilon; P, \bx_0)}

\def\mixinggammalb{\tau^*_\gamma(\varepsilon)}

\def\numsampk{N_{k}}
\def\numsampgamma{N_{\gamma}}
\def\numsamplbk{\underline{N}_k}
\def\numsampubk{\overline{N}_k}

\def\numsamplbgamma{\underline{N}_\gamma}
\def\numsamplbgammastar{\underline{N}_{\gamma^*}}

\def\numsampubgammaeps{\overline{N}_\gamma(\varepsilon)}
\def\solutiondensity{p_{\gamma}}
\def\solutiondensitystar{p_{\gamma}^*}

\def\Mgamma{M_{\gamma}}
\def\Pgamma{P_{\gamma}}
\def\Qgamma{Q_{\gamma}}
\def\Mreduced{M_{k}}
\def\Preduced{P_{k}}
\def\Preducedtwo{P_{k=2}}
\def\given{\; | \;}
\def\tpg{\tilde \pi_{\gamma}}

\title{Synthetic Census Data Generation via Multidimensional Multiset Sum}
\date{}

\author{Cynthia Dwork
  \thanks{Department of Computer Science, Harvard
  University}
  \and
  Kristjan Greenewald
  \thanks{MIT-IBM Watson AI Lab; IBM Research, Cambridge, MA, USA}
  \and
  Manish Raghavan
  \thanks{MIT Sloan School of Management and Department of
  EECS}
}

\begin{document}

\maketitle

\begin{abstract}

  The US Decennial Census provides valuable data for both research and policy
  purposes. Census data are subject to a variety of disclosure avoidance
  techniques prior to release in order to preserve respondent confidentiality.
  While many are interested in studying the impacts of disclosure avoidance
  methods on downstream analyses, particularly with the introduction of
  differential privacy in the 2020 Decennial Census, these efforts are limited by
  a critical lack of data: The underlying ``microdata,'' which serve as necessary
  input to disclosure avoidance methods, are kept confidential.

  In this work, we aim to address this limitation by providing tools to generate
  synthetic microdata solely from published Census statistics, which can then be
  used as input to any number of disclosure avoidance algorithms for the sake of
  evaluation and carrying out comparisons. We define a principled distribution
  over microdata given published Census statistics and design algorithms to sample
  from this distribution. We formulate synthetic data generation in this context
  as a knapsack-style combinatorial optimization problem and develop novel
  algorithms for this setting. While the problem we study is provably hard, we
  show empirically that our methods work well in practice, and we offer
  theoretical arguments to explain our performance. Finally, we verify that the
  data we produce are ``close'' to the desired ground truth.
\end{abstract}
\maketitle

\section{Introduction}
\label{sec:introduction}

Scholars, practitioners, and policy-makers rely on US Decennial Census data for
a wide range of research and decision-making tasks. For privacy reasons, Census
data are not released in full. All released data are instead subject to a
variety of disclosure avoidance practices. For example, data may be perturbed
before release or have their location information coarsened. The unperturbed
``microdata'' are kept secret for 72 years after their collection.

The Census Bureau updated its disclosure avoidance system to use differential
privacy for the release of the 2020 Decennial Census instead of its prior
swapping-based methodology~\citep{abowd2018us,abowd20222020}. This sparked
renewed interest in the properties of various privacy-preserving methods and
their impacts on downstream consumers of Census data products. For example,
decisions involving budgeting, voting rights and redistricting, and planning
rely on accurate and consistent Census
data~\citep{cohen2022census,steed2022policy}. Social scientists also rely on
Census data for a wide range of research including health and social
mobility~\citep{ruggles2019differential}. These stakeholders have begun to ask a
critical question: How reliable can these analyses be under privacy-preserving
techniques?

A key obstacle to answering this question is the secrecy of the data themself:
while the TopDown algorithm used by the Census in 2020 is public, the underlying
data on which it is run are not. Ideally, one could simply simulate multiple runs
of the TopDown algorithm (or any other proposed alternative) to characterize its
effects on downstream quantities of interest. But the algorithm(s) in question
take as input the underlying \textit{microdata}, which are only released 72
years after they are collected.

Prior research on disclosure avoidance and the Census has used a variety of
workarounds, including a limited sample of public Census demonstration
data~\citep{kenny2023evaluating,dick2023confidence}, older Census
data~\citep{bailie2023can,petti2019differential}, or heuristic methods to
generate microdata~\citep{christ2022differential,cohen2022census}. We discuss
these heuristics in greater detail in \Cref{sec:related}.

In this work, we seek to enable comprehensive research into the US Census,
including research on privacy, by providing a principled method to generate
synthetic Census microdata from publicly available data sources. Our aim is to
enable research into the impacts of privacy-preserving technology on Census data
beyond the above-mentioned heuristic workarounds. Moreover, our tools could
provide a starting point for Census data consumers to estimate and potentially
correct for the biases induced by disclosure avoidance algorithms by estimating
how they affect quantities of interest, which we discuss further in
\Cref{sec:conclusion}.

At a high level, we combine block-level aggregate statistics with a random
sample of microdata containing only coarse location
information~\citep{beckman1996creating}. Importantly, our goals and methods
differ from those of many Census-specific ``reconstruction attacks,'' which seek
to analyze privacy-preserving methods by testing whether and how many rows of
the microdata can be reconstructed from publicly released
information~\citep{abowd2018staring,abowd-declaration,dick2023confidence,francis2022note}
(see also: \citet{dinur2003revealing}). Reconstruction attacks typically seek to
find the \textit{most likely} microdata given the available information, which
will in general lead to a more homogeneous dataset at a population level. In
contrast, our aim is to sample from a representative distribution over
microdata. We do not use any auxiliary information (i.e., non-Census data
products), and we seek to generate representative, state-wide synthetic
microdata, instead of a fraction of the rows. We intend for researchers to
perform downstream analyses over multiple samples of this microdata. If a
finding (e.g.,~that a particular disclosure avoidance method biases a statistic
of interest) holds across multiple samples from this distribution, we may be
more concerned that it holds for the ground truth data as well. Note that we do
not intend for our synthetic data to be interpreted as ``ground truth''; our
goal is to provide synthetic data that are both faithful to published
information and statistically plausible.

At a technical level, we formulate synthetic data generation in this context as
a knapsack-style combinatorial optimization problem. Given aggregate statistics
for each Census block (e.g., number of households, number of individuals of each
race, \dots), we seek to sample households that, when put together, exactly
match the aggregate statistics reported by the Census. We design a Markov Chain
Monte Carlo algorithm to sample appropriate households. While the problem we
seek to solve is NP-hard, we provide both theoretical and empirical evidence to
show that our methods perform sufficiently well to be viable in this setting,
allowing us to sample datasets across entire US states. We provide code and
detailed instructions for others to generate their own synthetic data at
\url{https://github.com/mraghavan/synthetic-census}. Our implementation is
specific to the 2010 US Decennial Census, but our broader framework can be
adapted to general population synthesis tasks.\footnote{Changes to the 2020
  Census prevent our methods from being directly applicable. In particular, the
  2020 Census includes the ``Privacy-Protected Microdata File'' (PPMF), which is
  a 100\% enumeration of synthetic persons and
  households~\citep{uscensus2024ppmf}. However, because the persons and
  household files are separate, future work could adapt our techniques to create
synthetic microdata by combining these datasets.}

\paragraph*{Organization of the paper.} In \Cref{sec:problem-formulation}, we
formalize synthetic data generation in our setting as a combinatorial
optimization problem. We discuss related work in \Cref{sec:related}. In
\Cref{sec:mcmc}, we present and analyze a pair of Markov Chain Monte Carlo
(MCMC) algorithms, finding that our approach empirically performs much better
than generalizations of prior work. We describe our overall hybrid integer
linear programming-MCMC algorithm in \Cref{sec:overall}. In
\Cref{sec:evaluation}, we evaluate the representativeness of data sampled via
this algorithm. We discuss the implications and usage of our methods in
\Cref{sec:conclusion}.

\section{Problem Formulation}
\label{sec:problem-formulation}

\paragraph*{Notation.}

We will denote nonnegative integer vectors with bold lower-case letters (e.g.,
$\bx$) and use calligraphic upper-case letters (e.g., $\mcX$ or $\mc D$) to
denote sets and probability distributions. We will use bold upper-case letters
for nonnegative integer matrices (e.g., $\bV$). We will write $\bx_i$ to denote
the $i$th vector in a set and $\bx[i]$ to denote the $i$th entry of vector
$\bx$, indexing vectors beginning with 1. We use $\preceq$ and $\succeq$ to
denote a vector being element-wise $\le$ (resp. $\ge$) another vector. We will
use $\bx_{i \gets g}$ to denote $\bx$ with its $i$th entry replaced by the value
$g$. $[n]$ refers to the set $\{1, \dots, n\}$. For a random variable $\bX$ with
distribution $\sigma$ over a discrete set $\mcX$, we will write $\sigma(\bx) =
\Pr[\bX = \bx]$ for $\bx \in \mcX$ and $\sigma(\mc S) = \Pr[\bX \in \mc S]$ for
$\mc S \subseteq \mcX$. To refer to the conditional distribution of $\sigma$ on
$\mc S \subseteq \mcX$, we write $\sigma \given \bx \in \mc S$.

\subsection{Empirical setting}

Our goal is to generate synthetic microdata based on the 2010 US Decennial
Census. Cleaned Census responses are collected in a dataset known as the Census
Edited File (CEF) often referred to as ``microdata.'' To meet its statutory
privacy obligations, the Census Bureau does not release this dataset. Instead,
they apply a suite of disclosure avoidance techniques (including adding noise,
censoring outliers, etc.) before releasing aggregate statistics. Often,
statistics are released at the Census block level, where a Census block
typically consists of at most a few hundred households. In particular, ``Summary
File 1'' (SF1) provides granular demographic information for each Census block
\textit{after} disclosure avoidance techniques are used.

For each block, we will seek to sample a collection of households whose
characteristics match statistics reported in SF1.\footnote{We obtain SF1
data via IPUMS \citep{ipums}.} As a result, our data will exactly match SF1
along all the attributes we choose (detailed below). Importantly, our methods
preserve structural zeros: if SF1 reports zero people with certain
characteristics in a block, we ensure that our data have the same property.

The statistics included in SF1 are fairly detailed and include counts of
individuals with various attributes (e.g., number of Hispanic persons) and
detailed household types (e.g., number of households with three members headed
by a householder of two or more races).\footnote{We choose of a subset of the
  SF1 statistics to match, which we describe in more detail in
\Cref{app:encoding}.} For a given block, we will denote these statistics by a
nonnegative integer vector $\bc \in \Zplus^d$. An important feature we will rely
on is the fact that $\bc$ encodes the total number of households in a block,
which we will refer to as $m$. In our case, each vector has dimension $d = 135$,
where each dimension encodes a particular count. For more details, see
\Cref{app:encoding}.

In addition to SF1, the Census Bureau also releases the Public Use Microdata
Sample (PUMS), consisting of a 10\% sample of households across each
state.\footnote{\url{https://www.census.gov/data/datasets/2010/dec/stateside-pums.html}}
Crucially, for privacy reasons, the PUMS does not contain granular geographic
locations for each household; each household is annotated with the Public Use
Microdata Area (PUMA) in which it resides. For our purposes, we ignore PUMA
information and treat the PUMS as simply a statewide sample.\footnote{Each PUMA
has at least 100,000 individuals in it. We aggregate to the state level because
otherwise, the data are too sparse for our methods to be effective. As a result,
we will fail to capture regional variation in household composition that are not
explained by SF1. Intuitively, our formulation makes the assumption that,
conditioned on the SF1 counts, the distribution of households is
location-invariant within the state.} We will treat each distinct household in
the PUMS dataset as a vector $\bv_i \in \Zplus^d$ (encoded in the same
$d=135$-dimensional space as each block). As before, each entry represents the
count of a certain property, which can either be the number of individuals in a
household satisfying a particular demographic property or the (binary) indicator
for whether the household as a whole satisfies a property. The frequency of each
distinct household type in this encoding yields a distribution $\mc D$. Let $n$
be the number of distinct household types, each represented by a vector $\bv_i$,
in a given state (in our case, $n$ is on the order of a few thousand), and
define $\bV \in \Zplus^{d \times n}$ to be the matrix with columns $\bv_i$. It
will sometimes be convenient to refer to the set of household type vectors,
which we will denote $\mc V \triangleq \{\bv_1, \dots, \bv_n\}$.

With this data, our goal is as follows. A ``solution'' to a block $b$ is a
multiset $\bx \in \Zplus^n$ (represented as a vector of cardinalities of each
element) such that $\bV \bx = \bc_b$. Note that this linear equality exactly
captures the constraint that, when summed together, the characteristics of the
multiset of households exactly match those reported in SF1.

\subsection{Handling multiplicity}

If each block had a unique $\bx$ satisfying $\bV \bx = \bc_b$, then this would
suffice---we could find $\bx$ for each block $b$, producing the entire
microdata. This is not the case.\footnote{For privacy reasons, this is to be
expected.} As a simplified example, consider a block with four individuals, two
white and two Asian, split into two households of size two. Without further
information, there are two possibilities: the block could contain either two
racially homogeneous households or two multiracial households.

When faced with multiple possible solutions underlying a given block, what
should we do? We could of course choose arbitrarily. The Census' own
demonstration reconstruction attack appears to take this
approach~\citep{abowd2018staring}. But in the above example, we might believe
that statistically speaking, it is more likely that a block contains two
racially homogeneous households than two multiracial households. This intuition
is borne out in the PUMS microdata sample: multiracial households are far less
frequent than racially homogeneous households. How should this information
inform our sample when multiple possible reconstructions exist?

One approach would be to try to find the ``most likely'' reconstruction (for
some definition of likelihood we will have to make formal). If our goal was a
reconstruction attack, this might be the right choice. But if we do this for all
blocks, our overall sample will be quite biased. For example, if racially
homogeneous households are much more frequent than multiracial households, our
resulting dataset will contain very few multiracial households relative to what
we know about the overall population.

Instead, we take a different approach designed to produce a more representative
sample. Let $\mcX_b \triangleq \{\bx : \bV \bx = \bc_b\}$ be the set of all
valid solutions for a given block. In other words, $\mcX_b$ is the set of
combinations of households that, when summed together, match the aggregate
block-level counts. We will specify a distribution $\pi$ over $\mcX_b$ based on
the PUMS distribution $\mc D$ and seek to sample from $\mcX_b$ according to
$\pi$. Ideally, our choice of $\pi$ would lead to a representative statewide
sample, meaning that when sampling across the entire state, the expected
frequency of each household in our sampled microdata matches its empirical PUMS
frequency in $\mc D$. Unfortunately, this would require $\pi$ to depend on
$\bc_b$ for all $b$ in a state, which would be prohibitively expensive. Instead,
we specify a natural generative model to induce a distribution $\pi$. We
evaluate the representativeness of samples produced by $\pi$ in
\Cref{sec:evaluation}.

\subsection{A generative model}

Consider a generative model in which we are given the PUMS distribution $\mc D$
over the set of household types $\mc V \triangleq \{\bv_1, \dots, \bv_n\}$. Assume
that for a given Census block $b$, aggregate statistics $\bc_b$ are chosen
exogenously. (We will often drop the subscript $b$ for ease of notation.) Then,
a multiset $\bx$ of households is sampled i.i.d. according to $\mc D$. We are
interested in the distribution over multisets of households that this produces
conditioned on the event that the aggregate characteristics of sampled
households exactly matches the reported statistics: that is, conditioned on $\bV
\bx = \bc$.

Intuitively, we can think of this as the distribution induced via rejection
sampling; indeed, a basic (and prohibitively inefficient) algorithm to sample
from this distribution is to repeatedly sample households i.i.d. from $\mc D$
until either $\bV \bx = \bc$ or $(\bV \bx)[i] > \bc[i]$ for some $i$, and accept
the first $\bx$ that satisfies $\bV \bx = \bc$. (See \Cref{alg:rejection} in
\Cref{app:rejection}.) In this generative model, for $\bx$ such that $\bV \bx =
\bc$, the posterior $\pi$ induced by rejection sampling is given (up to a
normalizing constant) by
\begin{align*}
  \numberthis \label{eq:piD-def}
  \pi(\bx) \propto f(\bx) \triangleq \binom{\| \bx \|_1}{\bx} \prod_{i=1}^n
  \Pr_{\mc D}[\bv_i]^{\bx[i]},
\end{align*}
where the multinomial coefficient, defined $\binom{\| \bx \|_1}{\bx} =
\|\bx\|_1!/(\bx[1]! \dots \bx[n]!)$ accounts for the multiplicity of elements in
$\bx$. (Recall that $m = \|\bx\|_1$, the number of households in a block, is
known for each block.) For $\bx$ such that $\bV \bx \ne \bc$, we define
$\pi(\bx)$ to be 0. With this, we have specified a combinatorial optimization
problem: sample from $\mcX = \{\bx : \bV \bx = \bc\}$ according to the
distribution $\pi$ in \Cref{eq:piD-def}. We formalize this below.

\subsection{Formal problem definition}

Sampling from $\pi(\cdot)$ over $\mcX$ is closely related to the classical
subset sum problem. In subset sum, the goal is to choose a subset of integers
that sum to a desired target. This can be generalized to the multidimensional
subset sum problem, in which we choose a subset of integer vectors that sum to a
target vector. The problem we seek to solve (sampling from $\mcX$) can be
thought of as a (nonnegative) ``multidimensional multiset sum'' problem, similar
to multidimensional subset sum but with a slight modification: we may choose
vectors with replacement. We define the decision version of the multidimensional
subset sum (\probname) problem and its sampling analogue (\sampprobname) below:
\begin{problem}[\probname]
  Given a matrix $\bV \in \Zplus^{n \times d}$ and vector $\bc \in \Zplus^d$, 
  determine whether $\mcX = \{\bx \in \Zplus^n : \bV \bx = \bc\}$ is non-empty.
\end{problem}
\begin{problem}[\sampprobname]
  Given a matrix $\bV \in \Zplus^{n \times d}$ and a vector $\bc \in \Zplus^d$,
  let $\mcX = \{\bx \in \Zplus^n : \bV \bx = \bc\}$. Given a distribution
  $\sigma$ over $\mcX$ specified up to a normalizing constant, sample from
  $\sigma$.
\end{problem}
\probname\ is NP-hard (see \Cref{clm:hardness} in \Cref{app:hardness} for
details), making it hard to approximate $|\mcX|$ to within a constant
factor. As a result, \sampprobname\ is NP-hard~\citep{jerrum1986random}.
Moreover, relaxing the constraint to $\bV \bx \approx \bc$ is unlikely to make
the problem significantly easier, since approximate multidimensional subset sum
is NP-hard for $d >
2$~\citep{emiris2017approximating,magazine1984note,kulik2010there}. Our key
technical contribution is an algorithm for \sampprobname\ that works
sufficiently well in practice for the purposes of generating synthetic Census
microdata. We focus on sampling from $\pi$ given by~\Cref{eq:piD-def}, but our
algorithms work for any distribution $\sigma$ specified up to a normalizing
constant. Before describing our solution, we discuss related work on population
synthesis, Census disclosure avoidance, and similar combinatorial optimization
problems.

\section{Related work}
\label{sec:related}

\paragraph*{Population synthesis.}

Researchers have sought to combine microdata with tabular data to generate
synthetic populations for a variety of purposes including agent-based simulation
and transportation analyses~\citep{beckman1996creating}. At a high level, this
body of literature focuses on two steps: building models of population
distributions and sampling from those distributions to meet given constraints
\citep{muller2010population}. Much of the focus of this literature has been on
building better models of populations \citep[e.g.~][]{creecy2009feasibility,
farooq2013simulation, casati2015synthetic, sun2015bayesian,
sun2018hierarchical}, often by modeling complex relationships between
demographic attributes or hierarchical household-individual relationships. This
is particularly valuable when dealing with rich, high-dimensional data. More
recent work has also used machine-learning approaches to model
populations~\citep{albiston2024neural,gussenbauer2024simulation}. In contrast,
we take a particularly simple approach to population modeling, treating the PUMS
as a population model.

Given a population model, the literature contains a few high-level strategies to
produce a sample satisfying aggregate constraints like the ones imposed by SF1.
Early techniques like iterative proportional fitting (IPF)
\citep{deming1940least} seek to fit weights for each member of a population
distribution such that sampling from the population distribution according to
those weights would yield the desired aggregate counts in
expectation~\citep{birkin1988synthesis, beckman1996creating}. In our notation,
these methods find real-valued $\bx \in \R_{\ge 0}^n$ such that $\bV \bx = \bc$,
and they seek to align $\bx$ with the population model $\mc D$. They do not
guarantee that the constraints are met \textit{ex post}, i.e.,~after an actual
dataset is sampled. For large geographic regions, sampling this way will
approximately satisfy the given constraints. But for small geographic regions
(e.g.,~Census blocks), sampling schemes yield high variance, meaning generated
datasets will differ significantly from the given constraints. A class of
heuristics known as ``deterministic reweighting'' schemes seek to reduce the
variance produced by sampling~\citep{ballas2005simbritain,
creecy2009feasibility, lovelace2013truncate, casati2015synthetic}, but these
methods tend to either fail to provide exact guarantees or introduce bias into
the sample \citep{muller2010population}.

Most related to ours, a line of work draws on combinatorial optimization
heuristics like hill-climbing, simulated annealing, and genetic algorithms to
construct a sample that comes as close as possible to matching the aggregate
counts~\citep{williamson1998estimation, voas2000evaluation, harland2012creating,
ma2015synthetic, wu2022synthetic, whitworth2022synthacs,
gussenbauer2024simulation}. The quality of these approaches is often measured by
the (sometimes squared) difference between the aggregate counts and the sample
counts $\|\bV \bx - \bc\|_1$. In general, these heuristics fail to guarantee
$\bV \bx = \bc$.

Our primary technical contribution lies in sampling, not in modeling the
population distribution (since we adopt the PUMS as our population model). In
contrast to prior work, our work guarantees exact matches to the aggregate
counts. If this was our only goal, standard integer linear programming (ILP)
would suffice, since our problem setting is simple enough that we do not need to
rely on heuristic methods from combinatorial optimization. However, instead of
simply finding a single solution, our methods are designed to sample from a
known distribution over all possible solutions. We could in principle apply our
techniques to use more complex population models than the simple PUMS
distribution used here; we defer such investigations to future work.

\paragraph*{Analyzing Census disclosure avoidance systems.}

Prior work that has generated synthetic microdata for the purposes of analyzing
Census disclosure avoidance systems
\citep{cohen2022census,christ2022differential} has relied on heuristics that do
not produce reliable household-level data. \citet{cohen2022census} explicitly
note that their synthetic data do not contain household information, preventing
them from fully replicating the Census Bureau's disclosure avoidance system.
(They do perform experiments in which they arbitrarily group individuals into
households of size 5 in their synthetic data.) \citet{christ2022differential}
use a combination of heuristics to produce a limited sample of data by randomly
selecting blocks and pooling data. This enables them to generate
individual-level data that bear some resemblance to the ground truth. In
contrast, our methods generate state-wide microdata at the household level.
Household-level data are strictly more general than individual-level microdata,
since we can produce an individual-level dataset simply by enumerating
individuals in each household.

Our work bears some resemblance to recent ``reconstruction attacks'' on Census
data, which attempt to reconstruct rows of the dataset given Census statistics
and potentially external
information~\citep{abowd-declaration,dick2023confidence,francis2022note}. In
particular, the Census Bureau conducted a reconstruction attack which used
integer programming to combine information across multiple tables to produce
microdata~\citep{abowd2018staring}. Our work extends the scope of these efforts
in a number of ways. First, in order to facilitate the analysis of the impacts of disclosure
avoidance techniques like swapping,
we produce household-level data instead of
individual-level data.
Second, we integrate information about the distribution of
households from the PUMS, allowing us to produce a more representative dataset.
As discussed earlier, our goal of representativeness is quite different from
reconstruction: maximizing reconstruction ``accuracy'' would suggest choosing
the most likely sample, which would be fairly homogeneous and unrepresentative;
in contrast, we seek to produce a dataset that matches state-wide household
statistics, which introduces additional complexity. In principle, if we were
interested in adapting our methods to produce an effective reconstruction attack
at the household level, we could sample $\bx^* = \argmax_{\bx \in \mc X}
\pi(\bx)$. In \Cref{app:reduced-omitted}, we provide a linear approximation to
$\pi(\cdot)$ that could be used to solve this via integer linear programming.

\paragraph*{Multidimensional knapsack, subset sum, and related combinatorial
optimization problems.}

Prior work considers closely related problems of knapsack sampling/counting and
systems of linear Diophantine equations (equations of the form $\bV \bx = \bc$
with integrality constraints). While there exist polynomial-time approximation
schemes for knapsack
counting~\citep{morris2004random,dyer2003approximate,gopalan2011fptas,rizzi2014faster,gawrychowski2018faster,kayibi2018mixing,lawler1977fast},
these techniques do not generalize to our setting because (1) we require
\textit{exact}, not just \textit{feasible} solutions (i.e., $\bV \bx = \bc$
instead of $\bV \bx \le \bc$); and (2) these algorithms scale exponentially with
dimension $d$. Because \sampprobname\ is NP-hard to approximate, we should not
expect a polynomial-time approximation scheme. Other related work includes
generalizations of the knapsack problem~\citep{hendrix2015bounded} and dynamic
programming approaches~\citep{bossek2021exact}, which are empirically too slow
in our high-dimensional setting.

Several specialized algorithms for knapsack and subset sum-style problems have
appeared in the literature, some of which extend to multidimensional settings
\citep[e.g.,][]{cabot1970enumeration,ingargiola1977general,puchinger2010multidimensional,salkin1975knapsack,martello1987algorithms,pisinger1999exact}.
We do not experiment with them here since our own implementations of these
methods are unlikely to compete with general-purpose but highly optimized
integer linear programming packages, which we make heavy use of. Future work
might be able to take advantage of these to further optimize the methods we
develop here.

Our work draws most closely on the MCMC algorithm of \citet{dyer1993mildly}. A
direct adaptation of their algorithm (\Cref{sec:simple}) is still too slow in
our setting, but we develop a new MCMC approach that works better
empirically (\Cref{sec:reduced}). We characterize the performance of our
algorithms using a long line of theoretical results on the mixing time of Markov
chains~\citep{diaconis1991geometric,lawler1988bounds,jerrum1988conductance,sinclair1989approximate,jerrum1989approximating}.

\sampprobname\ can also be described as sampling nonnegative solutions to a
system of linear Diophantine equations. While these systems have been studied
extensively~\citep{blankinship1966algorithm,bradley1971algorithms,chou1982algorithms,lazebnik1996systems,bradler2016number,aardal2000solving,sanchez2016linear},
algorithms in this setting generally produce integer solutions, not
\textit{nonnegative} integer solutions. (Recall that we require nonnegativity
because a multiset of households cannot contain a negative number of copies of a
household.) The literature contains some results on determining the existence of
or bounding the number of nonnegative
solutions~\citep{bradler2016number,mahmoudvand2010number}, but these are not
algorithmic in nature.

In addition to all of these techniques, we have one more tool at our disposal:
Integer Linear Programming (ILP). For sufficiently small problems, we can use
highly optimized software packages\footnote{We use Gurobi
(\url{https://www.gurobi.com/}) in our experiments. Free alternative solvers can
also be used in its place.} for ILP to simply enumerate
$\mcX$ and sample according to $\pi$ as desired. In all of our instances (where
each instance is a Census block), we find that ILP suffices to determine whether
$\mcX$ is non-empty. Enumeration, however, has a clear downside: its complexity
scales with $|\mcX|$, which may be exponentially large in $n$, the number of
possible households. For example, using ILP to enumerate up up to 5000 elements
of $\mcX_b$ from each Census block in $b$ Alabama and Nevada, we plot the
distribution of $\min(|\mcX_b|, 5000)$ in \Cref{fig:heavy-tail}. Our results
suggest that $|\mcX_b|$ has a heavy-tailed distribution,\footnote{We do not
attempt to evaluate whether the power-law distribution fits these data well.
\Cref{fig:heavy-tail-AL,fig:heavy-tail-NV} are simply meant to be illustrative.}
meaning that enumerating $\mcX_b$ completely for each block is likely to be
computationally infeasible.\footnote{For reference, enumerating up to 5000
  solutions from $\mcX_b$ for each block $b$ takes thousands of CPU-hours for
  Alabama and Nevada, which have 135,838 and 35,916 non-empty Census blocks
respectively.} For instances where $|\mcX_b|$ is too large, we need more
efficient sampling algorithms. We turn to Markov Chain Monte Carlo methods for
this.

\begin{figure}[ht]
  \centering
  \begin{subfigure}[b]{0.48\textwidth}
    \includegraphics[width=\textwidth]{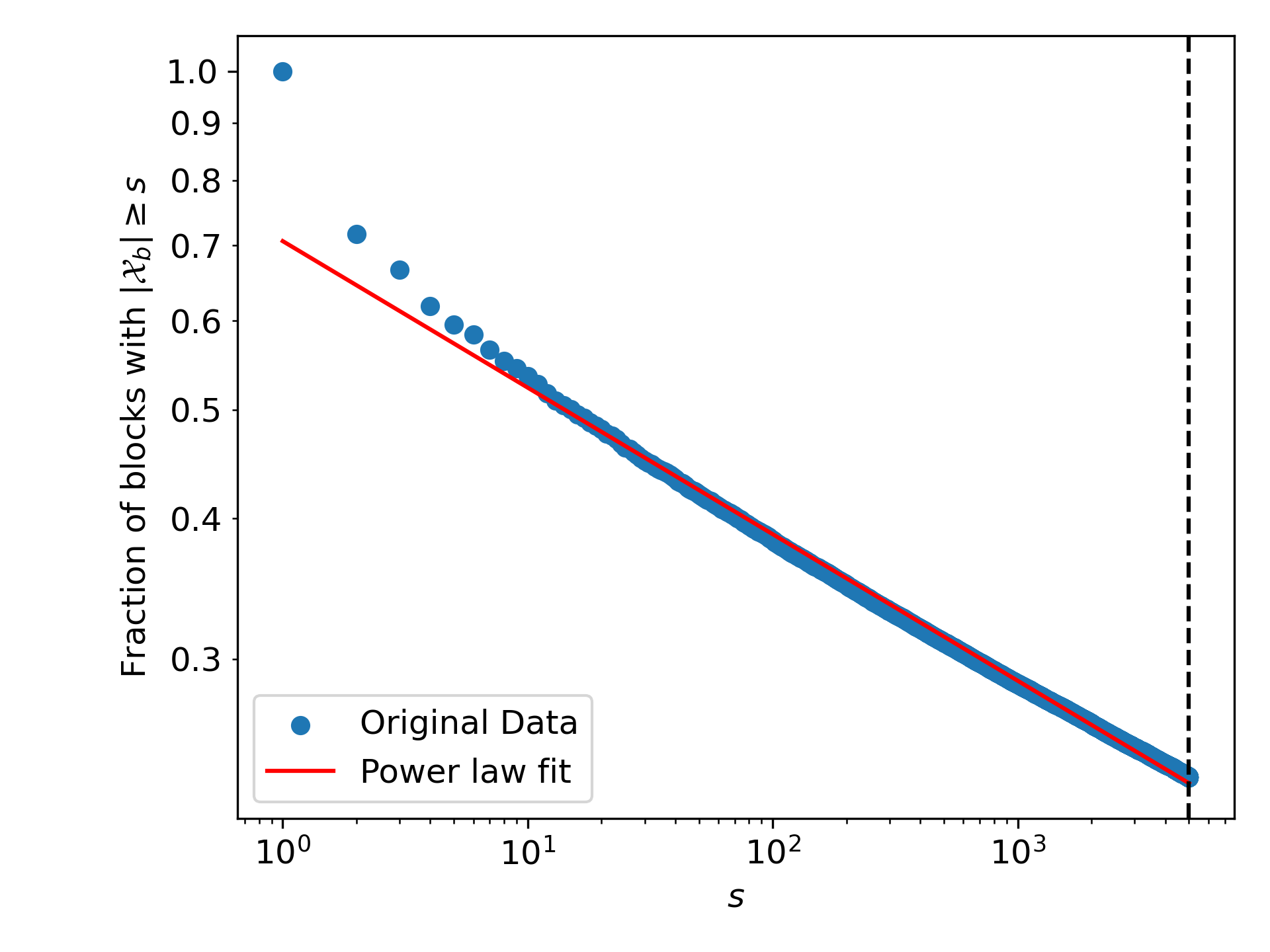}
    \caption{AL}
    \label{fig:heavy-tail-AL}
  \end{subfigure}
  \hfill
  \begin{subfigure}[b]{0.48\textwidth}
    \includegraphics[width=\textwidth]{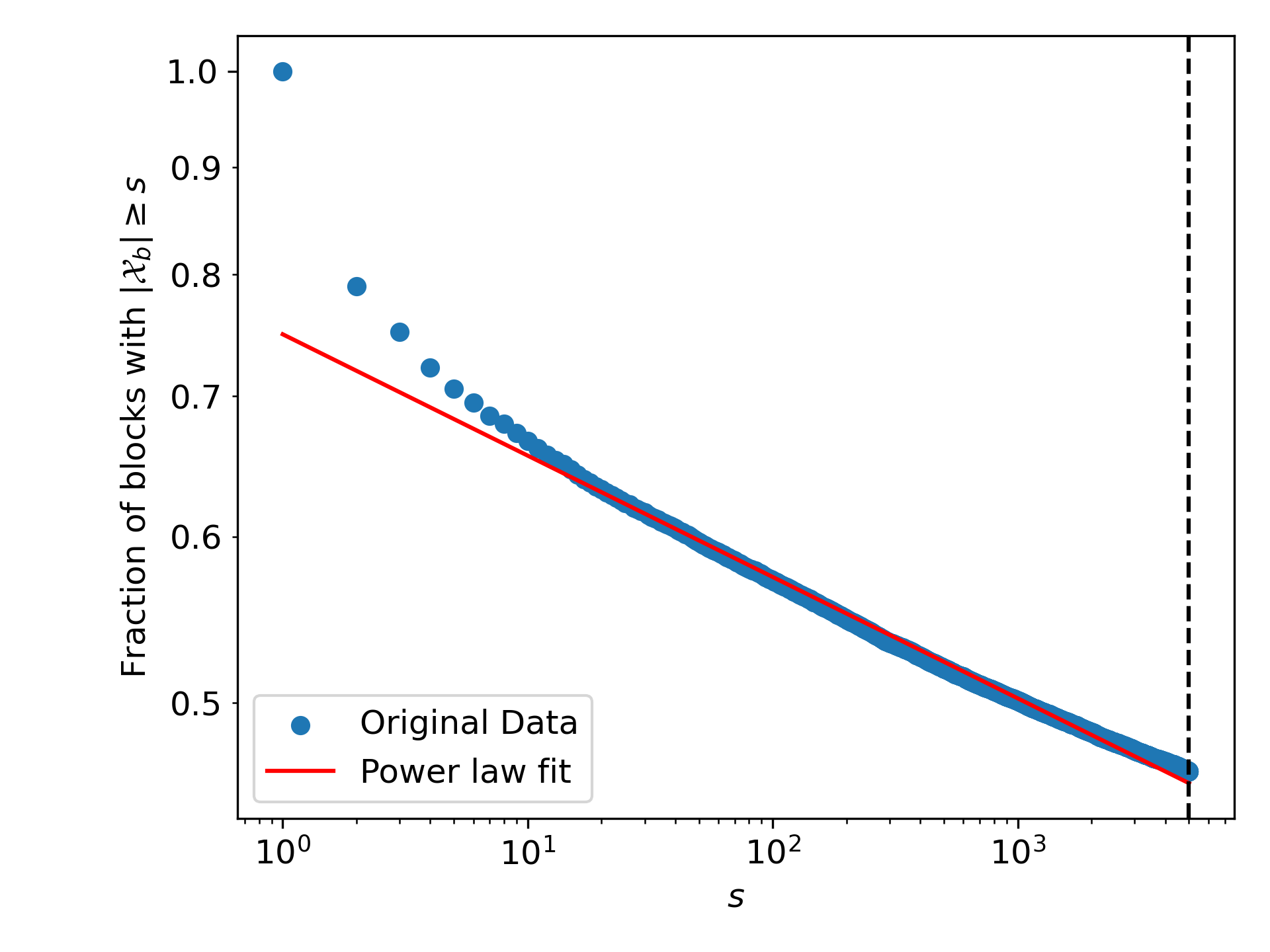}
    \caption{NV}
    \label{fig:heavy-tail-NV}
  \end{subfigure}
  \caption{$|\mcX_b|$, the number of solutions in each block $b$, appears to be
  heavy-tailed.}
  \label{fig:heavy-tail}
\end{figure}

\section{Markov Chain Monte Carlo Methods}
\label{sec:mcmc}

Markov Chain Monte Carlo methods can allow us to efficiently sample from
exponentially large state spaces and have been used in prior work on knapsack
sampling~\citep{dyer1993mildly,morris2004random}. At a high level, MCMC works by
defining a Markov chain over the solution space, performing a random walk on
this Markov chain, and yielding a solution after a fixed number of random walk
steps. Here, we develop MCMC techniques to solve \sampprobname. We present two
approaches: the ``simple'' chain, a modification of the algorithm of
\citet{dyer1993mildly}, and the ``reduced'' chain, which can be interpreted as a
truncation of the simple chain. These come with trade-offs, which we evaluate
both theoretically and empirically. At a high level, the simple chain is
guaranteed to converge to the desired stationary distribution but may mix
slowly. In contrast, the reduced chain may not converge to the desired
distribution, but empirically performs better in our setting. Neither comes with
strong theoretical guarantees; this is to be expected, since \sampprobname\ is
provably hard in the general case.

\paragraph*{Preliminaries.}

As is standard for MCMC, we will evaluate our algorithms in terms of the
computation required to produce an approximately random sample from $\pi$ on
$\mcX$. For distributions $\sigma, \sigma'$ over $\mcX$, define $\dTV$ to be the
total variation distance:
\begin{align*}
  \numberthis \label{eq:tvd}
  \dTV(\sigma, \sigma') \triangleq \frac{1}{2} \|\sigma - \sigma'\|_1 =
  \max_{\mc S \subset \mcX} |\sigma(\mc S) - \sigma'(\mc S)|.
\end{align*}
When using MCMC to sample from a state space, we typically begin with a starting
state chosen from some initial distribution $\sigma_0$, randomly transition for
$t$ steps, and return the final state. Let $\bX$ denote a random sample
generated this way, and
let $\sigma_{\bX}$ be
its distribution. Our goal is to produce an $\varepsilon$-approximate sample
from a target distribution $\sigma$:\footnote{Throughout this paper, we will
choose the target distribution to be $\pi$ as defined in~\Cref{eq:piD-def}.}
\begin{Defi}[$\varepsilon$-approximate sample]
  A random variable $\bX$ with distribution $\sigma_{\bX}$ is an
  $\varepsilon$-approximate sample from a distribution $\sigma$ if $\dTV(\sigma,
  \sigma_{\bX}) \le \varepsilon$.
\end{Defi}
This is closely related to the \textit{mixing time} of a Markov
chain with transition matrix $P$ with stationary distribution $\sigma$:
\begin{align*}
  \mixingtime(\varepsilon; P) \triangleq \min \left\{ t \in \Zplus :
  \max_{\sigma_0} \dTV(\sigma_0 P^t, \sigma) \le \varepsilon \right\}.
\end{align*}
However, because the algorithms we present each have a known initial state
$\sigma_0$, we are not interested in the worst-case over all $\sigma_0$. Instead, we
use a variant of the mixing time given a known starting distribution $\sigma_0$
with probability 1 on initial state $\bx_0$ and 0 elsewhere. Define
\begin{align*}
  \numberthis \label{eq:mixing-def}
  \mixingtimestate \triangleq \min \left\{ t \in \Zplus : \dTV(\sigma_0 P^t,
    \sigma)
  \le \varepsilon \right\}
\end{align*}
to be the number of iterations required to generate an $\varepsilon$-approximate
sample from the stationary distribution $\sigma$.
Standard spectral techniques yield bounds for $\mixingtimestate$. Let $P$ be the
transition matrix of an irreducible Markov chain. Assume that the chain is
``lazy,'' i.e., $P(\bx, \bx) \ge 1/2$ for all $\bx$.\footnote{This guarantees
that all eigenvalues are nonnegative.} Let $\lambda_2(P)$ be the second-largest
eigenvalue of $P$. Then, the \textit{relaxation time} of $P$ is defined as
$\relaxation \triangleq \frac{1}{1-\lambda_2(P)}$. Classical results tell us
that $\relaxation$ provides a tight characterization of $\mixingtimestate$:
\begin{theorem}[E.g., {\citet[Prop. 2.3]{guruswami2016rapidly}}; see also
  {\citet[Thm. 12.5]{levin2017markov}}]
  For an irreducible, aperiodic Markov chain,
  \begin{align*}
    \numberthis \label{eq:mt-bounds}
    \p{\relaxation - 1} \log \p{\frac{1}{2\varepsilon}} \le
    \mixingtimestate \le \relaxation \log
    \p{\frac{1}{\varepsilon \sigma(\bx_0)}}.
  \end{align*}
\end{theorem}
With these definitions, we are ready to describe and analyze our two MCMC-based
algorithms.

\subsection{The simple chain}
\label{sec:simple}

\subsubsection{Defining the simple chain}

We begin by adapting the algorithm of \citet{dyer1993mildly} which samples
(approximately) uniformly from the set of feasible knapsack solutions $\mc Y
\triangleq \{\bx : \bV \bx \preceq \bc\}$. Our adaptation has two key
differences: we are interested in exact solutions (i.e., $\mcX = \{\bx : \bV \bx
= \bc\}$), and we want to sample from a particular distribution $\pi$ as defined
in \Cref{eq:piD-def}. At a high level, we will design a Markov chain $\Mgamma$
parameterized by $\gamma \in \R$ with stationary distribution $\tpg$:
\begin{align*}
  \numberthis \label{eq:tpg-def}
  \tpg(\bx) \propto f(\bx) \exp(-\gamma \|\bV \bx - \bc \|_1)
\end{align*}
over $\mc Y$. Observe that the conditional distribution of $\tpg$ over $\mcX$ is
$\pi$, i.e., $(\tpg \given \bx \in \mcX) = \pi$. As a result, if we can sample
efficiently from $\tpg$, we can use rejection sampling to generate samples from
$\pi$ over $\mcX$. In other words, our plan will be to repeatedly generate
samples from $\mc Y$ according to $\tpg$ and accept the first sample that
happens to lie in $\mcX \subset \mc Y$. Our choice of $\tpg$ to penalize
lower-quality solutions is a standard technique (see, e.g.,
\citet{porod2021dynamics}), and $\gamma$ is often referred to as an ``inverse
temperature'' parameter. We next specify $\Mgamma$ with the desired stationary
distribution $\tpg$.

We begin with a few definitions. For $\bx, \bx' \in \mc Y$, let $h(\bx, \bx')$
denote the Hamming distance, or the number of entries in which $\bx$ and $\bx'$
disagree (i.e., $h(\bx, \bx') \triangleq \|\bx - \bx'\|_0$). For $\bx, \bx' \in
\mc Y$ such that $h(\bx, \bx') = 1$, let $\delta(\bx, \bx')$ be the index $i$ on
which they disagree, so $\bx[i] \ne \bx'[i]$. Finally, let $\Delta(\bx, i)
\triangleq \{\bx_{i \gets g} : \bV \bx_{i \gets g} \preceq \bc, ~ g \in
\Zplus\}$. (Recall that $\bx_{i \gets g}$ denotes replacing the $i$th entry of
$\bx$ with the value $g$.) Intuitively, $\Delta(\bx, i)$ yields the set of
feasible knapsack solutions obtained by changing the $i$th entry of $\bx$ to
some nonnegative integer $g$. ($g$ need not differ from the existing $i$th entry
of $\bx$.) Using a variant of Gibbs sampling, we define
our Markov chain $\Mgamma$ over the state space $\mc Y$ to have transition
probabilities
\begin{align*}
  \Pgamma(\bx, \bx') \triangleq
  \begin{cases}
    \frac{\tpg(\bx')}{2n\sum_{\bx'' \in \Delta(\bx, i)} \tpg(\bx'')}
    & h(\bx, \bx') = 1 \wedge \delta(\bx, \bx') = i \\
    0 & h(\bx, \bx') > 1 \\
    1 - \sum_{\bx'' \ne \bx} \Pgamma(\bx, \bx'') & \bx = \bx'
  \end{cases}.
\end{align*}
In what follows, we describe useful properties of $\Pgamma$ and provide
additional intuition behind it.

\subsubsection{Properties of the simple chain}

In order to show that running MCMC on $\Mgamma$ yields an
$\varepsilon$-approximate sample from $\tpg$, we will show that:
\begin{itemize}
  \item $\Mgamma$ is lazy and irreducible.
  \item The stationary distribution of $\Mgamma$ is $\tpg$.
  \item We can implement the transitions $\Pgamma$ with an efficient algorithm.
\end{itemize}
First, observe that by construction, $\Mgamma$ is ``lazy,'' meaning
$\Pgamma(\bx, \bx) \ge 1/2$. This is because
\begin{align*}
  \Pgamma(\bx, \bx) &= 1 - \sum_{\bx' \ne \bx} \Pgamma(\bx, \bx') \\
  &= 1  -\sum_{i=1}^n \sum_{\bx' \in \Delta(\bx, i) \backslash \{\bx\}}
  \frac{\tpg(\bx')}{2n\sum_{\bx'' \in \Delta(\bx, i)} \tpg(\bx'')} \\
  &\ge 1 - \sum_{i=1}^n \sum_{\bx' \in \Delta(\bx, i)}
  \frac{\tpg(\bx')}{2n\sum_{\bx'' \in \Delta(\bx, i)} \tpg(\bx'')} \\
  &= \frac{1}{2}.
\end{align*}
Moreover, any state is reachable from any other state since all states are
connected to the empty solution $\bx = \mathbf{0}$. The simple chain is thus
both lazy and irreducible, which guarantees that $\Mgamma$ is aperiodic and
$\Pgamma$ has nonnegative eigenvalues.

Next, we will show that $\tpg$ is the stationary distribution of $\Mgamma$. For
an aperiodic, irreducible Markov chain, a distribution $\sigma$ that satisfies
the so-called ``detailed balance equations'' given by \cref{eq:detailed-balance}
is its unique stationary distribution (see, e.g., \citet[Cor. 1.17 and Prop.
1.20]{levin2017markov}): For all $\bx, \bx' \in \mcX$,
\begin{align*}
  \numberthis \label{eq:detailed-balance}
  \sigma(\bx) P(\bx, \bx') = \sigma(\bx') P(\bx', \bx).
\end{align*}
In our case, this is clearly true for $\bx = \bx'$ and when $h(\bx, \bx') > 1$
(since $\Pgamma(\bx, \bx') = \Pgamma(\bx', \bx) = 0$ in this case). For $h(\bx,
\bx') = 1$, let $i = \delta(\bx, \bx')$. Note that $\Delta(\bx, i) =
\Delta(\bx', i)$ by construction. Therefore,
\begin{align*}
  \tpg(\bx) \Pgamma(\bx, \bx')
  &= \frac{\tpg(\bx) \tpg(\bx')}{2n\sum_{\bx'' \in \Delta(\bx, i)}
  \tpg(\bx'')}
  = \frac{\tpg(\bx) \tpg(\bx')}{2n\sum_{\bx'' \in \Delta(\bx', i)} \tpg(\bx'')}
  = \tpg(\bx') \Pgamma(\bx', \bx).
\end{align*}
Thus, $\tpg$ is the stationary distribution of $\Mgamma$.

Finally, we must show that given $\bx$, we can efficiently sample from
$\Pgamma(\bx, \cdot)$. The following algorithm does so: with probability $1/2$,
remain at $\bx$. With the remaining probability $1/2$, choose $i \in [n]$
uniformly at random and enumerate the set $\Delta(\bx, i)$. (In our instances,
entries of $\bc$ are on the order of hundreds at most, meaning $|\Delta(\bx,
i)|$ is relatively small.) Then, sample $\bx'$ proportional to $\tpg(\cdot)$
from $\Delta(\bx, i)$ and transition to $\bx'$.\footnote{For general $\sigma$,
this entire process may require time proportional to $|\Delta(\bx, i)|$.
However, our implementation takes advantage of the structure of $\tpg(\cdot)$ to
do this in time proportional to $\log(|\Delta(\bx, i)|)$.} We write this
formally as \Cref{alg:gibbs} in \Cref{app:simple-alg}.

\subsubsection{The number of samples needed for rejection sampling}

To generate a sample from $\mcX$ (not just from $\mc Y \supset \mcX$), given
some hyperparameter $t$, we run $\Mgamma$ for $t$ iterations beginning with the
empty solution as the start state (i.e., $\bx_0 = \mathbf{0}$).\footnote{When we
  evaluate the simple chain, the starting state will not matter for the lower
bounds we show.} This yields a sample $\bx \sim \pi_{\gamma,t}$, where we define
$\pi_{\gamma,t} \triangleq \pi_0 \Pgamma^t$ to be the distribution of the random
walk after $t$ steps. If $\bx \in \mcX$, then we return $\bx$; if not, we reject
the sample and repeat the process until we find some $\bx \in \mcX$. See
\Cref{alg:gibbs} for details.

Define $\solutiondensity(t) \triangleq \pi_{\gamma,t}(\mcX) = \Pr_{\bX \sim
\pi_{\gamma,t}}[\bX \in \mcX]$ to be the probability that a sample from
$\pi_{\gamma,t}$ is in $\mcX$. Because the number of samples needed for
rejection sampling is geometrically distributed with parameter
$\solutiondensity(t)$, the total expected number of Markov chain iterations to
produce a sample from $\mcX$ is $t/\solutiondensity(t)$. To fully specify our
algorithm, we must choose $t$ to be sufficiently large to produce an
$\varepsilon$-approximate sample from $\mcX$.

We might think that it suffices to choose $t \ge \mixingtime(\varepsilon;
\Pgamma, \bx_0)$, since this implies that $\dTV(\pi_{\gamma,t}, \tpg) \le
\varepsilon$. However, this is insufficient: We want an
$\varepsilon$-approximate sample from $\mcX$, not from $\mc Y$. Roughly
speaking, a sample with approximation error $\varepsilon$ for $\tpg$ may have
approximation error on the order $\varepsilon / \tpg(\mcX)$ on $\mcX$. We
formalize this in \Cref{lem:eps-sol-density}, which we prove in
\Cref{app:simple-theory}.
\begin{restatable}{lemma}{epssoldensity}
  \label{lem:eps-sol-density}
  Let $\sigma$ and $\sigma'$ be distributions defined on $\mc Y$. Let
  $\sigma_{\mcX}$ and $\sigma_{\mcX}'$ be their respective conditional
  distributions on $\mcX \subset \mc Y$, i.e., for $\bx \in \mcX$,
  $\sigma_{\mcX}(\bx) = \sigma(\bx) / \sigma(\mcX)$. If $\dTV(\sigma, \sigma')
  \le \varepsilon$, then $\dTV(\sigma_{\mcX}, \sigma_{\mcX}') \le
  3\varepsilon/(2\sigma(\mcX))$. For any $\varepsilon$, there exist instances
  for which this is tight to within a constant factor.
\end{restatable}

Let $\solutiondensitystar \triangleq \lim_{t \to \infty} \solutiondensity(t) =
\tpg(\mcX)$. By \Cref{lem:eps-sol-density}, to generate an
$\varepsilon$-approximate sample from $\mcX$, it suffices to choose $t \ge
\mixinggammalb \triangleq \mixingtime(2\solutiondensitystar \varepsilon/3;
\Pgamma, \bx_0)$, since
\begin{align*}
  \dTV(\pi_{\gamma,t} \given \bx \in \mcX, \tpg \given \bx \in \mcX)
  \le \frac{3 \dTV(\pi_{\gamma,t}, \tpg)}{2 \tpg(\mcX)}
  \le \frac{3 (2\solutiondensitystar \varepsilon/3)}{2 \solutiondensitystar}
  = \varepsilon.
\end{align*}
In other words, a $2\solutiondensitystar \varepsilon/3$-approximate sample from
$\tpg$ over $\mc Y$ yields an $\varepsilon$-approximate sample on from $\pi$
over $\mcX$. This is tight to within a constant factor. Of course, we know
neither $\solutiondensitystar$ nor $\mixingtime(\cdot; \Pgamma, \bx_0)$
\textit{a priori}; we experimentally determine them for a subset of our
instances in \Cref{sec:simple-empirical}. Choosing $t$ to be $\mixinggammalb$,
the expected number of MCMC iterations needed to produce an
$\varepsilon$-approximate sample from $\mcX$ is
\begin{align*}
  \numberthis \label{eq:num-samp}
  \numsampgamma(\varepsilon) \triangleq
  \frac{\mixinggammalb}{\solutiondensity(\mixinggammalb)},
\end{align*}
where the numerator is the number of MCMC iterations needed per sample from $\mc
Y$ and the denominator is the probability that we accept a sample (i.e., the
probability that it lies in $\mcX$).

\subsubsection{Empirical results for the simple chain}
\label{sec:simple-empirical}

To evaluate the performance of \Cref{alg:gibbs}, our goal will be to determine
$\numsampgamma(\varepsilon)$ for our instances. To do so, we will bound both
$\mixinggammalb$ and $\solutiondensity(\mixinggammalb)$. Observe that because
$\dTV(\pi_{\gamma,\mixinggammalb}, \tpg) \le 2\solutiondensitystar \varepsilon/3$,
by definition of total variation distance \cref{eq:tvd},
\begin{align*}
  \numberthis \label{eq:sol-density-bounds}
  \p{1 - \frac{2\varepsilon}{3}} \solutiondensitystar \le
  \solutiondensity(\mixinggammalb) \le \p{1 + \frac{2
  \varepsilon}{3}} \solutiondensitystar.
\end{align*}
Let $\relaxationgamma \triangleq \frac{1}{1 - \lambda_2(\Pgamma)}$ be the
relaxation time of $\Mgamma$. Combining
\cref{eq:mt-bounds,eq:sol-density-bounds} yields
\begin{align*}
  \numberthis \label{eq:num-samp-lb}
  \numsamplbgamma(\varepsilon) \triangleq \frac{(\relaxationgamma-1) \log
  \p{\frac{3}{4\varepsilon \solutiondensitystar}}}{\p{1 +
  \frac{2\varepsilon}{3}}\solutiondensitystar}
  \le \numsampgamma(\varepsilon)
  \le \frac{\relaxationgamma \log \p{\frac{3}{2 \varepsilon \solutiondensitystar
  \tpg(\bx_0)} }}{\p{1 - \frac{2\varepsilon}{3}}\solutiondensitystar}
  \triangleq \numsampubgammaeps.
\end{align*}
We will use these tight lower and upper bounds for $\numsampgamma(\varepsilon)$
to characterize the performance of the simple chain. We choose $\varepsilon =
1/(2e)$ by convention and write $\numsampgamma = \numsampgamma(1/(2e))$.

\begin{figure}[ht]
  \centering
  \begin{subfigure}[b]{0.48\textwidth}
  \includegraphics[width=\textwidth]{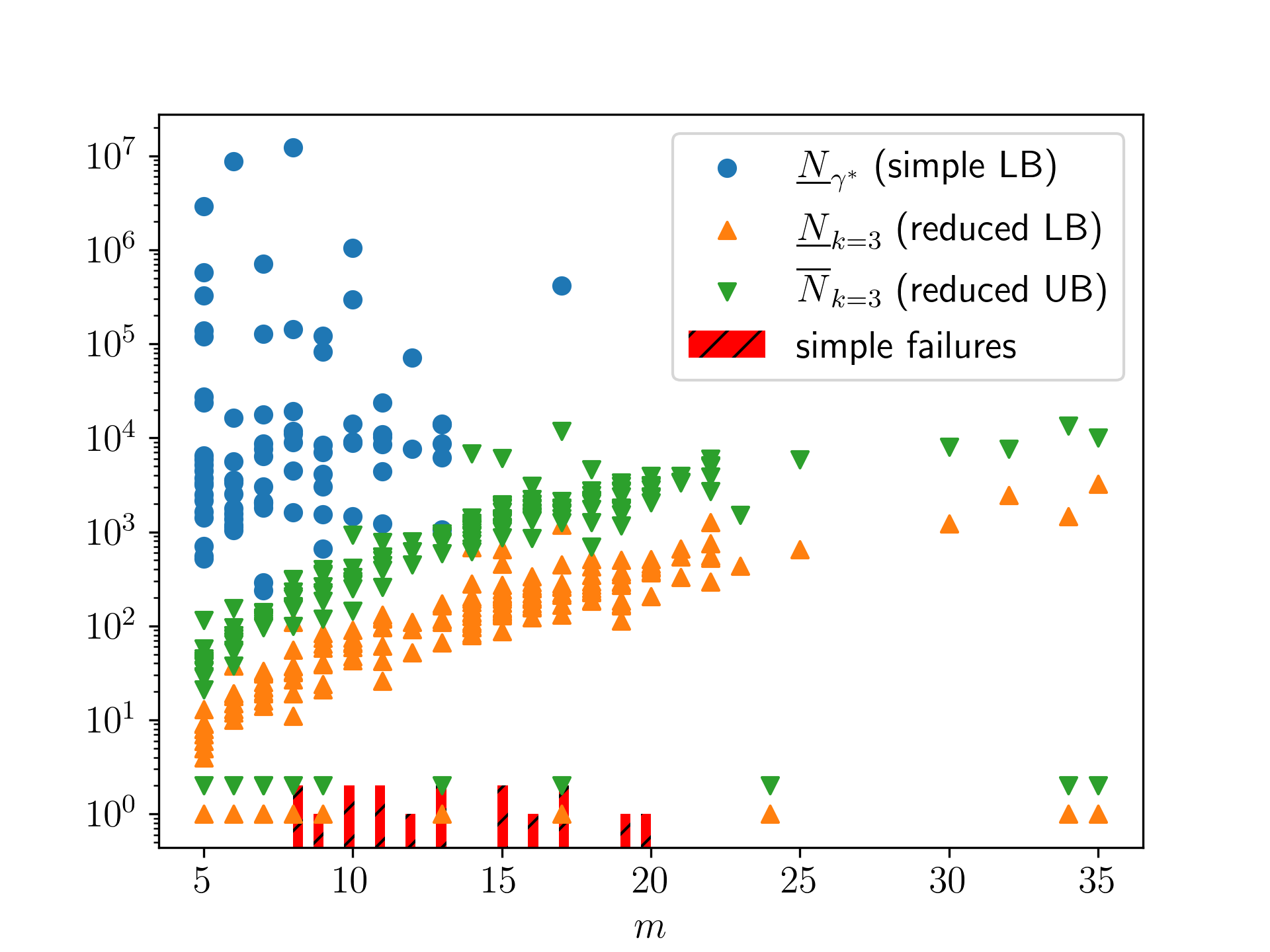}
    \caption{AL}
    \label{fig:all-mt-AL}
  \end{subfigure}
  \hfill
  \begin{subfigure}[b]{0.48\textwidth}
  \includegraphics[width=\textwidth]{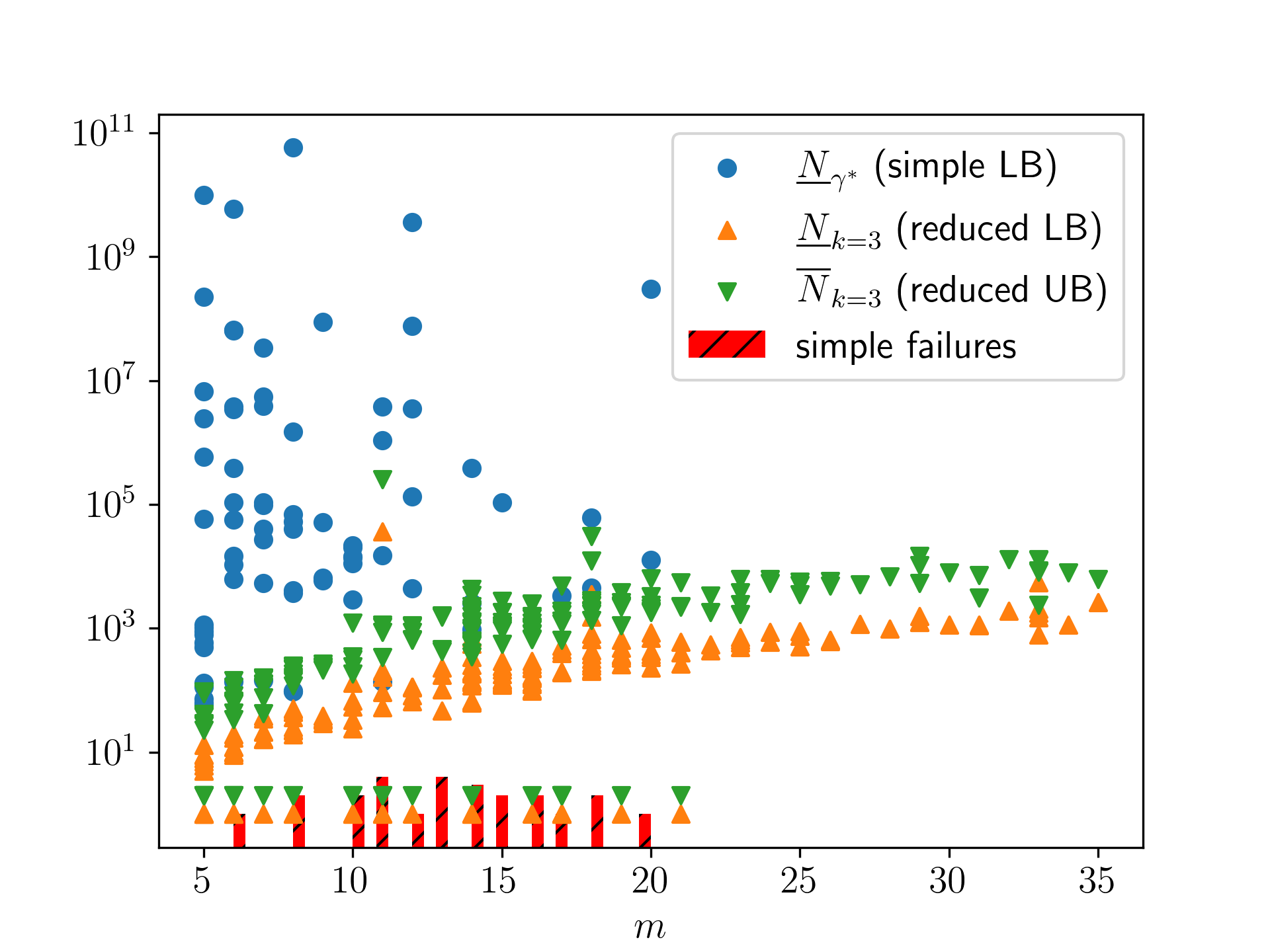}
    \caption{NV}
    \label{fig:all-mt-NV}
  \end{subfigure}
  \caption{There are blocks for which the number of iterations required to
    generate an $\varepsilon$-approximate sample from $\mcX$ using the simple
    chain is large. The reduced chain (described in \Cref{sec:reduced}) requires
  orders of magnitude fewer iterations in the worst case.}
  \label{fig:all-mt}
\end{figure}

To compute $\numsamplbgamma$ for a given Census block, we will need to compute
$\lambda_1(\Pgamma)$ and $\solutiondensitystar$. Unfortunately, this requires
$\Omega(|\mc Y|)$ time and space just to write $\Pgamma$. This can be
prohibitive for our instances since $\mc Y$ can be exponentially (in $n$) large. Our
approach is as follows: we choose a subset of ``small'' instances and compute
the full transition matrix $\Pgamma$ for a random sample of those instances. We
will show that in this sample, there exist instances for which $\numsamplbgamma$
is prohibitively large, making \Cref{alg:gibbs} impractical in our setting.

Recall that $m$ is the number of households in a Census block. We sample 100
blocks where $5 \le m \le 20$ and $|\mcX| < 5000$ in two states: Alabama and
Nevada (chosen arbitrarily). For each block, we seek to compute $\Pgamma$ for
$\gamma \in \Gamma \triangleq \{0.0, 0.2, 0.4, 0.6, 0.8, 1.0, 1.2\}$. For a
significant number of blocks (indicated on \Cref{fig:all-mt} in red bars; 25 in
NV and 17 in AL), we fail to compute $\Pgamma$ within a time limit of 8 hours.
In \Cref{app:additional-empirical}, we show that the globally optimal choice for
$\gamma$ on our sample is $\approx 0.8$.\footnote{Due to numerical instability
in computing eigenvalues of $\Pgamma$, we sometimes underestimate
$\numsamplbgamma$ for larger values of $\gamma$.}

In \Cref{fig:all-mt}, we plot $\numsamplbgammastar$ for the blocks in our
sample, where for each block, $\numsamplbgammastar \triangleq \min_{\gamma \in
\Gamma} \numsamplbgamma$ using blue \textcolor{NavyBlue}{$\CIRCLE$} markers.
(This is simply taking the ``best'' choice of $\gamma$ for each block.) Observe
that $\numsamplbgammastar$ can be quite large, even for small instances. There
are blocks with fewer than 10 households for which $\numsamplbgammastar$ is on
the order of $10^{11}$. In contrast, the ``reduced'' approach we will develop in
\Cref{sec:reduced} appears to perform much better: the analogous lower and upper
bounds for the reduced chain are small for all of our instances
(\Cref{fig:all-mt}, orange \textcolor{orange}{$\blacktriangle$} and green
\textcolor{ForestGreen}{\rotatebox[origin=c]{180}{$\blacktriangle$}} markers
respectively).\footnote{While we compare the number of MCMC iterations as
  opposed to computation time here, our experiments show that computation time
  per iteration is similar for the two Markov chains we consider. See
\Cref{app:additional-empirical} for details.}

\subsubsection{Theoretical results for the simple chain}

Before we describe this improved approach, we briefly provide a theoretical
characterization of how $\numsamplbgamma$ varies with $\gamma$ to gain some
intuition as to why the simple chain performs poorly and how $\numsampgamma$
depends on $\gamma$. Intuitively, we will see that varying $\gamma$ creates a
trade-off between $\solutiondensitystar$ and $\relaxationgamma$. (Recall that
$\relaxationgamma$ is the relaxation time of $\Mgamma$.) Since $\numsampgamma
\approx \relaxationgamma / \solutiondensitystar$, these opposing forces empirically
make the simple chain impractical in our setting.

We begin by characterizing how $\solutiondensitystar$ varies with $\gamma$. All
proofs are deferred to \Cref{app:simple-theory}.
\begin{restatable}{lemma}{soldensity}
  \label{lem:sol-density-increasing}
  $\solutiondensitystar$ is monotonically increasing in $\gamma$, and
  \begin{align*}
    \lim_{\gamma \to -\infty} \solutiondensitystar = 0
    &&\text{and}&&
    \lim_{\gamma \to \infty} \solutiondensitystar = 1.
  \end{align*}
\end{restatable}

As we might expect, large values of $\gamma$ make it more likely that we sample
an exact solution $\bx \in \mcX$, since larger values of $\gamma$ penalize
inexact solutions more.
Unfortunately, this comes at a cost, which
we show in \Cref{lem:mixing-time}: in the limit, as $\gamma$ increases,
$\relaxationgamma$ increases exponentially.\footnote{A stronger version of this
  result would claim that, analogously to \Cref{lem:sol-density-increasing},
  mixing time increases monotonically with $\gamma$. While this may be true,
proving monotonicity over the temperature parameter is notoriously difficult
(see, e.g.,~\citet{nacu2003glauber,kargin2011relaxation} and \citet[Ch. 26 Open
Question 2]{levin2017markov}).}
\begin{restatable}{lemma}{expmixingtime}
  \label{lem:mixing-time}
  For the stationary distribution $\tpg$ as defined in \Cref{eq:tpg-def},
  $\relaxationgamma = \Omega(\exp(\gamma \min_i \|\bv_i\|_1))$.
\end{restatable}
We prove \Cref{lem:mixing-time} using a conductance argument and Cheeger's
inequality~\citep{lawler1988bounds,jerrum1988conductance}. Taken together,
\Cref{lem:sol-density-increasing,lem:mixing-time} describe the trade-off in our
choice of $\gamma$: for small values of $\gamma$, samples from $\tpg$ rarely
fall in $\mcX$, making rejection sampling inefficient. For large values of
$\gamma$, the mixing time of the simple chain increases exponentially, requiring
many iterations to generate each sample. These results tell us that the optimal
choice of $\gamma$ is finite.

\begin{restatable}{lemma}{finitegamma}
  \label{lem:finite-gamma}
  For every instance, there is some finite $\gamma$ that minimizes $\numsampgamma$.
\end{restatable}
\Cref{fig:num-samp-gamma} visualizes this trade-off. Each line represents a
different Census block. Red triangle \textcolor{red}{$\blacktriangle$} markers
show low-precision estimates for the spectral gap $1-\lambda_2$, which will
typically mean we underestimate $\numsamplbgamma$. Observe that each line
appears to be quasiconvex, with a minimum in our range of choices for $\gamma$.

\begin{figure}[ht]
  \centering
  \begin{subfigure}[b]{0.48\textwidth}
    \includegraphics[width=\textwidth]{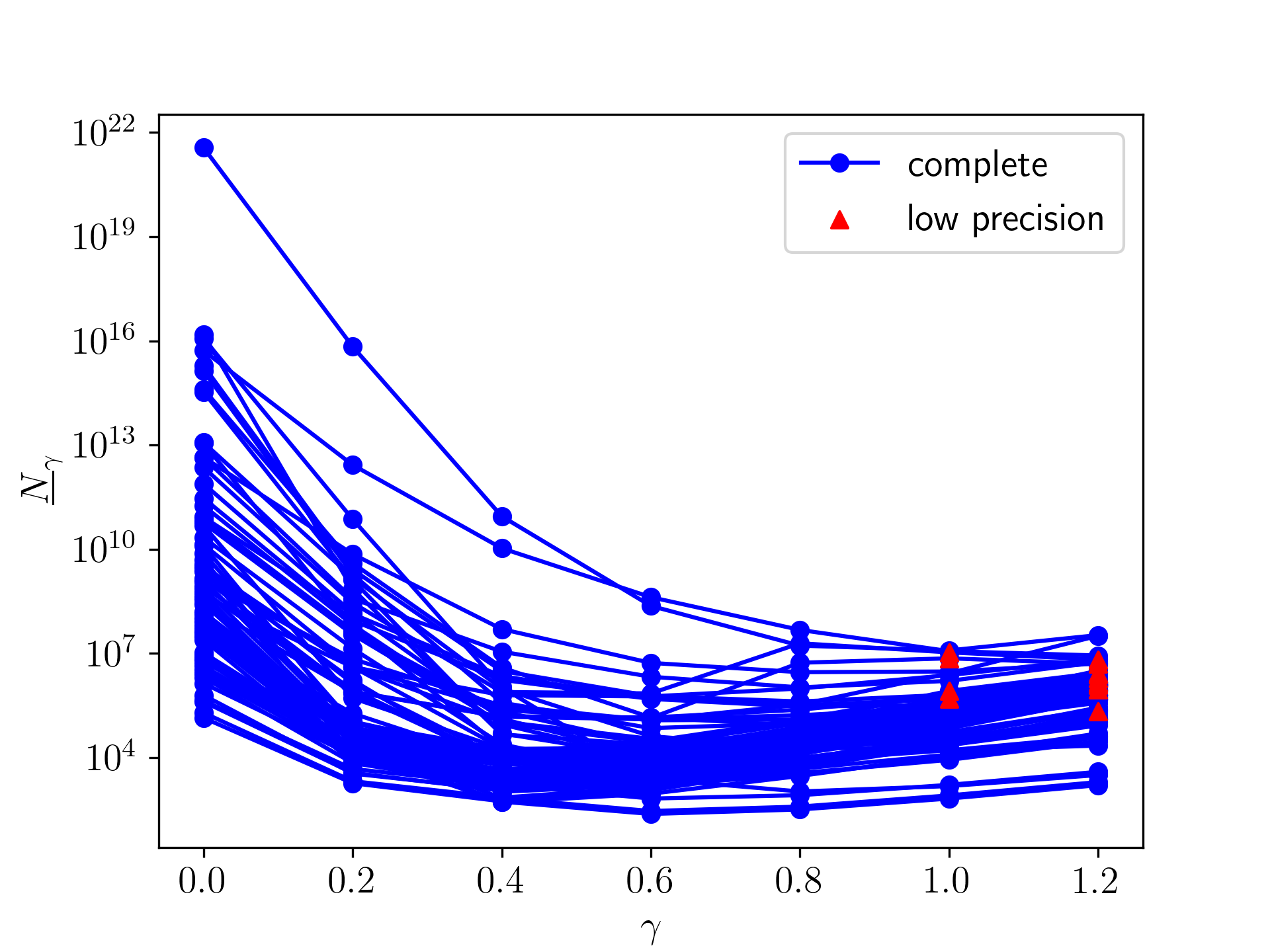}
    \caption{AL}
    \label{fig:num-samp-gamma-AL}
  \end{subfigure}
  \hfill
  \begin{subfigure}[b]{0.48\textwidth}
    \includegraphics[width=\textwidth]{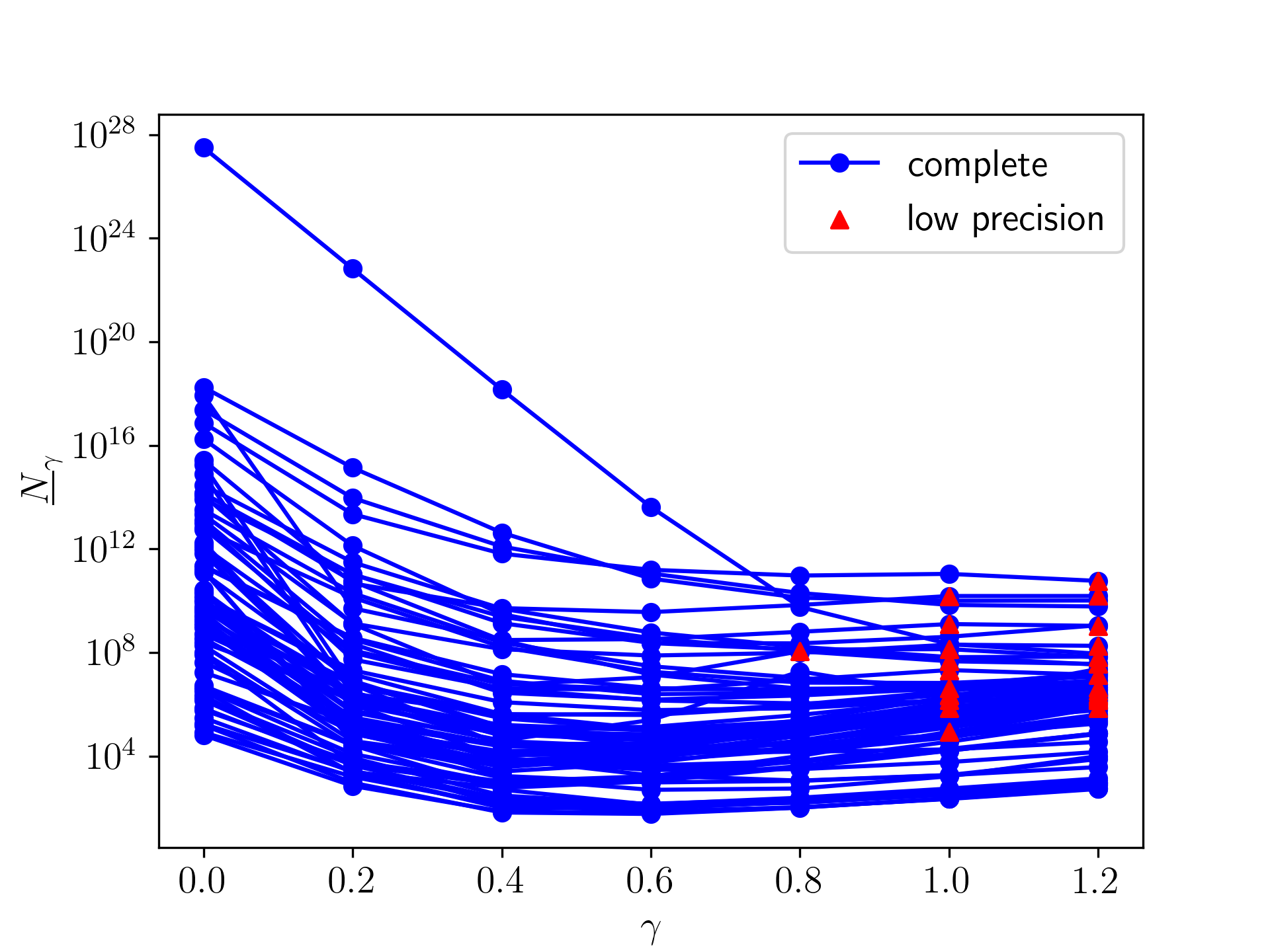}
    \caption{NV}
    \label{fig:num-samp-gamma-NV}
  \end{subfigure}
  \caption{$\numsamplbgamma$ as a function of $\gamma$.}
  \label{fig:num-samp-gamma}
\end{figure}

\subsection{The reduced chain}
\label{sec:reduced}

Part of the reason why the simple chain requires many iterations in expectation
to produce a valid sample from $\mcX$ is its potentially large state space:
because $\mc Y$ can be much larger than $\mcX$, a random walk according to
$\Pgamma$ may yield samples from $\mc Y \backslash \mcX$ with very high
probability. This motivates our approach: the ``reduced'' chain. We design a
Markov chain $\Mreduced$ with state space $\mcX$ instead of $\mc Y$.

\subsubsection{Defining the reduced chain}

The reduced chain is parameterized by an integer $k \ge 2$. Intuitively, given a
solution $\bx \in \mcX$, we randomly remove $k$ elements from $\bx$ and replace
them with another multiset of $k$ elements (found via ILP) such that the
resulting sum still exactly equals the constraint $\bc$. For small $k$, this can
be done fairly quickly. In our experiments, we use $k \in \{2, 3, 4\}$.

\def\pairs{\mc A}

We again use a variant of Gibbs sampling to induce the desired stationary
distribution. We define the following transition matrix $\Preduced$ for our
Markov chain $\Mreduced$. Let $\pairs(\bx, \bx')$ be the set of distinct pairs of
multisets $(\bz, \bz')$, each of size $k$, such that we can transform $\bx$ into
$\bx'$ by removing elements from $\bz$ and replacing them with elements from
$\bz'$. Formally, this is
\begin{align*}
  \pairs(\bx, \bx') \triangleq \{(\bz, \bz') \in \Zplus^n \times \Zplus^n: (\bV
    \bz = \bV \bz') \wedge (\bx - \bz = \bx' - \bz' \succeq \mathbf{0}) \wedge
  (\|\bz\|_1 = \|\bz'\|_1 = k)\}.
\end{align*}
Let $\mc Z(\bz) \triangleq \{\bz' \in \Zplus^n : (\bV \bz' = \bV \bz) \wedge
(\bz' \succeq \mathbf{0}) \wedge \|\bz'\|_1 = k\}$.\footnote{Equivalently, we
could write $\mc Z(\bz) = \{\bz' \in \Zplus^n : \exists \bx' \in \mcX \text{
such that }(\bz, \bz') \in \mc A(\bx, \bx')\}$.} Then, for $\bx, \bx' \in
\mcX$,\footnote{Because the elements in $\mc V$ are unique, it is impossible to
have $h(\bx, \bx') = 1$ for $\bx, \bx' \in \mcX$.}
\begin{align*}
  \Preduced(\bx, \bx') \triangleq
  \begin{cases}
    \frac{1}{2 \binom{m}{k}} \sum_{(\bz, \bz') \in \pairs(\bx, \bx')}
    \frac{\pi(\bx')}{\sum_{\bz'' \in \mc Z(\bz)} \pi(\bx - \bz + \bz'')} & 2 \le
    h(\bx, \bx') \le 2k \\
    0 & h(\bx, \bx') > 2k \\
    1 - \sum_{\bx'' \ne \bx} \Preduced(\bx, \bx'') & \bx = \bx'
  \end{cases}.
\end{align*}
Because $\mc Z(\bz) = \mc Z(\bz')$ for $(\bz, \bz') \in \pairs(\bx, \bx')$,
$\pi$ satisfies the detailed balance conditions $\pi(\bx) \Preduced(\bx, \bx') =
\pi(\bx') \Preduced(\bx', \bx)$, making it a stationary distribution of
$\Mreduced$. Unfortunately, $\pi$ is not necessarily the \textit{unique}
stationary distribution of $\Mreduced$, which we will discuss below. Note that
$\Preduced$ is lazy by construction. In \Cref{alg:reduced} in
\Cref{app:reduced-alg}, we show that we can sample efficiently from
$\Preduced(\bx, \cdot)$ as long as we can efficiently compute $\mc Z(\bz)$ (via
ILP), which empirically we can in our setting. The complexity of this
computation depends on the problem dimension $d$ and the parameter choice $k$,
\textit{not} on the total counts $\bc$.

\paragraph*{Choosing an initial state $\bx_0$.}

To fully specify the algorithm, we must choose an initial state $\bx_0$. In
general, we could find an arbitrary $\bx_0 \in \mcX$ via ILP; however, in order
to produce tighter bounds on performance, we aim to begin with a
``better-than-average'' $\bx_0$ (i.e., such that $\pi(\bx_0) \ge 1/|\mcX|$).
To see why, observe that the upper bounds in \Cref{eq:mt-bounds} depend on
$\log(1/\pi(\bx_0))$, so larger values of $\pi(\bx_0)$ will reduce the number of
MCMC iterations we need to produce a sample.
We defer these details to \Cref{app:reduced-alg}. Note that finding any $\bx_0
\in \mcX$ requires that we can solve the decision problem \probname\ via ILP
despite its NP-hardness; empirically, we can for all of our instances.

\paragraph*{The reduced chain is not necessarily irreducible.}

In an irreducible Markov chain, any state is reachable from any other state,
i.e., for any $\bx, \bx' \in \mcX$, there exists $t$ such that $\Preduced^t(\bx,
\bx') > 0$. This is not necessarily the case for $\Preduced$: it is possible to
have disconnected components in the state space graph (see \Cref{fig:components}
for an example). When this happens, $\pi$ is not the only stationary
distribution of $\Mreduced$, and our algorithm will not sample from $\pi$ as
desired. We evaluate whether this is the case for our instances in what follows.

\subsubsection{Empirical results for the reduced chain}
\label{sec:reduced-empirical}

We are interested in $\numsampk(\varepsilon) \triangleq \mixingtime(\varepsilon;
\Preduced, \bx_0)$, defined to be the number of Markov chain iterations required
to generate an $\varepsilon$-approximate sample from $\mcX$. If $\Mreduced$ is
not irreducible, $\numsampk(\varepsilon)$ is undefined. As in
\Cref{sec:simple-empirical}, we derive bounds for $\numsampk(\varepsilon)$ and
evaluate these bounds for a random sample of instances.
For irreducible $\Mreduced$, let $\relaxationk \triangleq \frac{1}{1 -
\lambda_2(\Preduced)}$. Using \cref{eq:mt-bounds},
\begin{align*}
  \numberthis \label{eq:numsampk-bounds}
  \numsamplbk(\varepsilon)
  \triangleq (\relaxationk-1) \log
  \p{\frac{1}{2\varepsilon}}
  \le \numsampk(\varepsilon)
  \le \relaxationk \log \p{\frac{1}{\varepsilon \pi(\bx_0)}}
  \triangleq
  \numsampubk(\varepsilon).
\end{align*}
We again choose $\varepsilon = 1/(2e)$. Because our choice of $\bx_0$ satisfies
$\pi(\bx_0) \ge 1/|\mcX|$, the lower and upper bounds are within an $O(\log
|\mcX|)$ factor of one another. In \Cref{lem:state-space-ub} in
\Cref{sec:reduced-theory}, we show that $\log |\mcX| = O(m \log n)$, meaning
$\numsampk = O(\relaxationk \cdot m \log n)$ (see \Cref{cor:reduced-mt} below).

We re-use the same sample of 100 blocks per state that we used to evaluate the
simple chain in \Cref{sec:simple-empirical}. In addition, because $|\mcX|$ can
be much smaller than $|\mc Y|$, we are able to compute $\lambda_2(\Preduced)$
for larger instances. We therefore sample another 100 blocks for each of AL and
NV with $14 \le m \le 35$ and $|\mcX| < 5000$.\footnote{In Alabama and Nevada,
  72\% and 45\% of blocks respectively have $m \le 35$, $|\mcX| < 5000$, and
  sufficiently complete data for our use (see \Cref{app:encoding} for details
about excluded blocks).} For each block, we analyze $k \in \{2, 3,
4\}$.\footnote{For some blocks, our computation times out (beyond a limit of 8
  hours) for $k=4$. Because our results suggest that $k=3$ suffices, these
failures have little effect on our conclusions.}

\begin{figure}[ht]
  \centering
  \begin{subfigure}[b]{0.48\textwidth}
  \includegraphics[width=\textwidth]{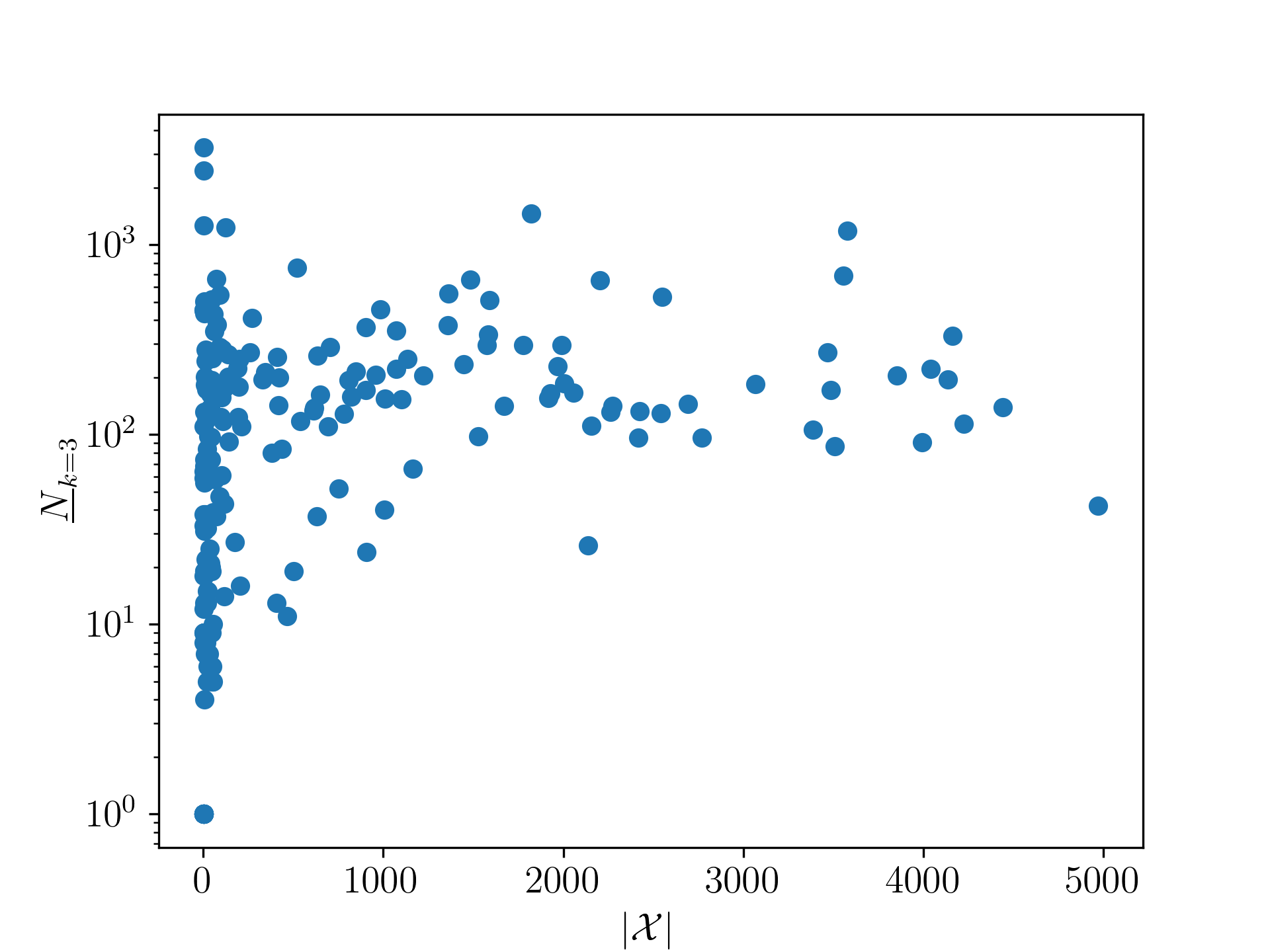}
    \caption{AL}
    \label{fig:reduced_num_solutions_mt-AL}
  \end{subfigure}
  \hfill
  \begin{subfigure}[b]{0.48\textwidth}
  \includegraphics[width=\textwidth]{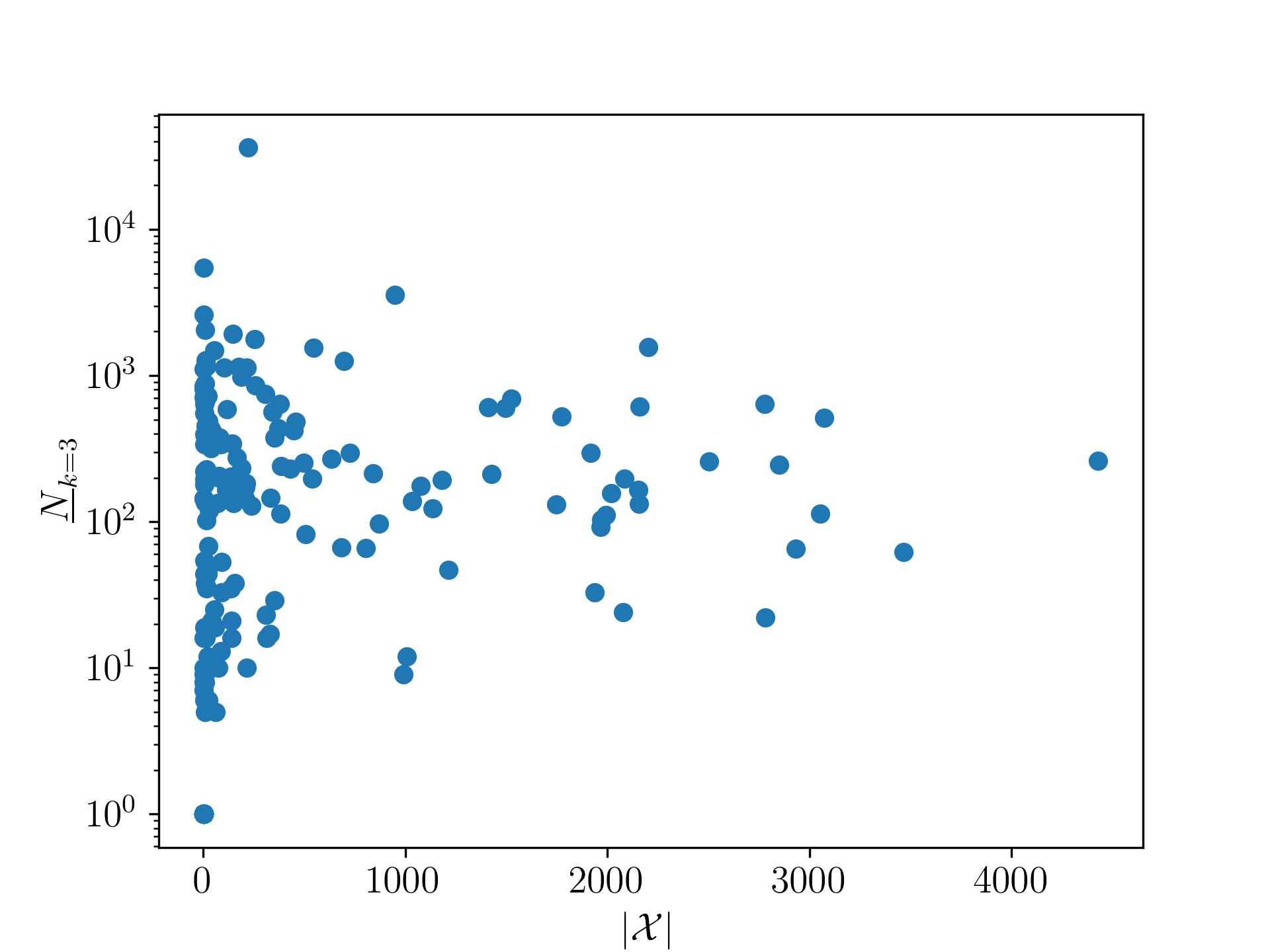}
    \caption{NV}
    \label{fig:reduced_num_solutions_mt-NV}
  \end{subfigure}
  \caption{There is no clear relationship between $|\mcX|$ and the central
  tendency of $\numsamplbk$ for $k=3$.}
  \label{fig:reduced_num_solutions_mt}
\end{figure}

\Cref{fig:all-mt} shows $\numsamplbk$ and $\numsampubk$ as a function of $m$ for
$k=3$. The reduced chain requires orders of magnitude fewer iterations than the
simple chain to generate an $\varepsilon$-approximate sample from $\mcX$.
Interestingly, $\numsampk$ appears to be much more consistent as a function of
$m$ than $\numsamplbgamma$ is. Our sample includes two cutoffs (on $m \le 35$
and $|\mcX| < 5000$), and we might worry that $\numsampk$ will grow dramatically
beyond these cutoffs. While we cannot extrapolate beyond them, \Cref{fig:all-mt}
provides evidence that $\numsampk$ is not growing too quickly near the $m=35$
cutoff. \Cref{fig:reduced_num_solutions_mt} shows that there is no clear
relationship between the central tendency of $\numsampk$ and $|\mcX|$, and some
of the largest values of $\numsampk$ come from blocks with small $|\mcX|$.
Again, while we cannot extrapolate beyond the cutoffs in our sample,
\Cref{fig:all-mt,fig:reduced_num_solutions_mt} do not suggest that $\numsampk$
will be prohibitively large outside of these cutoffs.

\begin{figure}[ht]
  \centering
  \begin{subfigure}[b]{0.48\textwidth}
    \includegraphics[width=\textwidth]{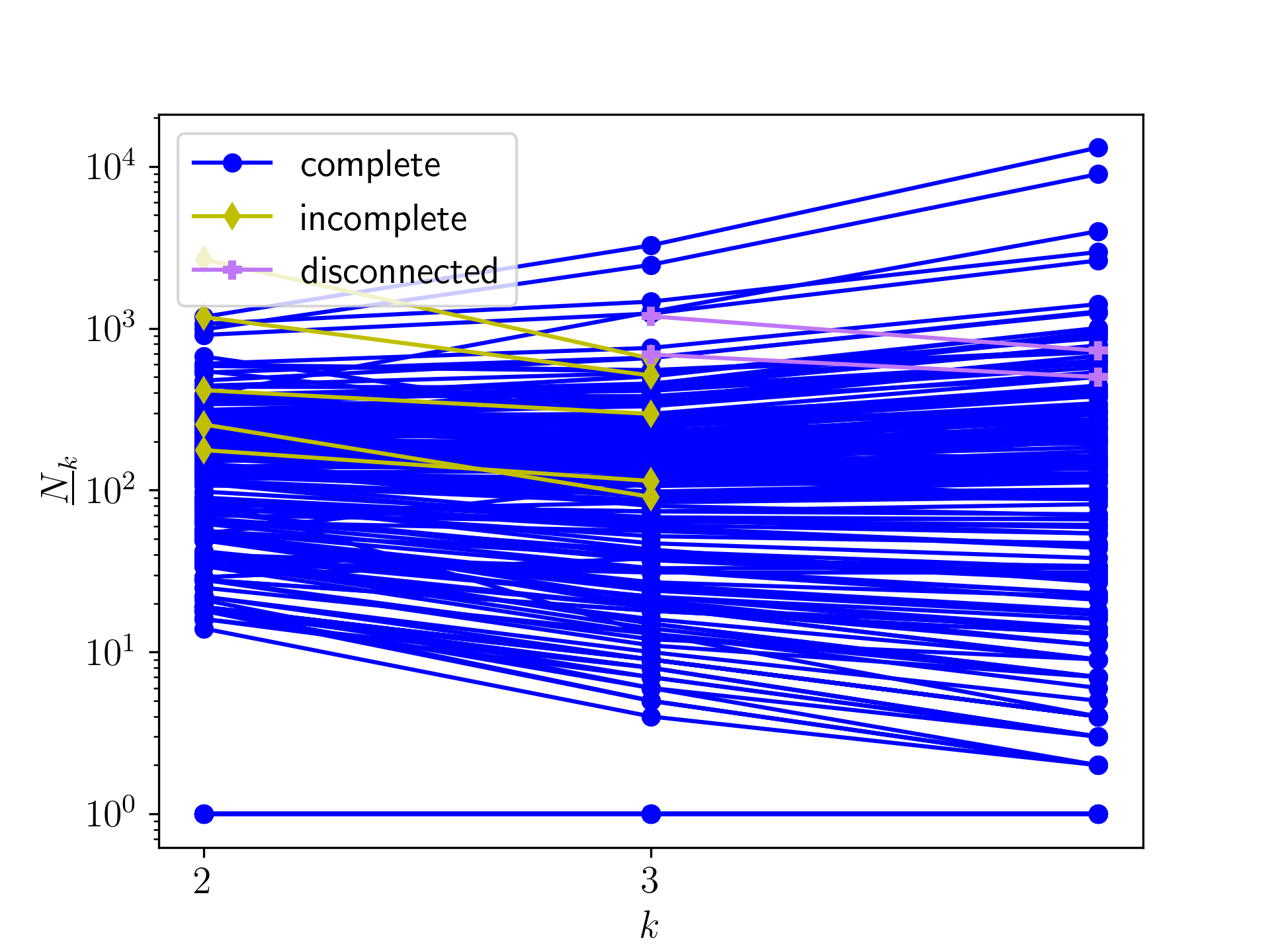}
    \caption{AL}
    \label{fig:reduced-k-mt-AL}
  \end{subfigure}
  \hfill
  \begin{subfigure}[b]{0.48\textwidth}
    \includegraphics[width=\textwidth]{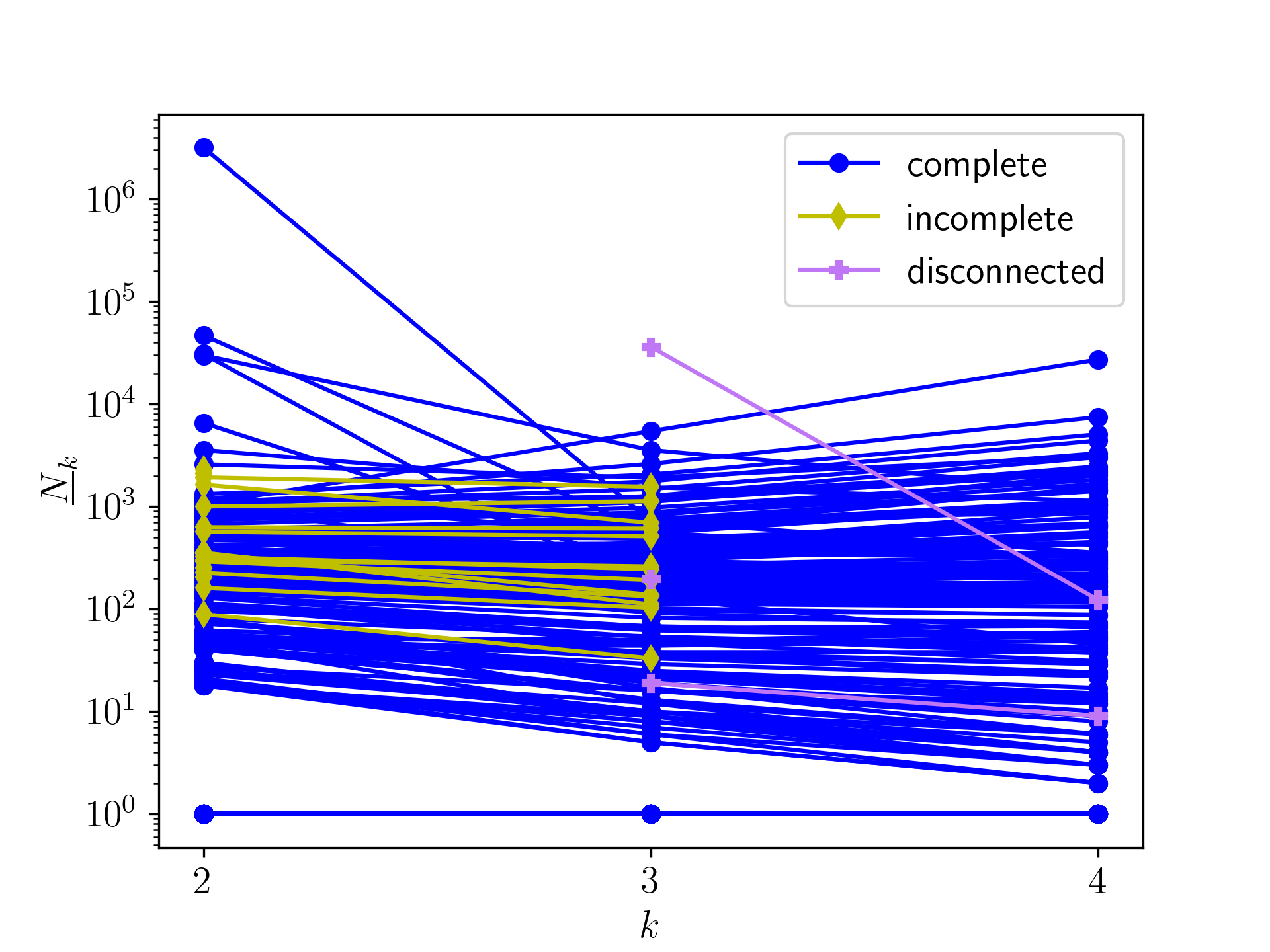}
    \caption{NV}
    \label{fig:reduced-k-mt-NV}
  \end{subfigure}
  \caption{$k=3$ appears to be optimal because (1) there are blocks for which
    $M_{k=2}$ is disconnected, and (2) there are blocks for which $M_{k=2}$ is
  connected but has high mixing time.}
  \label{fig:reduced-k-mt}
\end{figure}

We find that for $k=2$, $\Mreduced$ is not irreducible (i.e., the state
space graph is disconnected) for 2 out of 200 blocks in AL and 3 out of 200
blocks in NV. On the other hand, all blocks in our sample yield irreducible
Markov chains for $k \ge 3$. (By construction, if the reduced chain is
irreducible for some $k$, it is irreducible for any $k' > k$.) Based on these
results, $k=3$ appears to be a good choice in our setting.
\Cref{fig:reduced-k-mt} provides a more detailed comparison across $k \in \{2,
3, 4\}$. Purple plus \textcolor{DarkOrchid}{$\Plus$} markers denote blocks where
the reduced chain is not irreducible for $k=2$. Gold diamond
\textcolor{Goldenrod}{$\blacklozenge$} markers denote blocks for which
computation times out (beyond an 8-hour limit) for $k=4$.

Finally, we comment on the computational cost of running the reduced chain in
practice. Our results suggest that roughly $10^5$ MCMC iterations suffices to
generate an $\varepsilon$-approximate sample. In
\Cref{app:additional-empirical}, we provide experimental results showing that
running this many iterations on an M1 MacBook Pro takes roughly 10 seconds per
Census block. Because we process each Census block independently, this can be
run efficiently in parallel on a standard computing cluster, requiring tens of
days of computing for the states in question (see
\Cref{app:additional-empirical} for more details).

\subsubsection{Theoretical results for the reduced chain}
\label{sec:reduced-theory}

Here, we provide theoretical results for the reduced chain. In
\Cref{app:bad-examples}, we construct families of pathological instances
where (1) the reduced chain is disconnected for any $k < |\mcX|$
(\Cref{ex:disconnected}), and (2) the reduced chain is connected for $k=2$ but
has exponentially high mixing time (\Cref{ex:exp-mixing}). Despite this,
$\Mreduced$ appears to perform quite well for our instances, even though we
cannot obtain general subexponential bounds on $\numsampk$ unless P = NP. Here,
we attempt to provide intuition for why $\Mreduced$ performs well. First, we
bound $\log |\mcX|$. All proofs are deferred to
\Cref{app:proofs-reduced}.
\begin{restatable}{lemma}{statespaceub}
  \label{lem:state-space-ub}
  $\log |\mcX| < \log |\mc Y| \le m (1 + \log(n+1))$.
\end{restatable}
Recall that $m$ is the number of households in a Census block. Given that $m$ is
reasonably small in our instances (typically on the order of tens, and at most a
few thousand), \Cref{lem:state-space-ub} and~\Cref{eq:numsampk-bounds} imply
that $\numsampk$ is primarily determined by the relaxation time $\relaxationk =
1/(1 - \lambda_2(\Preduced))$.

\begin{corollary}
  \label{cor:reduced-mt}
  If $\pi(\bx_0) \ge \frac{1}{|\mcX|}$, then
  $\numsampk = \relaxationk \cdot O\p{m \log n}$.
\end{corollary}
We cannot in general provide bounds on $\lambda_2(\Preduced)$, but our empirical
results suggest that it is sufficiently small (bounded away from $1$) in our
setting. We also cannot guarantee that $\Mreduced$ is irreducible; however,
we next provide a general condition under which $\Mreduced$ is irreducible. This
condition is empirically not met in our case, but as we will argue, it is
\textit{nearly} met, and might provide intuition as to why $\Mreduced$ appears
to be irreducible in practice.

\def\ajmin{a_j^{\min}}
\def\ajmax{a_j^{\max}}

For attribute $j \in [d]$, let $\ajmin$ and $\ajmax$ denote the minimal and
maximal value that any household vector $\bv_i$ takes in that position.
Formally, $\ajmin \triangleq \min_{i \in [n]} \bv_i[j]$, and $\ajmax$ is defined
analogously. Let $\mc A_i \subset \Zplus$ be the set $\{\ajmin, \ajmin + 1,
\dots, \ajmax\}$. And finally, let $\mc A \subset \Zplus^d$ be the set of
integer vectors $\mc A_1 \times \mc A_2 \times \dots \times \mc A_d$.
Informally, $\mc A$ is the set of integer vectors contained in the
hyperrectangle formed by the ranges of each attribute. By this definition $\mc V
\subseteq \mc A$, since each household type vector must lie within $\mc A$. In
what follows, we provide a condition under which $\Mreduced$ is irreducible.

\begin{restatable}{lemma}{hyperrectangle}
  \label{lem:hyperrectangle}
  If $\mc V = \mc A$, then $\Mreduced$ is irreducible for $k \ge 2$.
\end{restatable}

\Cref{lem:hyperrectangle} provides some insight as to why we might expect
connectivity in $\Mreduced$. Most of the attributes in our $d$-dimensional
encoding represent counts of individuals within each household who belong to
certain demographic groups (e.g., number of adults, number of Hispanic
individuals, etc.). Given two solutions $\bx, \bx' \in \mcX$, when should we
able to transition between them by only changing a small number of households
(i.e.,~$k=2$) at a time? Intuitively, if $\bx$ contains two households $\bv_i$
and $\bv_j$, suppose we were to shuffle the individuals within those households,
producing two new household types $\bv_{i'}$ and $\bv_{j'}$. If $\mc V = \mc A$,
both of these new household types will also be in $\mc V$. And because we're
using the same set of individuals, we have not changed the overall demographic
composition of the block. As a result, replacing $\bv_i, \bv_j$ with $\bv_{i'},
\bv_{j'}$ will still yield a valid solution $\bx'' \in \mc X$. Doing this
repeatedly allows us to find a path from $\bx$ to $\bx'$, meaning $\Mreduced$ is
irreducible.

Thus, if $\mc V$ is very ``dense,'' meaning all possible household types within
some range have nonzero probability, $\Mreduced$ is irreducible. In practice,
$\mc V$ is not this dense; there exist some household types we never observe in
the PUMS, meaning $\mc V \subset \mc A$. On the other hand,
\Cref{lem:hyperrectangle} provides a particularly strong notion of connectivity.
In the proof, we show that when $\mc V = \mc A$, there are multiple disjoint
paths between pairs $\bx, \bx' \in \mcX$, not just the single path required for
irreducibility. Future work may be able to relax this assumption to show
irreducibility under weaker conditions, for example, when household types are
missing with some probability $p$~\citep[see, e.g.,][]{percolation1999grimmett}.

Finally, even when $\Mreduced$ is not irreducible, the reduced chain improves
solution quality: a standard coupling argument shows that for any initial
distribution $\pi_0$, $\|\pi_0 \Preduced - \pi\|_1 \le \|\pi_0 - \pi\|_1$, even
when $\Mreduced$ is not irreducible. Moreover, we show that $\Mreduced$
converges to a distribution that is in a sense optimal given initial
distribution $\pi_0$:

\begin{restatable}{lemma}{reducedconvergence}
  \label{lem:reduced-converge}
  Let $\{\mc C_i\}_{i=1}^c$ be the connected components of the undirected graph
  induced by $\Preduced$.
  Given initial state $\bX_0 \sim \pi_0$, \Cref{alg:reduced} converges to
  $\pi^+(\pi_0) \triangleq \lim_{t \to \infty} \pi_0
  \Preduced^t$ satisfying
  \begin{align*}
    \pi^+(\pi_0) \in \argmin_{\pi' : \pi'(\mc C_i) = \pi_0(\mc C_i) ~ \forall i
    \in [c]} \dTV(\pi, \pi').
  \end{align*}
  In other words, $\pi^+(\pi_0)$ is as close as possible to $\pi$ subject to the
  constraint that it puts the same probability mass in each connected component
  that $\pi_0$ does.
\end{restatable}

\subsubsection{The reduced chain as a truncation of the simple chain.}

\def\hplane{\mathcal{I}}
\def\lattice{\mathcal{L}}

We conclude this section with some intuition on the relationship between the
simple chain and the reduced chain. We can think of the simple chain as
operating on a lattice $\lattice$ over $\Zplus^n$, where each $\bx \in \mc Y$
represents a lattice node. The constraint $\bV \bx = \bc$ is an intersection of
hyperplanes $\hplane$ over $\R^n$, and the set of exact solutions $\mcX$ is the
intersection of $\lattice$ and $\hplane$. The simple chain corresponds to a
random walk over the lattice that never goes ``above'' $\hplane$. The parameter
$\gamma$ effectively biases this random walk away from the origin by penalizing
nodes exponentially in their distance from $\hplane$. In contrast, the reduced
chain only operates on $\mcX = \lattice \cap \hplane$. But because transitions
are defined by replacing $k$ elements at a time, two states $\bx, \bx' \in \mcX$
have a transition under $\Preduced$ if and only if there is a path between them
in $\lattice$ of length at most $2k$.

This intuition leads
to a slightly different way to modify the simple chain: instead of penalizing
inexact solutions with $\exp(-\gamma \|\bV \bx - \bc\|_1)$ as in
\Cref{eq:tpg-def}, we could truncate the lattice with a threshold penalty like
$\mathbbm{1}[\|\bV \bx - \bc\|_1 \le \omega]$ for some parameter $\omega \in
\Zplus$. This would prevent the Markov chain from ever transitioning to states
that are distance $>\omega$ from $\mc I$, effectively restricting the state
space of the simple chain to be \begin{align*}
  \mc Y_\omega \triangleq \{\bx : \bV \bx \preceq \bc \wedge \|\bV \bx - \bc\|_1 \le
  \omega \}.
\end{align*}
We can write the reduced chain by simply setting $\omega = 0$ in this
definition (i.e., $\mc Y_{\omega=0} = \mcX$). If $\mcX$ occupies sufficiently
large mass within $\mc Y_\omega$ for some $\omega > 0$, rejection sampling on
this chain might lead to a practical algorithm in our setting. Of course, as is
the case with the reduced chain, we would no longer be able to guarantee that
this modification to the simple chain remains irreducible. We defer such
investigations, as well as more complex methods to modify the state space, to future work.

\section{A Practical Algorithm Based on the Reduced Chain}
\label{sec:overall}

We conclude with a description of how we implement our algorithm in practice.
Because we cannot guarantee that the reduced chain is irreducible, we draw
inspiration from \Cref{lem:reduced-converge} to choose an initial distribution
$\pi_0$ over initial states to reduce total variation distance to the desired
distribution even when the reduced chain is not irreducible. Our algorithm has
three main hyperparameters: $k$, $t$, and $N$. As before, $k$ gives the number
of elements we remove and replace at each iteration of the reduced chain. Larger
values of $k$ increase connectivity at the expense of increased computational
cost. The number of MCMC iterations we run is $t$, where in our experiments
(\Cref{sec:reduced-empirical}) we found that $t$ on the order of $10^5$ seems
appropriate.

$N$ governs our initial state: instead of choosing a single initial
state, we use ILP to enumerate the best $N$ solutions according to a linear
approximation to $\pi$ (see \Cref{app:reduced-alg} for details). If a particular
block has $|\mcX_b| < N$, we will enumerate all of them, meaning we can simply
sample according to $\pi$ with no need for MCMC. If, however, $|\mcX_b| \ge N$,
we will enumerate some $\mc S \subseteq \mcX$ such that $|\mc S| = N$. We choose
our initial distribution to be $\pi_0 = (\pi \given \bx \in \mc S)$, meaning we
sample a random $\bx_0$ from $\mc S$ according to $\pi$ and use this as our
starting state for MCMC. For a formal description of this algorithm, see
\Cref{alg:overall} in \Cref{app:final-alg}. This choice of $\pi_0$ has two
benefits:
\begin{enumerate}
  \item If $\Mreduced$ is not irreducible, the stationary distribution we will
    converge to is $\pi^+(\pi_0)$ as defined in \Cref{lem:reduced-converge}.
  \item If $\Mreduced$ is irreducible, this choice of $\pi_0$ provides slightly
    better mixing time guarantees than starting with a single deterministic
    $\bx_0$:
\end{enumerate}

\begin{restatable}{lemma}{overall}
  \label{lem:overall}
  If $M_k$ is irreducible, then \Cref{alg:overall} produces an
  $\varepsilon$-approximate sample from $\pi$ after at most
  \begin{align*}
    \mixingtime(\varepsilon; P_k, \bX) \le \relaxationk \log \p{\frac{1}{2\varepsilon
    \sqrt{\pi(\mc S)}}}
  \end{align*}
  Markov chain iterations.
\end{restatable}
(Recall that if $|\mcX_b| < N$, then \Cref{alg:overall} does not use MCMC at
all.) Thus, the number of MCMC iterations needed reduces as $\mc S$ captures
more probability mass. We prove \Cref{lem:overall} in \Cref{app:final-alg}.

\section{Evaluating and Enforcing Representativeness}
\label{sec:evaluation}

We now return to our original goal: generating a dataset such that the
marginal distribution over households approximates the PUMS distribution. The
PUMS data are very granular, consisting of detailed information on each
individual in each household. In contrast, the synthetic microdata we produce
are much coarser. Consistent with the SF1 data release, we report the following
pieces of information for each household in our microdata:\footnote{We could
  instead provide more granular information like household status and
householder race. This would lead to an increase in total variation distance. We
choose to omit them here because downstream analyses we seek to support
(including running privacy-preserving algorithms) only use these attributes.}
\begin{itemize}
  \item Number of persons of each of the 7 race categories
  \item Number of Hispanic persons
  \item Number of adults
\end{itemize}
Let $\mc L \subset \Zplus^9$ be the set of all low-dimensional
household ``types'' defined this way. For each $i \in [n]$, let $\ell_i$ be the
low-dimensional projection of $\bv_i$ onto these features. Note that this means
for $i \ne j$, it is possible to have $\ell_i = \ell_j$ if $\bv_i$ and $\bv_j$
agree on these 9 attributes.
Our goal will be to compare the
distribution over types induced by our microdata to the distribution of types
found in the PUMS.

To do so, we introduce additional notation. Let $\mc B$ be the set of blocks in
a given state. For a block $b \in \mc B$, let $\pi_b$ be the distribution $\pi$
over $\mcX_b$, the set of solutions for block $b$,
as defined in \Cref{eq:piD-def}. Let $m_b$ be the number of
households in block $b$. Then, let $q_\ell$ be the expected frequency of households
of type $\ell$ in our microdata. For $\ell \in \mc L$,
\def\suml{\sum_{\stackrel{i \in [n]}{\ell_i = \ell}}}
\begin{equation}
  \label{eq:qt-def}
  q_\ell \triangleq \frac{1}{|\mc B|} \sum_{b \in \mc B} \frac{1}{m_b}
\suml \EE{\bX \sim \pi_b}{\bX[i]}.
\end{equation}
We can compare this directly with $p_\ell$, defined to be the frequency of
households of type $\ell$ in the PUMS:
\begin{align*}
  p_\ell \triangleq \suml \Pr_{\mc D}[\bv_i].
\end{align*}

Ideally, we would have $q_\ell = p_\ell$ for all $\ell \in \mc L$, meaning our methods
are unbiased with respect to overall household distribution. However, this will
in general not be the case for a variety of reasons. First, as detailed further
in \Cref{app:encoding}, data may be incomplete, meaning we cannot always sample
from $\mc X_b$ for every block $b$. (In such cases, we coarsen the constraints
we enforce to ensure we still get reasonable results.) More subtly, even with
perfect data, our induced distribution is not unbiased. Consider the following
example with three blocks:
\begin{example}[Empirical frequencies may be biased]
  \label{ex:bias}
  Block A contains one household with two Black persons (type 1); block B
  contains two households, each of which contains one white person and one Black
  person (type 2); and block C contains one household with two white persons
  (type 3).
\end{example}

Overall, the frequency of the household types 1, 2, and 3 are $1/4$, $1/2$, and
$1/4$ respectively. Given these frequencies, our approach will correctly
reconstruct blocks A and C, since there is a unique solution to each of these
blocks. Block A can only be reconstructed as a single household of type 1, and
block C requires a single household of type 3. However, block B admits multiple
possible solutions: it can can be reconstructed either as one household of type
1 and one household of type 3, or as two households of type 2. Under the
distribution $\pi$ as defined in \Cref{eq:piD-def}, these will each be sampled
with probabilities 1/3 and 2/3 respectively. As a result, despite the fact that
$(p_1, p_2, p_3) = (1/4, 1/2, 1/4)$, our induced frequencies will be $(q_1, q_2,
q_3) = (1/3, 1/3, 1/3)$. Thus, in general, $q_\ell \ne p_\ell$.

\subsection{Empirical differences in the household distribution.}

Given that our induced distribution will not be unbiased in general, we
empirically evaluate just how close it is. To do so, we use \Cref{alg:overall}
to sample a synthetic dataset given by $\{\bx_b\}_{b \in \mc B}$ across an
entire state. Note that there is an additional source of error here: we do not
in general approximately sample from $\pi_b$ because when we are unable to
completely enumerate $\mcX_b$ (which is precisely when we need the MCMC methods
developed earlier) and the reduced chain is not connected, we instead sample
from $\pi^+(\pi \given \bx \in \mc S)$ as defined in
\Cref{lem:reduced-converge}.

We will compare the PUMS distribution over types $p_\ell$ to the empirical
frequencies $\hat q_\ell$ over the household types, defined analogously to
\Cref{eq:qt-def}:
\begin{align*}
  \hat q_\ell \triangleq \frac{1}{|\mc B|} \sum_{b \in \mc B} \frac{1}{m_b} \suml
  \bx_b[i].
\end{align*}
In what follows, we only report statistics involving $\hat q$ based on a single
run of \Cref{alg:overall}. Due to the computational cost involved, we are unable
to provide confidence intervals. Let $\mc P$ be the distribution over $\mc L$
with probabilities given by $p_\ell$, and let $\empdist$ be the distribution with
probabilities $\hat q_\ell$. We can now evaluate the distance between our empirical
distribution $\mc P$ and $\empdist$ using the total variation distance $\dTV(\mc
P, \empdist) = \frac{1}{2} \sum_{\ell \in \mc L} |p_\ell - \hat q_\ell|$.

For Alabama, we find that $\dTV(\Dpums, \empdist) = \ALTVDUnadjustedAll$.
However, there is a source of error that is beyond our control and, when
accounted for, reduces this gap significantly: The PUMS distribution is
inconsistent with the distribution reported by SF1. For example, according to
SF1, the proportion of single-person households in Alabama is \ALSynSizeOne,
whereas according to the PUMS, it is \ALPUMSSizeOne. Because we are constrained
to exactly match the number of single-person households reported by SF1, we are
guaranteed to incur at least this error relative to the PUMS.

To account for
this, in our analysis, we compare against a modified distribution $\Drw$
derived by
reweighting $\mc P$ to to match the SF1 data on characteristics that are
completely determined by SF1.
By doing so, we attempt to control for
discrepancies between the PUMS
distribution and SF1 that are outside our control. Specifically, we reweight $\mc P$
to match SF1 on both overall household size (up to size 7) and
ethnicity among single-person households (where the proportion of such
households is exactly determined by SF1 table \texttt{P28H}), producing
$\Drw$.\footnote{Note that this does not affect our original data
  generation process since the constraints exactly specify the number of
households of each size.}
For example, for a household $\ell \in \mc L$ containing two individuals, we define
\begin{align*}
  \Pr_{\ell \sim \Drw}[\ell] \triangleq \Pr_{\ell \sim \mc P}[\ell] \cdot
  \frac{\Pr_{\ell' \sim
    \empdist}[\ell' \text{ contains two individuals}]}{\Pr_{\ell' \sim \mc
    P}[\ell'
  \text{ contains two individuals}]}.
\end{align*}
By construction, $\Drw$ will match the SF1 distribution on our chosen set of
characteristics $\mc C$. In our case $\mc C \triangleq \{1H, 1N, 2, 3, 4, 5, 6,
7+\}$ where $1H$ and $1N$ represent single-person households containing a
Hispanic individual and a non-Hispanic individual respectively. These
characteristics form a partition of the set of all households, and SF1
completely determines $\Pr_{\ell \sim \empdist}[\ell \text{ is of type } c]$ for all
$c \in \mc C$.\footnote{For a small fraction of blocks with incomplete
information, this is not strictly true. See \Cref{app:encoding} for details.}
Thus, for any $c \in \mc C$,
\begin{align*}
  \Pr_{\ell \sim \Drw}[\ell \text{ is of type } c]
  &= \sum_{\ell : \ell \text{ is of type } c} \Pr_{\ell \sim \Drw}[\ell] \\
  &= \sum_{\ell : \ell \text{ is of type } c} \Pr_{\ell \sim \mc P}[\ell] \cdot
  \frac{\Pr_{\ell' \sim \empdist}[\ell' \text{ is of type } c]}{\Pr_{\ell' \sim
    \mc P}[\ell'
  \text{ is of type } c]} \\
  &= \frac{\Pr_{\ell' \sim \empdist}[\ell' \text{ is of type } c]}{\Pr_{\ell' \sim \mc
  P}[\ell' \text{ is of type } c]} \sum_{\ell : \ell \text{ is of type } c}
  \Pr_{\ell \sim
  \mc P}[\ell] \\
  &= \frac{\Pr_{\ell' \sim \empdist}[\ell' \text{ is of type } c]}{\Pr_{\ell' \sim \mc
  P}[\ell' \text{ is of type } c]} \cdot \Pr_{\ell \sim \mc P}[\ell \text{ is of type }
  c] \\
  &= \Pr_{\ell' \sim \empdist}[\ell' \text{ is of type } c].
\end{align*}
Again, because SF1 completely determines $\Pr_{\ell \sim \empdist}[\ell \text{ is of
type } c]$, adjusting $\mc P$ in this way effectively accounts for the fact that
our two data sources do not exactly agree---that is, $\Pr_{\ell \sim \mc P}[\ell
\text{ is of type } c]$ is inconsistent with SF1.

Under this adjustment, again in Alabama, $\dTV(\Drw, \empdist) =
\ALTVDAdjustedAll$. Finally, for data incompleteness reasons detailed in
\Cref{app:encoding}, $\mc X_b$ is empty for some blocks, and we can only solve a
coarsened version of our original problem. If we only consider the distribution
$\empdistrw$ that excludes these blocks, we find $\dTV(\Drw, \empdistrw) =
\ALTVDAdjustedAccOne$.

It's not immediately obvious how to interpret these TVD figures. On the one
hand, we're only measuring TVD at the state level; TVD at the block level would
necessarily be larger, though we cannot measure it without access to complete
microdata. On the other hand, we know that the PUMS sample is biased, meaning
there is some irreducible TVD between it and any microdata consistent with SF1.
While we have partially accounted for this by adjusting the PUMS on household
size (going from $\empdist$ to $\empdistrw$), there may be irreducible TVD along
other dimensions. (For example, if the expected number of persons of a
particular race given by the PUMS is different from the number of persons of
that race reported by SF1 in the state, this would be another source of
irreducible TVD.) While we cannot conclusively resolve the sources of error at
play, we will find some evidence below that at least some of this error is due
to our methods, not irreducible TVD.

A similar trend holds for NV, summarized in \Cref{tab:tvd} (labeled ``no
reweighting''), though the total variation distance is substantially larger.
This is consistent with a broader trend: Nevada requires more MCMC iterations
per block, requires MCMC more often (i.e., $|\mcX_b|$ is more often greater than
$N$; see \Cref{fig:heavy-tail}), and requires more computation time overall
relative to Alabama. We speculate that this is because Nevada tends to have
substantially more people per block (an average of 75.8, as opposed to 35.3 for
Alabama) as well as a larger support for the household distribution $\mc D$
(i.e., more distinct household types). Replicating our work across more states
will likely lead to more insight here.

\begin{table}[htpb]
  \centering
  \begin{tabular}{c c c c c}
    \toprule
    State & $\lambda$ & $\dTV(\Dpums, \empdist)$ & $\dTV(\Drw, \empdist)$ & $\dTV(\Drw,
    \empdistrw)$ \\
    \midrule
    AL & no reweighting & \ALTVDUnadjustedAll & \ALTVDAdjustedAll &
    \ALTVDAdjustedAccOne \\
       & $10^{-3}$ & \ALlzzoTVDUnadjustedAll & \ALlzzoTVDAdjustedAll &
       \ALlzzoTVDAdjustedAccOne \\
       & $10^{-5}$ & \ALlzzzzoTVDUnadjustedAll & \ALlzzzzoTVDAdjustedAll &
       \ALlzzzzoTVDAdjustedAccOne \\
    \midrule
    NV & no reweighting & \NVTVDUnadjustedAll & \NVTVDAdjustedAll &
    \NVTVDAdjustedAccOne \\
       & $10^{-3}$ & \NVlzzoTVDUnadjustedAll & \NVlzzoTVDAdjustedAll &
       \NVlzzoTVDAdjustedAccOne \\
       & $10^{-5}$ & \NVlzzzzoTVDUnadjustedAll & \NVlzzzzoTVDAdjustedAll &
       \NVlzzzzoTVDAdjustedAccOne \\
    \bottomrule
  \end{tabular}
  \caption{Total variation distance between empirical distribution and statewide
  PUMS.}
  \label{tab:tvd}
\end{table}

\subsection{Improving representativeness via reweighting.}

When the PUMS distribution differs substantially from our empirical frequencies,
we can adjust the inputs to our algorithm to try to close the gap.
\Cref{alg:overall} depends on the input distribution $\mc D$, and we can
\textit{reweight} $\mc D$ (drawing some inspiration from \citet{wang2024post})
to produce a new distribution $\mc D_\lambda$ such that, when we run
\Cref{alg:overall}, the resulting empirical frequencies $\empdist$ are closer to
the original PUMS distribution $\mc P$. We provide a heuristic in what follows,
showing that our reweighting scheme reduces total variation distance to $\Drw$.

Suppose we run \Cref{alg:overall} once, resulting in empirical household type
frequencies $\hat q_\ell$. Consider the distribution $\mc D_\lambda$ defined as
follows for some smoothing parameter $\lambda$:
\begin{align*}
  \Pr_{\mc D_\lambda}[\bv_i] \propto \Pr_{\mc D}[\bv_i] \cdot \p{\frac{p_{\ell_i} +
  \lambda}{\hat q_{\ell_i} + \lambda}}.
\end{align*}
We might expect that if we run \Cref{alg:overall} on $\mc D_\lambda$ instead of
$\mc D$, we will end up with a lower total variation distance in our household
type distribution. Intuitively, this is because $\mc D_\lambda$ attempts to
``correct'' for the bias induced by \Cref{alg:overall}, boosting the frequency
of households that appear infrequently in $\empdist$ relative to $\Dpums$ and
reducing the frequencies of households that are overrepresented. \Cref{tab:tvd}
shows that reweighting $\mc D$ decreases total variation distance, though it
does not bring it down to 0. It is possible that additional adjustments could
reduce TVD further.

Given \Cref{ex:bias}, we might wonder whether multiracial households are
likely to be underrepresented in our data. This is indeed the case in both
states we analyze, where the frequency of multiracial households drops from
$\ALPUMSMRPct\%$ (AL PUMS) to $\ALSynMRPct\%$  (AL synthetic data) and from
$\NVPUMSMRPct\%$ (NV PUMS) to $\NVSynMRPct\%$ (NV synthetic data). At least
facially, this can be remedied by artificially boosting the frequencies of
multiracial households in $\mc D_\lambda$. As a crude example, if we let
$\lambda=10^{-5}$ and double the frequency of each multiracial household in $\mc
D_\lambda$ (before normalizing it), we raise the frequency of multiracial
households in $\empdist$ to $\NVmrtwoSynMRPct\%$ in Nevada, much closer to the
figure of $\NVPUMSMRPct\%$ found in the NV PUMS. Again, further work is needed
to determine whether this has broader effects on the realism of our synthetic
microdata. Practitioners interested in a particular aspect of the distribution
(e.g.,~multiracial households or single-parent households) should explicitly
verify that the distribution of those characteristics in the synthetic microdata
matches the PUMS distribution.

\section{Discussion and Limitations}
\label{sec:conclusion}

We have presented algorithms to generate synthetic data and demonstrated that
these algorithms are efficient in the Census context. While these data ``look
like'' microdata, it is important to recognize that no sample from this
distribution should be treated as ground truth; analyses should look for
consistent findings over multiple samples from this distribution. The
distribution chosen here is principled, but future work could test whether
alternative distributions produce more realistic data, and whether this paradigm
of synthetic data generation via combinatorial optimization is appropriate in
contexts beyond the Census.

Finally, we comment on potential uses for our synthetic microdata. One clear use
case is researchers who seek to \textbf{study the properties of disclosure avoidance
systems}. For example, a researcher could generate synthetic datasets, run a
disclosure avoidance method on it (e.g.,~TopDown \citep{abowd2018us}), and
perform fine-grained analyses of its privacy and accuracy properties. A growing
line of work seeks to do just this
\citep[e.g.,][]{cohen2022census,christ2022differential,kenny2021impact,bailie2023can},
and the methods presented here will enable a more systematic approach. Moreover,
scientists who depend on Census data may be interested in \textbf{biases induced by
disclosure avoidance algorithms}
\citep[e.g.~][]{mueller20222020,ruggles2019differential,hauer2021differential,santos2020differential,winkler2021differential};
a potential use for our methods is to estimate and perhaps correct for these
biases. If a researcher is interested in a particular statistic measured using
Census data, they could estimate the effect that a disclosure avoidance tool
will have on that statistic over the distribution of synthetic data we generate.
As noted in \Cref{sec:introduction}, we suggest that all analyses be
performed across multiple samples of synthetic data. Used appropriately, the
methods presented here can be a valuable tool for researchers to understand the
properties of privacy-preserving algorithms.

Our methods could easily be extended to other population synthesis settings,
potentially with more complex population models \citep[see,
e.g.,~][]{chapuis2022generation}. This might introduce additional complexity:
Unlike our PUMS-derived distribution, population models in general can have
exponentially large support (recall that $n$ is the size of the support of $\mc
D$). If, for example, every possible household exists with nonzero probability,
writing down the matrix $\bV$ (let alone solving an integer program with it) can
be computationally infeasible. To get around this, one could take $K$ samples
from the population distribution $\mc D$ to get an empirical distribution
$\hat{\mc D}_K$ with bounded support. We could even make this more efficient by
ruling out samples that are incompatible with the SF1 counts $\bc$. (For
example, if a block has no 5-person households, we would not need to include
5-person households when sampling $\hat{\mc D}_K$.) Doing so independently for
each block would still yield an exact match to the aggregate counts, though we
may need to make $K$ sufficiently large to ensure that a valid solution exists.

Our work has important limitations, both technical and conceptual. Technically,
our methods are not guaranteed to converge to the desired stationary
distribution, a limitation we are unlikely to overcome due to the NP-hardness of
the underlying sampling problem. As noted in \Cref{ex:bias} our choice of
distribution $\pi$ does not lead to an unbiased set of households. The
heuristics we present in \Cref{sec:evaluation} reduce but do not eliminate this
bias, and researchers should take care that synthetic data match the PUMS data
on characteristics important to their analyses. We are also limited by the fact
that the PUMS is only a sample of households, meaning low-frequency household
types may never appear.

Future work could extend our methods in a number of different ways. Our
particular problem formulation sought to match particular statistics on race,
ethnicity, gender, and household size. SF1 data contain far more information
than these, however, and in theory one could seek to match other characteristics
as well. For instance, researchers interested in studying older populations
may wish to match statistics from table \texttt{P25} (``Households by Presence
of People 65 Years and Over, Household Size, and Household Type''). Of course,
the more constraints we add, the larger the dimension $d$ of the problem. This
could increase the computational cost of finding solutions or overconstrain the
problem, making it impossible to solve exactly. A better understanding of this
trade-off is needed to take advantage of all of the data provided by SF1. Beyond
the Decennial Census, research and public policy often rely on other Census
products like the American Community Survey. Extending our methods to these
other data releases may be a fruitful avenue for future research. Finally, our
formulation treats the state-wide household distribution as invariant
conditional on the block-level statistics. If, however, the researcher has more
information about geographic variation that the block-level statistics do not
encode (for example, that a particular region has more multiracial households
than another), they could in principle encode that information by using a
different distribution over households $\mc D$ for each block.

Ultimately, there is plenty of room to customize the methods we have presented
here for any particular use case. Our results have shown that it is possible to
efficiently generate samples from a plausible distribution over synthetic Census
microdata. Our hope is that researchers will be able to tailor our approach
to fit their needs, enabling more comprehensive research into privacy-preserving
methods and their broader impacts on downstream analyses.

\paragraph*{Acknowledgements.}
The authors thank Parikshit Gopalan for helpful conversations. This work was
supported in part by a Google Research Scholar award and the Center for Research
on Computation and Society (CRCS), Harvard.

\bibliographystyle{abbrvnat}
\bibliography{refs}

\clearpage
\appendix
\paragraph*{Organization of the appendix.}

\Cref{app:encoding} contains detailed information on our encoding of Census data
into $d$-dimensional vectors. We present a rejection sampling algorithm that is
equivalent to our algorithm under our choice of $\pi$ in \Cref{app:rejection}.
In \Cref{app:simple-omitted,app:reduced-omitted}, we provide omitted proofs and
details from \Cref{sec:simple,sec:reduced} respectively. We present additional
empirical results for both the simple and reduced chains in
\Cref{app:additional-empirical}. We formally show that \probname\ is NP-hard in
\Cref{app:hardness}. In \Cref{app:bad-examples}, we construct pathological
synthetic examples under which the reduced chain is either not irreducible or
has exponentially large mixing time. Finally, in \Cref{app:final-alg}, we
describe our final algorithm and theoretically characterize its performance.

\section{Data Encoding Details}
\label{app:encoding}

Census data are encoded with 7 race categories (white, Black, American Indian and
Native Alaskan, Asian, Hawaiian and Pacific Islander, other, and two or more)
and two ethnicity categories (Hispanic or non-Hispanic). We encode each distinct
household $\bv_i$ and block-level statistic $\bc$ as a 135-dimensional integer
vector. The encoding is as follows, annotated with the dimensionality of each
component:
\begin{itemize}
  \item Number of individuals of each (race, ethnicity)
    combination \textbf{(14)}
  \item Number of adults (defined as 18 or older) of
    each race \textbf{(7)}
  \item Number of Hispanic adults \textbf{(1)} \item
    (Householder race) $\times$ (whether or not the household is a family)
    $\times$ (household size in $\{1, 2, 3, 4, 5, 6, 7+\}$) \textbf{(98)}. Note
    that these form a partition, so every household has a 1 in exactly one of
    these 98 components.
  \item For households with a Hispanic householder, (whether or not the
    household is a family) $\times$ (household size in $\{1, 2, 3, 4, 5, 6,
    7+\}$) \textbf{(14)}\footnote{Because ethnicity is encoded as a binary
      variable and the remaining two attributes are completely determined by the
    (race) $\times$ (is family) $\times$ (household size) variables, including
  both (Householder is Hispanic) and (Householder is not Hispanic) is redundant
here.}
  \item Number of households (a constant 1 for each household) \textbf{(1)}
\end{itemize}
Each row of the PUMS sample contains all of this information. We
choose this particular encoding because this information can also be found in
block-level aggregate data from SF1.\footnote{SF1 data can be accessed at
\url{https://www.nhgis.org/}.} Each block-level vector $\bc$ is the sum of the
these characteristics of all households in the block. These totals can be found
in a combination of SF1 tables \texttt{P3}, \texttt{P28A}--\texttt{H}, and
\texttt{P16A}--\texttt{H}.
Some of these tables do not account for \textit{group quarters} (as opposed to
households), which exist in around 1\% of blocks in both AL and NV.
We exclude these blocks from our analysis. In principle, our methods would also
apply when reconstructing group quarters. Moreover, there are blocks for which
no solution to our optimization problem exists (i.e., no combination of PUMS
households matches the reported block totals). This can happen because the PUMS
data are incomplete, since a ``rare'' household might be needed to match the
block-level constraints. We exclude these blocks as well (about 5\% in AL and
14\% in NV). In our publicly available code, when we encounter such blocks, we
simply reduce the dimensionality of the vectors and apply the same techniques to
solve them. For example, instead of matching the number of adults of each race
and ethnicity, we simply match the total number of adults in the block.

There are $n = 2,745$ distinct households (encoded as 135-dimensional vectors)
in the AL PUMS and $n = 4,094$ distinct households in the NV PUMS. The data we
use are subject to disclosure avoidance systems before their release. They do
not represent ``ground truth''; nevertheless, we treat them as the source of the
distribution we seek to sample from. Downstream analyses should be careful not
to rely on the exact locations of very rare households, since these are unlikely
to be accurate.

\section{Rejection Sampling Formulation}
\label{app:rejection}

Here, we present a rejection sampling algorithm (\Cref{alg:rejection}) that
samples from $\mcX$ according to the distribution given in \cref{eq:piD-def}.
This algorithm is far too inefficient to be practical; it is intended simply to
illustrate and provide motivation for the problem specification here.

\algrenewcommand\algorithmiccomment[1]{\hfill $\triangleright$ #1}

\begin{algorithm}
  \caption{\textproc{RejectionSampling}$(\bV, \bc, \mc D)$}
  \label{alg:rejection}
  \begin{algorithmic}
    \Repeat
    \Comment{Repeat until a sample is accepted}
    \State{$\bx \gets \mathbf{0} \in \Zplus^{n}$}
    \Comment{Empty multiset}
    \While{$\bV \bx \preceq \bc$}
    \Comment{$\preceq$ denotes element-wise $\le$.}
    \State{Sample $\bv_i$ from $\mc D$}
    \State{$\bx[i] \gets \bx[i] + 1$}
    \EndWhile
    \Until{$\bV \bx = \bc$}
    \State{\textbf{return} $\bx$}
  \end{algorithmic}
\end{algorithm}

\section{Omitted Proofs and Details from \Cref{sec:simple}}
\label{app:simple-omitted}

\subsection{Formal algorithm description for the simple chain}
\label{app:simple-alg}

\begin{algorithm}
  \caption{\textproc{Simple}($\bV, \bc, f, \gamma, t$)}
  \label{alg:gibbs}
  \begin{algorithmic}
    \Repeat
    \State $\bx \gets \mathbf{0} \in \Zplus^n$
    \For{$t$ iterations}
    \WithProbability{$0.5$}
    \State \textbf{continue}
    \EndWithProbability
    \State $i \gets$ a random integer in $[n]$ such that $\bv_i
    \preceq \bc$
    \Comment{Ignore ineligible households}
    \State $\br \gets \bc - \bV \bx_{i \gets 0}$
    \State $g_{\max} \gets \max~\{g \in \Zplus : g \bv_i \preceq \br\}$
    \State $\mc S \gets \{\bx_{i \gets g} : 0 \le g \le g_{\max}\}$
    \State $\bx \gets$ a random sample from $\tpg \given \bx \in \mc S$ (i.e.,
    $\propto f(\bx) \exp(-\gamma \|\bV \bx - \bc\|)$ on $\mc S$)
    \EndFor
    \Until{$\bV \bx = \bc$}
    \State \Return $\bx$
  \end{algorithmic}
\end{algorithm}

\subsection{Theoretical results for the simple chain}
\label{app:simple-theory}

First, we prove \Cref{lem:eps-sol-density}. We restate each result before
proving it.
\epssoldensity*

\begin{proof}
  First, we show that bounded TVD implies bounded TVD on conditional
  distributions, inflated by a factor of $1/\sigma(\mcX)$.
  \begin{align*}
    \dTV(\sigma, \sigma') &\le \varepsilon \\
    \|\sigma-\sigma'\|_1 &\le 2\varepsilon \\
    \sum_{\bx \in \mcX} |\sigma(\bx) - \sigma'(\bx)| &\le 2\varepsilon \\
    \sum_{\bx \in \mcX} \left|\frac{\sigma(\bx)}{\sigma(\mcX)} -
    \frac{\sigma'(\bx)}{\sigma(\mcX)}\right| &\le \frac{2\varepsilon}{\sigma(\mcX)} \\
    \sum_{\bx \in \mcX} \left|\frac{\sigma(\bx)}{\sigma(\mcX)} -
    \frac{\sigma'(\bx)}{\sigma'(\mcX)} + \frac{\sigma'(\bx)}{\sigma'(\mcX)} -
    \frac{\sigma'(\bx)}{\sigma(\mcX)}\right| &\le \frac{2\varepsilon}{\sigma(\mcX)} \\
    \sum_{\bx \in \mcX} \left|\frac{\sigma(\bx)}{\sigma(\mcX)} -
    \frac{\sigma'(\bx)}{\sigma'(\mcX)}\right| - \left|\frac{\sigma'(\bx)}{\sigma'(\mcX)} -
    \frac{\sigma'(\bx)}{\sigma(\mcX)}\right| &\le \frac{2\varepsilon}{\sigma(\mcX)}
    \tag{triangle inequality} \\
    \sum_{\bx \in \mcX} \left|\sigma_{\mcX}(\bx) - \sigma_{\mcX}'(\bx)\right| -
    \sigma'(\bx)\left|\frac{1}{\sigma'(\mcX)} - \frac{1}{\sigma(\mcX)}\right| &\le
    \frac{2\varepsilon}{\sigma(\mcX)} \\
    \|\sigma_{\mcX} - \sigma_{\mcX}'\|_1 - \sigma'(\mcX)\left|\frac{1}{\sigma'(\mcX)} -
    \frac{1}{\sigma(\mcX)}\right| &\le \frac{2\varepsilon}{\sigma(\mcX)} \\
    \|\sigma_{\mcX} - \sigma_{\mcX}'\|_1 &\le \frac{2\varepsilon}{\sigma(\mcX)} +
    \sigma'(\mcX) \frac{|\sigma(\mcX) - \sigma'(\mcX)|}{\sigma'(\mcX) \sigma(\mcX)}  \\
    \|\sigma_{\mcX} - \sigma_{\mcX}'\|_1 &\le \frac{2\varepsilon}{\sigma(\mcX)} +
    \frac{|\sigma(\mcX) - \sigma'(\mcX)|}{\sigma(\mcX)} \\
    \|\sigma_{\mcX} - \sigma_{\mcX}'\|_1 &\le \frac{3\varepsilon}{\sigma(\mcX)}
    \tag{$\dTV(\sigma, \sigma') \le \varepsilon$} \\
    \dTV(\sigma_{\mcX}, \sigma_{\mcX}') &\le \frac{3\varepsilon}{2\sigma(\mcX)}
  \end{align*}
  To show that this is tight to within a constant factor, consider an instance
  where $\mcX = \{\bx_0, \bx_1\}$ and assume without loss of generality
  $\sigma(\bx_0) \le \sigma(\bx_1)$. For a given $\varepsilon$ define $\sigma'$
  as follows:
  \begin{align*}
    \sigma'(\bx_0) &= \sigma(\bx_0) + \varepsilon \\
    \sigma'(\bx_1) &= \max(0, \sigma(\bx_1) - \varepsilon) \\
    \sigma'(\bx) &= \begin{cases}
      \sigma(\bx) & \sigma(\bx_1) \ge \varepsilon \\
      \sigma(\bx) \p{1 - \frac{\varepsilon - \sigma(\bx_1))}{1 - \sigma(\mcX)}} & \sigma(\bx_1) < \varepsilon
    \end{cases}
    \tag{$\bx \notin \mcX$}
  \end{align*}
  First, observe that if $\sigma(\bx_1) < \varepsilon$, then $\sigma'(\bx_1) = 0$ and
  $\dTV(\sigma_{\mcX}, \sigma'_{\mcX}) \ge \frac{1}{2}$. This is within a factor of 2
  of our upper bound since trivially $\dTV(\sigma_{\mcX}, \sigma'_{\mcX}) \le 1$. If
  $\sigma(\bx_1) \ge \varepsilon$, then
  \begin{align*}
    \dTV(\sigma_{\mcX}, \sigma'_{\mcX})
    &= \frac{1}{2} \|\sigma_{\mcX} - \sigma'_{\mcX}\|_1 \\
    &= \frac{1}{2} \sum_{\bx \in \mcX} |\sigma_{\mcX}(\bx) - \sigma'_{\mcX}(\bx)| \\
    &= \frac{1}{2} \sum_{\bx \in \mcX} \left|\frac{\sigma(\bx)}{\sigma(\mcX)} -
    \frac{\sigma'(\bx)}{\sigma'(\mcX)} \right| \\
    &= \frac{1}{2} \sum_{\bx \in \mcX} \left|\frac{\sigma(\bx)}{\sigma(\mcX)} -
    \frac{\sigma'(\bx)}{\sigma(\mcX)} \right| \tag{$\sigma'(\mcX) = \sigma(\mcX)$} \\
    &= \frac{1}{2\sigma(\mcX)} \sum_{\bx \in \mcX} \left| \sigma(\bx) - \sigma'(\bx) \right| \\
    &= \frac{1}{2\sigma(\mcX)} (|\sigma(\bx_0) - \sigma'(\bx_0)| + |\sigma(\bx_1) -
    \sigma'(\bx_1)|) \\
    &= \frac{\varepsilon}{\sigma(\mcX)}.
  \end{align*}
\end{proof}

We next provide a theoretical explanation for the poor performance of the simple
chain here, which we will use to motivate the reduced chain in
\Cref{sec:reduced}. In particular, we characterize how $\gamma$ impacts
$\numsampgamma$. We show how varying $\gamma$ creates a trade-off between raising
$\solutiondensitystar$ and lowering $\mixinggammalb$. First, we show that
$\solutiondensitystar$ increases with $\gamma$, ranging from 0 to 1 as $\gamma$
ranges from $-\infty$ to $\infty$.

\soldensity*
\begin{proof}
  First, we show that $\solutiondensitystar$ is monotonically increasing in
  $\gamma$. By definition,
  \begin{align*}
    \frac{d}{d\gamma} \solutiondensitystar
    &= \frac{d}{d\gamma} \frac{\tpg(\mcX)}{\tpg(\mc Y)} \\
    &= \frac{d}{d\gamma} \frac{\sum_{\bx \in \mcX}f(\bx) \exp(-\gamma \|\bV \bx
    - \bc\|_1)}{\sum_{\bx \in \mc Y} f(\bx) \exp(-\gamma \|\bV \bx - \bc\|_1)} \\
    &= \frac{d}{d\gamma} \frac{\sum_{\bx \in \mcX}f(\bx)}{\sum_{\bx \in
    \mc Y}f(\bx) \exp(-\gamma \|\bV \bx - \bc\|_1)}
    \tag{$\bV \bx - \bc = 0$ for $x \in \mcX$} \\
    &= \p{\sum_{\bx \in \mcX} f(\bx)} \frac{d}{d\gamma} \frac{1}{\sum_{\bx \in
    \mc Y}f(\bx) \exp(-\gamma \|\bV \bx - \bc\|_1)}.
  \end{align*}
  It suffices to show
  \begin{equation*}
    \sum_{\bx \in \mc Y} \frac{d}{d\gamma} f(\bx) \exp(-\gamma \|\bc - \bV
    \bx\|_1)
    = \sum_{\bx \in \mc Y} -\|\bc - \bV \bx\|_1 f(\bx) \exp(-\gamma \|\bc - \bV
    \bx\|_1)
    < 0.
  \end{equation*}
  Next, we show that $\solutiondensitystar$ has the claimed limits.
  For any $\bx' \in \mc Y \backslash \mcX$ and
  $\gamma < 0$,
  \begin{align*}
    \exp(-\gamma \|\bc - \bV \bx'\|_1) f(\bx')
    \ge \exp(-\gamma) f(\bx').
  \end{align*}
  Therefore,
  \begin{align*}
    \lim_{\gamma \to -\infty} \solutiondensitystar
    &= \lim_{\gamma \to -\infty} \frac{\tpg(\mcX)}{\tpg(\mc Y)} \\
    &= \lim_{\gamma \to -\infty} \frac{\sum_{\bx \in \mcX} f(\bx) \exp(-\gamma
    \|\bc - \bV \bx\|_1)}{\sum_{\bx \in \mc Y} f(\bx) \exp(-\gamma \|\bc - \bV
    \bx\|_1)} \\
    &\le \lim_{\gamma \to -\infty} \frac{\sum_{\bx \in \mcX} f(\bx) \exp(-\gamma
      \|\bc - \bV \bx\|_1)}{\sum_{\bx \in \mc Y \backslash \mcX} f(\bx)
      \exp(-\gamma \|\bc - \bV \bx\|_1)} \\
    &\le \lim_{\gamma \to -\infty} \frac{\sum_{\bx \in \mcX}
    f(\bx)}{\exp(-\gamma) \sum_{\bx' \in \mc Y \backslash \mcX} f(\bx')} \\
    &= 0.
  \end{align*}
  Similarly, for $\bx' \in \mc Y \backslash \mcX$ and $\gamma > 0$,
  \begin{align*}
    \exp(-\gamma \|\bc - \bV \bx'\|_1) f(\bx')
    \le \exp(-\gamma) f(\bx').
  \end{align*}
  Therefore,
  \begin{align*}
    \lim_{\gamma \to \infty} \solutiondensitystar
    &= \lim_{\gamma \to \infty} \frac{\tpg(\mcX)}{\tpg(\mc Y)} \\
    &= \lim_{\gamma \to \infty} \frac{\sum_{\bx \in \mcX} f(\bx) \exp(-\gamma
    \|\bc - \bV \bx\|_1)}{\sum_{\bx \in \mc Y} f(\bx) \exp(-\gamma \|\bc - \bV
    \bx\|_1)} \\
    &= \lim_{\gamma \to \infty} \frac{\sum_{\bx \in \mcX} f(\bx)}{\sum_{\bx
      \in \mc Y \backslash \mcX} f(\bx) \exp(-\gamma \|\bc - \bV \bx\|_1) +
    \sum_{\bx \in \mcX} f(\bx)} \\
    &\ge \lim_{\gamma \to \infty} \frac{\sum_{\bx \in \mcX}
      f(\bx)}{\exp(-\gamma) \sum_{\bx \in \mc Y \backslash \mcX} f(\bx) +
      \sum_{\bx \in \mcX} f(\bx)} \\
    &= 1.
  \end{align*}
\end{proof}

\expmixingtime*
\begin{proof}
  We will provide a lower an upper bound on the conductance of $\Mgamma$ and use
  Cheeger's inequality to lower-bound the mixing time. The conductance
  \citep[see, e.g.,][]{guruswami2016rapidly} of $\Mgamma$ is defined as
  \begin{equation}
    \Phi(\Mgamma) = \min_{\substack{\mc S \subset \mc Y\\ 0 < \tpg(\mc S)
    \le 1/2}} \frac{\Qgamma(\mc S, \overline{\mc S})}{\tpg(\mc S)},
  \end{equation}
  where $\Qgamma(\bx, \bx') = \tpg(\bx) \Pgamma(\bx, \bx')$ and $\Qgamma(\mc S,
  \overline{\mc S}) = \sum_{\bx \in \mc S, \bx' \notin \mc S} \Qgamma(\bx,
  \bx')$. We will show that in particular,
  $\Qgamma(\mcX, \overline{\mcX})/\tpg(\mcX)$ is small, providing an upper
  bound for $\Phi(\Mgamma)$.

  To do so, let $p_i = \Pr_{\mc D}[\bv_i]$ and $\ell$ = $\|\bv_i\|_1$. Observe
  that for any $\bx \in \mcX$ and $i \in [n]$, if $\gamma$ is sufficiently
  large,
  \begin{align*}
    \frac{\tpg(\bx)}{\tpg(\Delta(\bx, i))}
    &= \frac{\tpg(\bx)}{\sum_{g = 0}^{\bx[i]} \tpg(\bx_{i \gets g})} \\
    &= \frac{f(\bx)}{\sum_{g = 0}^{\bx[i]} f(\bx_{i \gets g}) \exp(-\gamma\|\bc
    - \bV \bx_{i \gets g}\|_1)} \\
    &= \frac{f(\bx)}{\sum_{g = 0}^{\bx[i]} f(\bx_{i \gets (\bx[i]-g)})
    \exp(-\gamma\|\bc - \bV \bx_{i \gets (\bx[i]-g)}\|_1)} \\
    &= \frac{f(\bx)}{\sum_{g = 0}^{\bx[i]} f(\bx_{i \gets (\bx[i]-g)})
    \exp(-\gamma g \|\bv_i\|_1)} \\
    &= \frac{1}{1 + \sum_{g = 1}^{\bx[i]} \frac{f(\bx_{i \gets
    (\bx[i]-g)})}{f(\bx)} \exp(-\gamma g \ell)} \\
    &\ge 1 - \sum_{g = 1}^{\bx[i]} \frac{f(\bx_{i \gets
    (\bx[i]-g)})}{f(\bx)} \exp(-\gamma g \ell) \tag{$\frac{1}{1+a} \ge 1-a$ for
    $a \ge 0$} \\
    &= 1 - \sum_{g = 1}^{\bx[i]} \frac{\binom{k -
    g}{\bx[i]-g}}{\binom{k}{\bx[i]} p_i^{g}} \exp(-\gamma g \ell) \\
    &\ge 1 - \sum_{g = 1}^{\bx[i]} p_i^{-g}
    \exp(-\gamma g \ell) \\
    &= 1 - \sum_{g = 1}^{\bx[i]} (p_i \exp(\gamma \ell))^{-g} \\
    &\ge 1 - \sum_{g = 1}^{\infty} (p_i \exp(\gamma \ell))^{-g} \\
    &= 1 - \frac{p_i^{-1} \exp(-\gamma \ell)}{1 - p_i^{-1} \exp(-\gamma \ell)}
    \tag{for $p_i^{-1} \exp(-\gamma \ell) < 1$} \\
    &\ge 1 - 2(\min_i p_i)^{-1} \exp(-\gamma \ell)
    \tag{for $p_i^{-1} \exp(-\gamma \ell) < 1/2$} \\
    &= 1 - R \exp(-\gamma \ell)
    \numberthis \label{eq:tpg-ratio-lb}
  \end{align*}
  for $R \triangleq 2(\min_i p_i)^{-1}$. With this,
  \begin{align*}
    \frac{\Qgamma(\mcX, \overline{\mcX})}{\tpg(\mcX)}
    &= \frac{\sum_{\bx \in \mcX, \bx' \in \overline{\mcX}_{\bx}} \tpg(\bx)
    \Pgamma(\bx, \bx')}{\tpg(\mcX)} \\
    &= \frac{\sum_{\bx \in \mcX} \tpg(\bx) \sum_{\bx' \in
    \overline{\mcX}_{\bx}} \Pgamma(\bx, \bx')}{\tpg(\mcX)} \\
    &= \frac{\sum_{\bx \in \mcX} \tpg(\bx) \sum_{i \in [n]} \sum_{\bx' \in
      \overline{\mcX}_{\bx} \cap \Delta(\bx, i)} \frac{\tpg(\bx')}{2n
    \tpg(\Delta(\bx, i))} }{\tpg(\mcX)} \\
    &= \frac{\sum_{\bx \in \mcX} \tpg(\bx) \sum_{i \in [n]}
      \frac{\tpg(\Delta(\bx, i)) - \tpg(\bx)}{2n \tpg(\Delta(\bx, i))}
    }{\tpg(\mcX)} \\
    &= \frac{\sum_{\bx \in \mcX} \tpg(\bx) \frac{1}{2n} \sum_{i \in [n]} \p{1 -
    \frac{\tpg(\bx)}{\tpg(\Delta(\bx, i))}}}{\tpg(\mcX)} \\
    &\le \frac{\sum_{\bx \in \mcX} \tpg(\bx) \frac{1}{2n} \sum_{i \in [n]} R
    \exp(-\gamma \ell)}{\tpg(\mcX)} \tag{by \Cref{eq:tpg-ratio-lb}} \\
    &= \frac{\sum_{\bx \in \mcX} \tpg(\bx) \frac{1}{2} R \exp(-\gamma
    \ell)}{\tpg(\mcX)} \\
    &= \frac{\tpg(\mcX) R \exp(-\gamma \ell)}{2\tpg(\mcX)} \\
    &= \frac{R \exp(-\gamma \ell)}{2}.
  \end{align*}
  As a result,
  \begin{equation}
    \Phi(\Mgamma)
    \le \frac{\Qgamma(\mcX, \overline{\mcX})}{\tpg(\mcX)}
    \le \frac{R \exp(-\gamma \ell)}{2}.
  \end{equation}
  Next, we lower-bound $\lambda_2(\Pgamma)$ with Cheeger's
  inequality~\citep{jerrum1988conductance,lawler1988bounds}:
  \begin{align*}
    \numberthis \label{eq:cheegers}
    1 - 2\Phi(\Mgamma) \le \lambda_2
  \end{align*}
  which yields
  \begin{equation}
    \lambda_2 \ge 1 - \exp(-\ell \gamma) R.
  \end{equation}
  Finally, we bound the relaxation time:
  \begin{align*}
    \relaxationgamma
    &= \frac{1}{1-\lambda_2} \\
    &\ge \frac{1}{1-(1 - \exp(-\ell \gamma) R)} \\
    &= \frac{1}{\exp(-\ell \gamma) R} \\
    &= \Omega(\ell \gamma).
  \end{align*}
\end{proof}

Taken together,
\Cref{lem:sol-density-increasing,lem:mixing-time} describe the trade-off in our
choice of $\gamma$: for small values of $\gamma$, samples from $\tpg$ rarely
fall in $\mcX$, making rejection sampling inefficient. For large values of
$\gamma$, the mixing time of the simple chain increases exponentially, requiring
many iterations of the simple chain to generate each sample. These results tell
us that the optimal choice of $\gamma$ is finite.

\finitegamma*

\begin{proof}
  Using \Cref{eq:num-samp-lb,lem:sol-density-increasing},
  \begin{align*}
    \lim_{\gamma \to -\infty} \numsampgamma
    &\ge \lim_{\gamma \to -\infty} \numsamplbgamma
    = \lim_{\gamma \to -\infty} \frac{(\relaxationgamma-1) \log
    \p{\frac{3}{4\varepsilon \solutiondensitystar}}}{\solutiondensitystar\p{1 +
    \frac{2\varepsilon}{3}}}
    = \infty.
  \end{align*}
  Using \Cref{eq:num-samp-lb,lem:mixing-time},
  \begin{align*}
    \lim_{\gamma \to \infty} \numsampgamma
    &\ge \lim_{\gamma \to \infty} \numsamplbgamma
    = \lim_{\gamma \to \infty} \frac{(\relaxationgamma-1) \log
    \p{\frac{3}{4\varepsilon \solutiondensitystar}}}{\solutiondensitystar\p{1 +
    \frac{2\varepsilon}{3}}}
    = \lim_{\gamma \to \infty} \Omega(\exp(\ell \gamma))
    = \infty.
  \end{align*}
  Since $\numsampgamma$ is finite for any finite $\gamma$, it must be minimized
  by some finite $\gamma$.
\end{proof}

\section{Omitted Proofs and Details from \Cref{sec:reduced}}
\label{app:reduced-omitted}

\subsection{An algorithm based on the reduced chain}
\label{app:reduced-alg}

\begin{algorithm}
  \caption{\textproc{Reduced}($\bV, \bc, f, k, t, \bx_0$)}
  \label{alg:reduced}
  \begin{algorithmic}
    \State $\bx \gets \bx_0$
    \For{$t$ iterations}
    \WithProbability{$0.5$}
    \State \textbf{continue}
    \EndWithProbability
    \State $\bz \gets$ random $k$ items from $\bx$ without replacement
    \WithProbability{$1 - 1/\prod_i \binom{\bx[i]}{\bz[i]}$}
    \State \textbf{continue}
    \EndWithProbability
    \State $\mc Z(\bz) \gets \{\bz' : \bV \bz = \bV \bz'\}$
    \Comment{Using ILP; can be cached for efficiency}
    \State $\mc S \gets \varnothing$
    \For{$\bz' \in \mc Z(\bz)$}
    \State $\bx' \gets \bx - \bz + \bz'$
    \State Add $\bx'$ to $\mc S$
    \EndFor
    \State $\bx \gets$ a random $\bx \in \mc S$ according to $\pi \given \bx
    \in \mc S$ (i.e., proportional to $f(\cdot)$ on $\mc S$)
    \EndFor
    \State \Return $\bx$
  \end{algorithmic}
\end{algorithm}

\paragraph*{Choosing an initial state $\bx_0$.}

In order to guarantee that the initial state for the reduced chain $\bx_0$
satisfies $\pi(\bx_0) \ge 1/|\mcX|$, we use the fact that ILP solvers like
Gurobi support maximizing linear objectives subject to a linear constraint.
While our chosen $f(\cdot)$
from \cref{eq:piD-def} is not linear in $\bx$, we can instead use the following
linear approximation to $\log f$:
\begin{align*}
  \numberthis \label{eq:pi-approx}
  L(\bx) = \sum_{i=1}^n \bx[i] \log \Pr_{\mc D}[\bv_i].
\end{align*}
Note that to obtain $L$, we dropped the multinomial coefficient in
\cref{eq:piD-def}. To obtain an intial $\bx_0$ with large $\pi(\bx_0)$, we use
Gurobi to enumerate the $N$ solutions with the largest values of $L$ and, of
those, choose the one that maximizes $f$. Gurobi additionally provides a bound
on the largest objective value of any solution not yet enumerated, which we can
translate into a bound on $f$. This allows us to guarantee that our chosen
initial $\bx_0$ satisfies $\pi(\bx_0) \ge 1/|\mcX|$.

Finally, we remark on hyperparameters $k$ and $t$. As discussed in
\Cref{sec:reduced-empirical}, $k=3$ appears to be a good choice. For $t$,
experiments with Alabama and Nevada suggest $t \approx 10^5$ suffices;
however, given that AL and NV have different characteristics, it is possible
that other states will require larger values of $t$.

\subsection{Omitted proofs from \Cref{sec:reduced-theory}}
\label{app:proofs-reduced}

We restate each result from \Cref{sec:reduced-theory} before proving it.

\statespaceub*

\begin{proof}
  We take inspiration from a long line of work that bounds the number of
  feasible solutions to a knapsack problem
  problem~\citep{beged1972lower,padberg1971remark,achou1974number,lambe1974bounds,lambe1977upper,hujter1988improved,mahmoudvand2010number}.
  The first inequality is trivially true since $\mcX \subset \mc Y$.
  We use the following combinatorial identity (see, e.g., \citet[eq.
  (5.2)]{feller1968introduction}). Let $A_{r, n}$ be the number of distinct
  tuples $\br \in \Zplus^n$ such that $\sum_{i=1}^n \br[i] = r$. Then,
  \begin{align*}
    A_{r, n} = \binom{n+r-1}{r}.
    \numberthis \label{eq:feller}
  \end{align*}
  Let $j$ be the index such that $\bv_i[j] = 1$ for all $i$, which exists by
  assumption since $\bc$ encodes the number of elements in each solution.
  Observe that every $\bx \in \mc Y$ satisfies 
  \begin{align*}
    \bV \bx &\le \bc \\
    (\bV \bx)[j] &\le \bc[j] \\
    \sum_{i=1}^n \bx[i] \bv_i[j] &\le m \\
    \sum_{i=1}^n \bx[i] &\le m
  \end{align*}
  This means that each $\bx \in \mc Y$ satisfies $\sum_{i=1}^n \bx[i] = \ell$
  for some $\ell \in \{0, \dots, m\}$. Thus, by \cref{eq:feller},
  \begin{align*}
    |\mc Y|
    &\le \sum_{\ell=0}^m A_{\ell, n} \\
    &= \sum_{\ell=0}^m \binom{n+\ell-1}{\ell} \\
    &= \sum_{\ell=0}^m \binom{n+\ell-1}{n-1} \\
    &= \binom{n+m}{n} \tag{hockey stick identity} \\
    &= \binom{n+m}{m}.
  \end{align*}
  Therefore,
  \begin{align*}
    \log |\mc Y|
    &\le \log \binom{n+m}{m} \\
    &\le \log \p{\frac{e(n+m)}{m}}^m \\
    &= m \log \p{\frac{e(n+m)}{m}} \\
    &\le m \p{1 + \log \p{\frac{n}{m} + 1}} \\
    &\le m \p{1 + \log (n + 1)}.
  \end{align*}
  The second line is a standard upper bound on binomial coefficients and follows
  from the fact that $m! \ge \frac{m^m}{e^{m-1}}$. 
\end{proof}

\hyperrectangle*

\begin{proof}
  We will show by induction that for any $\bx, \bx' \in \mcX$, $\bx$ can be
  transformed to $\bx'$ through a sequence of transitions under $\Preducedtwo$.
  Because we're considering the case where $\mc V$ is the intersection between a
  lattice and a hyperrectangle, it suffices to consider each dimension
  independently. Thus, we proceed under the assumption that $d = 1$. Let $\by$
  and $\by'$ respectively be the matrix representations of $\bx$ and $\bx'$
  obtained by listing all of their elements $\bv_i$ in order, with repeats for
  multiplicity. Thus, $\by, \by' \in \Zplus^{m}$. Note that this is a one-to-one
  mapping.

  We proceed by induction on $s \in \{1, \dots, m\}$. Assume $\by[i] = \by'[i]$
  for all $i < s$. We will show the existence of a path from $\bx$ to some
  $\bx^+$ such that the corresponding $\by^+[i] = \by'[i]$ for all $i \le s$. If
  $\by[s] = \by'[s]$, we are done; otherwise, assume without loss of generality
  that $\by[s] < \by'[s]$.

  Let $j$ be an index such that $j > s$ and $\by[j] > \by'[j]$.
  This must exist since $\|\by\|_1 = \bV \bx = \bV \bx' = \|\by'\|_1$, $\by[s]'
  > \by[s]$, and $\by[i] = \by'[i]$ for all $i < s$. We will construct a valid
  solution $\bx''$ by
  removing $(\by[s], \by[j])$ from $\bx$ and replacing them with $(\by[s] + 1,
  \by[j] - 1)$. By convexity, these two new elements are both in $\mc V$.
  Repeating this argument produces a path from $\bx$ to $\bx^+$ where $\bx^+[i]
  = \bx'[i]$ for all $i \le s$ as desired.
\end{proof}

\reducedconvergence*
\begin{proof}
  For clarity of notation, we will write $\pi^+$ instead of $\pi^+(\pi_0)$. Let
  $\mc C(\bx_0)$ be the partition to which $\bx_0$ belongs, i.e., $\bx_0 \in \mc
  C(\bx_0)$. For a connected component $\mc C$, consider the Markov chain
  $\Mreduced(\mc C)$ defined over the states in $\mc C$. $\Mreduced(\mc
  C(\bx_0))$ is lazy and irreducible and $\pi$ satisfies the detailed balance
  equations. As a result, the stationary distribution of $\Mreduced(\mc
  C(\bx_0))$ is $\pi \given \bx \in \mc C(\bx_0)$. Because $\bx_0$ is sampled
  from $\pi_0$, we can write $\pi^+$ as a linear combination of conditional
  distributions:
  \begin{align*}
    \numberthis \label{eq:pi-plus}
    \pi^+
    = \sum_{\bx \in \mcX} \pi_0(\bx) \cdot (\pi \given \bx \in \mc C(\bx))
    = \sum_{i \in [c]} \pi_0(\mc C_i) \cdot (\pi \given \bx \in \mc C_i).
  \end{align*}
  Therefore,
  \begin{align*}
    \dTV(\pi^+, \pi)
    &= \frac{1}{2} \sum_{\bx \in \mcX} |\pi^+(\bx) - \pi(\bx)| \\
    &= \frac{1}{2} \sum_{i \in [c]} \sum_{\bx \in \mc C_i} |\pi^+(\bx) -
    \pi(\bx)| \\
    &= \frac{1}{2} \sum_{i \in [c]} \sum_{\bx \in \mc C_i} \left|\pi_0(\mc C_i)
    \cdot (\pi \given \bx \in \mc C_i)(\bx) - \pi(\bx)\right| \tag{by
    \Cref{eq:pi-plus}} \\
    &= \frac{1}{2} \sum_{i \in [c]} \sum_{\bx \in \mc C_i} \left|\frac{\pi(\bx)
    \pi_0(\mc C_i)}{\pi(\mc C_i)} - \pi(\bx)\right| \\
    &= \frac{1}{2} \sum_{i \in [c]} \sum_{\bx \in \mc C_i} \pi(\bx)
    \left|\frac{\pi_0(\mc C_i)}{\pi(\mc C_i)} - 1 \right| \\
    &= \frac{1}{2} \sum_{i \in [c]} \pi(\mc C_i) \left|\frac{\pi_0(\mc
    C_i)}{\pi(\mc C_i)} - 1 \right| \\
    &= \frac{1}{2} \sum_{i \in [c]} \left|\pi_0(\mc C_i) - \pi(\mc C_i) \right|.
    \numberthis \label{eq:dtv-pi-plus}
  \end{align*}
  For any distribution $\pi'$ such that $\pi'(\mc C_i) = \pi_0(\mc C_i)$ for all
  $i \in [c]$,
  \begin{align*}
    \dTV(\pi', \pi)
    &= \frac{1}{2} \sum_{\bx \in \mcX} |\pi'(\bx) - \pi(\bx)| \\
    &= \frac{1}{2} \sum_{i \in [c]} \sum_{\bx \in \mc C_i} |\pi'(\bx) -
    \pi(\bx)| \\
    &\ge \frac{1}{2} \sum_{i \in [c]} \left| \sum_{\bx \in \mc C_i} \pi'(\bx) -
    \pi(\bx) \right| \\
    &= \frac{1}{2} \sum_{i \in [c]} \left|\pi'(\mc C_i) - \pi(\mc C_i) \right|
    \\
    &= \frac{1}{2} \sum_{i \in [c]} \left|\pi_0(\mc C_i) - \pi(\mc C_i) \right|
    \tag{by assumption} \\
    &= \dTV(\pi^+, \pi) \tag{by \Cref{eq:dtv-pi-plus}}
  \end{align*}
  as claimed.
\end{proof}

\section{Additional Empirical Results}
\label{app:additional-empirical}

\paragraph*{Computation time.}

To compare the computation time per iteration of the simple and reduced chains,
we select a random 20 blocks from Alabama and perform $10^5$ MCMC iterations of
each method. We choose blocks such that:
\begin{itemize}
  \item $|\mcX| > 1$ (some blocks have no solution with our formulation, which
    we ignore; in practice, we can solve these by matching a subset of
    statistics. See \Cref{app:encoding} for details.)
  \item $m > 3$ (since we will be running the reduce chain with $k=3$, and the
    reduced chain is trivial when $m \le k$)
\end{itemize}
For each block, we run the simple chain with $\gamma = 0.8$ and the reduced
chain with $k=3$ for $10^5$ iterations. Note that $\gamma$ is unlikely to affect
computation time; it primarily influences mixing time. We find that $10^5$
iterations requires on average 10.22 seconds (std. dev. 2.71) and 8.23 seconds
(std. dev. 6.29) for the simple and reduced chains respectively. Experiments are
run on an Apple M1 MacBook Pro with 32GB memory. For reference, enumerating 5000
solutions via ILP for 20 randomly selected Alabama blocks (with at least 5000
solutions each) takes 4.58 seconds (std. dev. 0.95) on the same hardware. As a
result, running the reduced chain on blocks with $|\mcX_b| \ge 5000$ does not
dramatically increase computation time.

As a final point of reference, running \Cref{alg:overall} on a computing cluster
with $N=5,000$, $k=3$, and $t=10^5$ takes around 21 days (Alabama) and 37 days
(Nevada) of 12-threaded CPU time to sample a dataset for the entire state. This
is highly parallelizable since our algorithm operates independently on each
block.

\begin{figure}[ht]
  \centering
  \begin{subfigure}[b]{0.48\textwidth}
    \includegraphics[width=\textwidth]{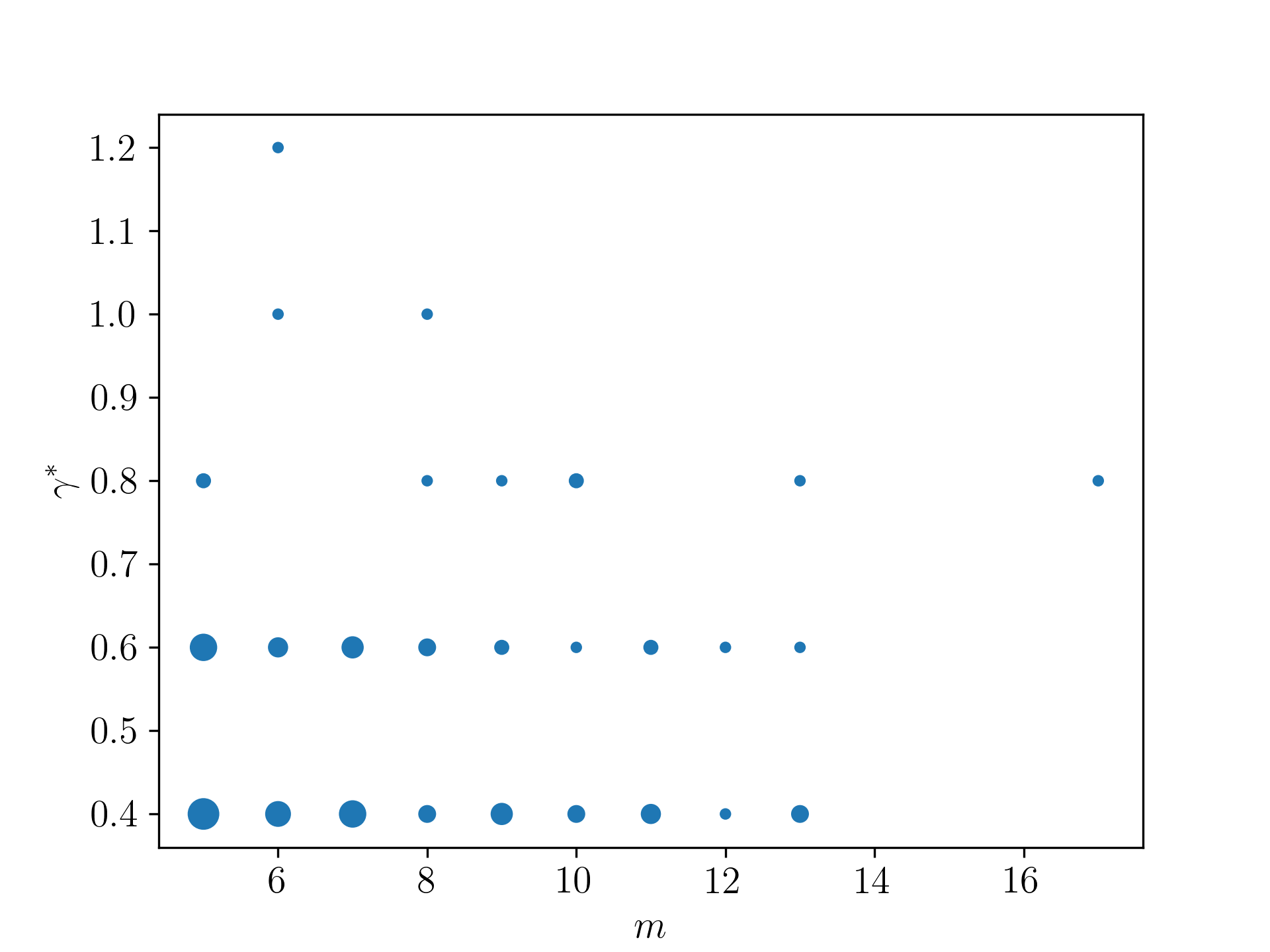}
    \caption{AL}
    \label{fig:opt-gamma-AL}
  \end{subfigure}
  \hfill
  \begin{subfigure}[b]{0.48\textwidth}
    \includegraphics[width=\textwidth]{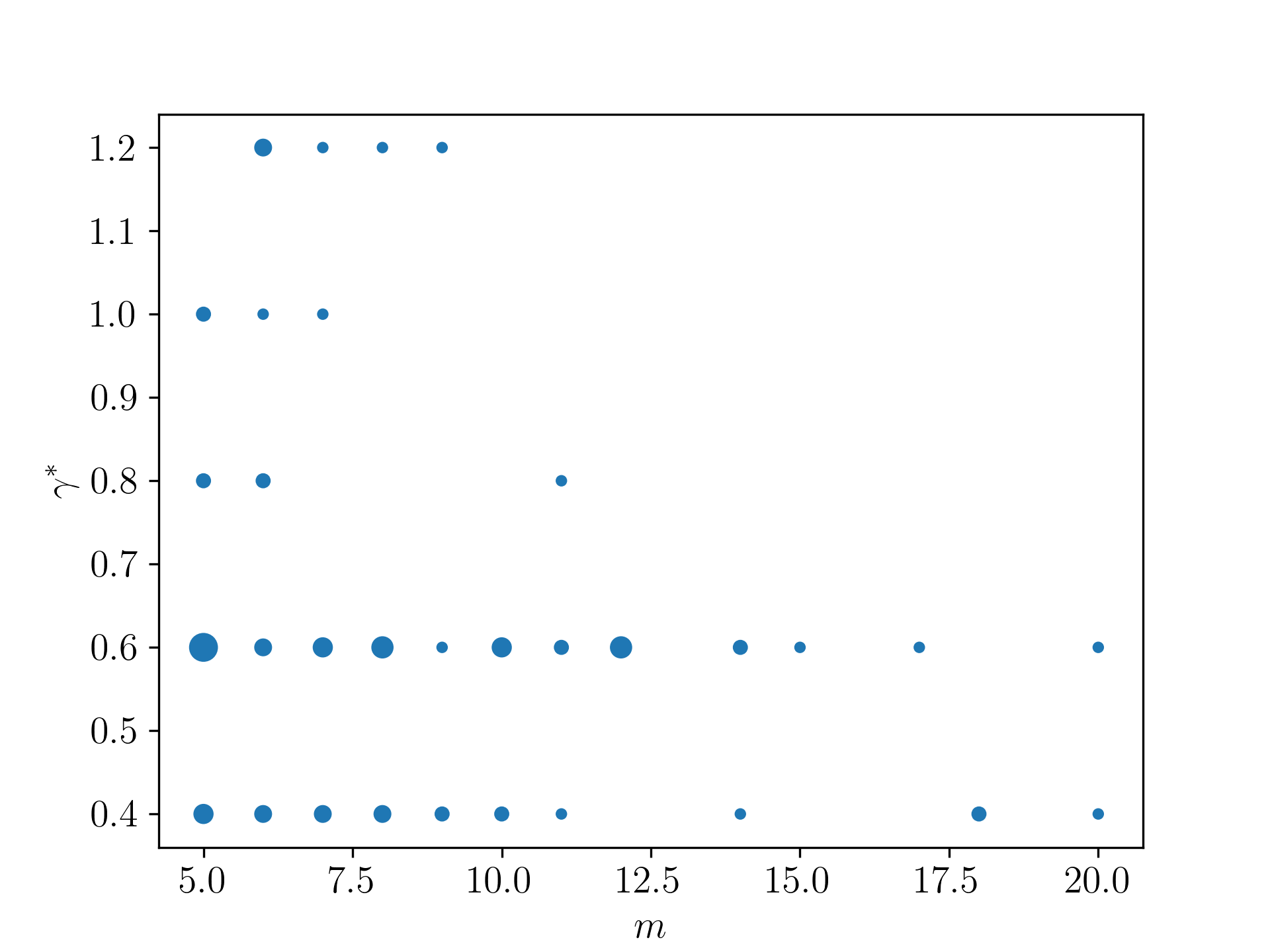}
    \caption{NV}
    \label{fig:opt-gamma-NV}
  \end{subfigure}
  \caption{Optimal choice of $\gamma$ vs. $m$. Marker size is proportional to
  the number of such blocks.}
  \label{fig:opt-gamma}
\end{figure}

\begin{figure}[ht]
  \centering
  \begin{subfigure}[b]{0.48\textwidth}
    \includegraphics[width=\textwidth]{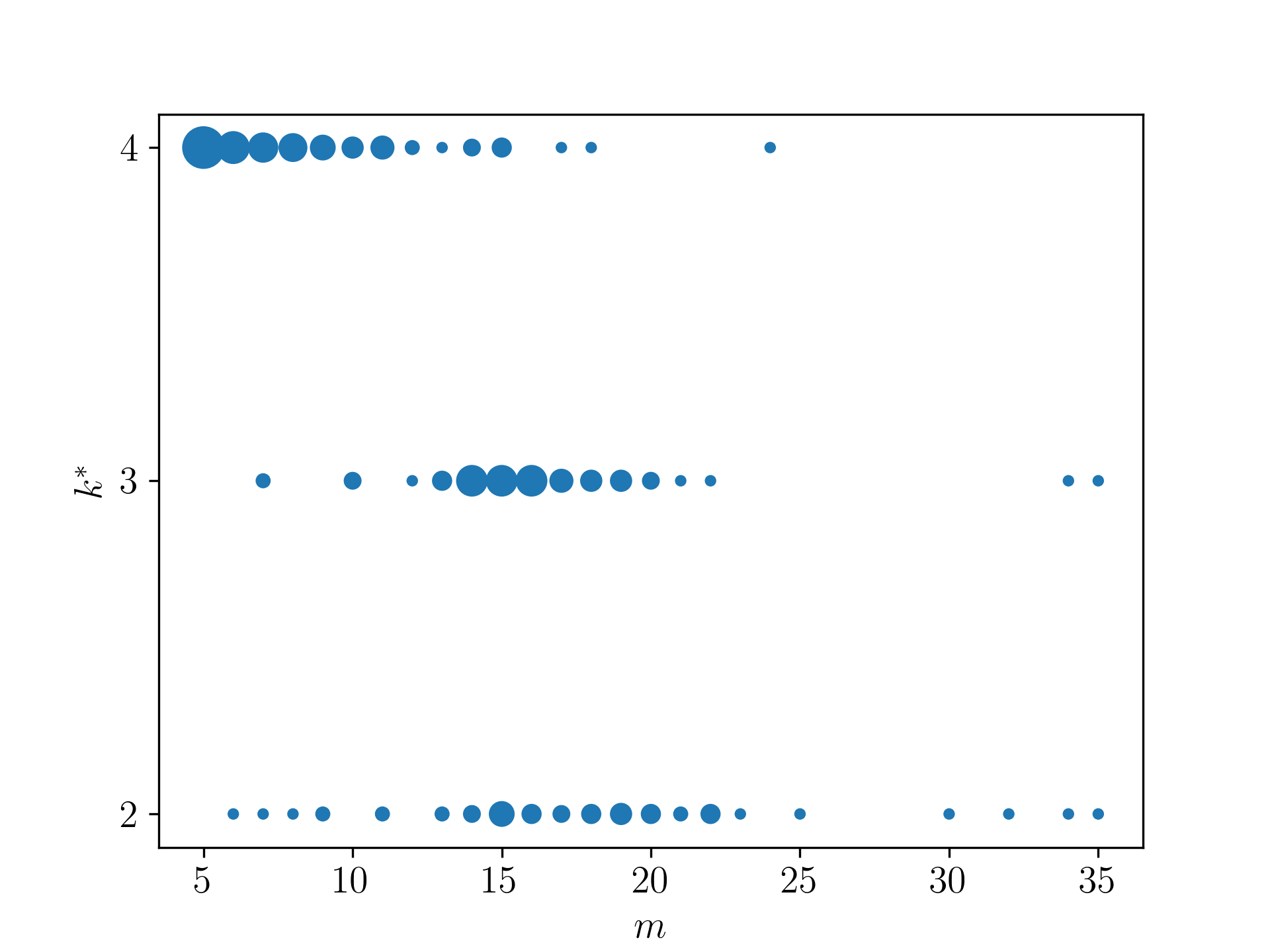}
    \caption{AL}
    \label{fig:opt-k-AL}
  \end{subfigure}
  \hfill
  \begin{subfigure}[b]{0.48\textwidth}
    \includegraphics[width=\textwidth]{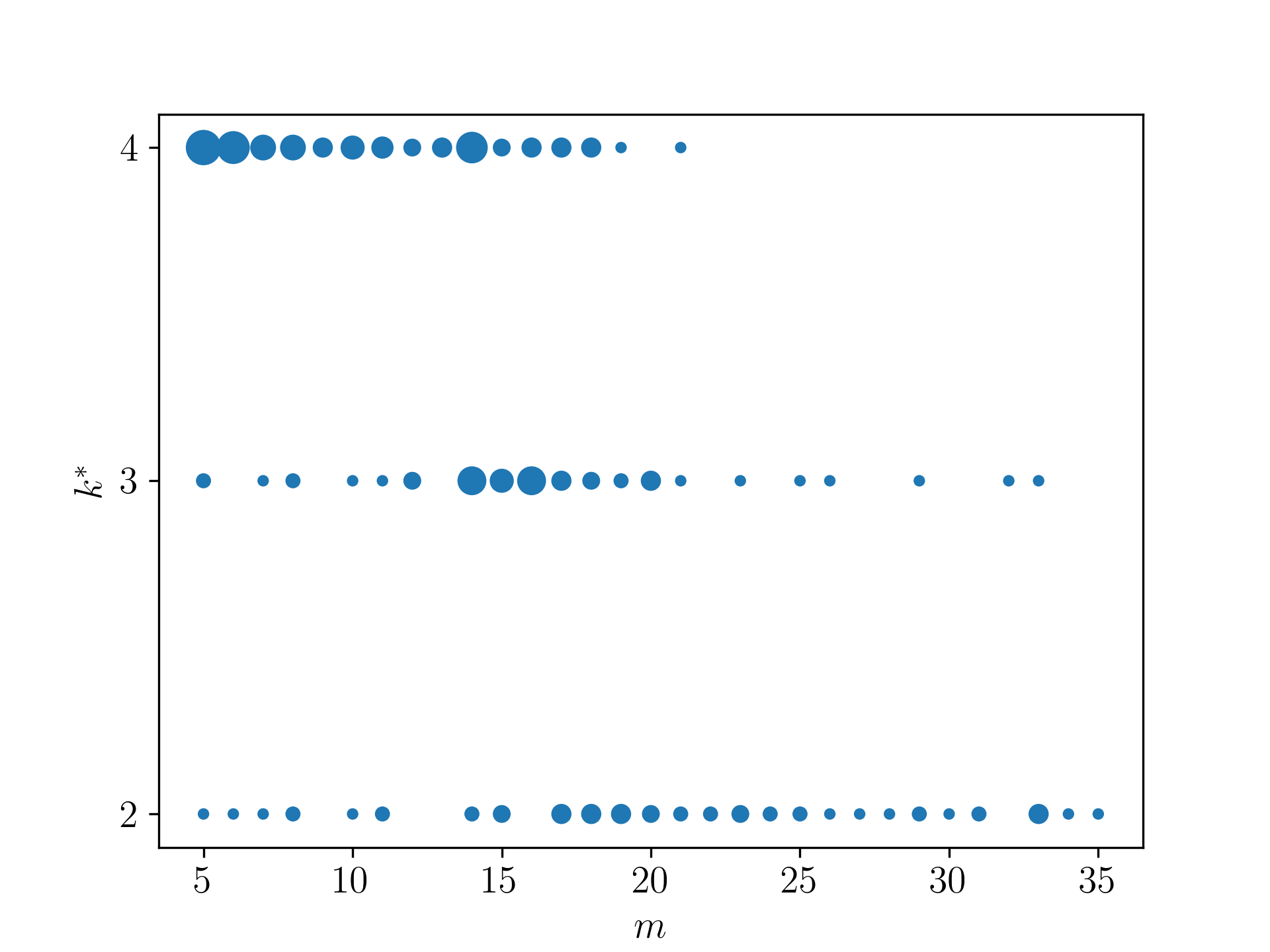}
    \caption{NV}
    \label{fig:opt-k-NV}
  \end{subfigure}
  \caption{Optimal choice of $k$ vs. $m$. Marker size is proportional to the
  number of such blocks.}
  \label{fig:opt-k}
\end{figure}

\paragraph*{Optimal choices of $\gamma$ and $k$.}

\Cref{fig:opt-gamma} shows the relationship between $m$ and the $\gamma^*$ that
minimizes $\numsamplbgammastar$. There are a few blocks for which $\gamma = 1.2$
appears to be optimal, but this is often because low-precision estimates for
$\lambda_2(P_{\gamma=1.2})$ lead to underestimates of
$\underline{N}_{\gamma=1.2}$. Computing $\lambda_2(\Pgamma)$ is computationally
very expensive for larger values of $\gamma$, which is why we don't attempt to
evaluate $\gamma > 1.2$. \Cref{fig:opt-k} provides analogous results for the
reduced chain, finding that $k = 4$ appears to be optimal for smaller blocks,
but as blocks get larger, smaller values of $k$ are better.

\begin{figure}[ht]
  \centering
  \includegraphics[width=0.8\textwidth]{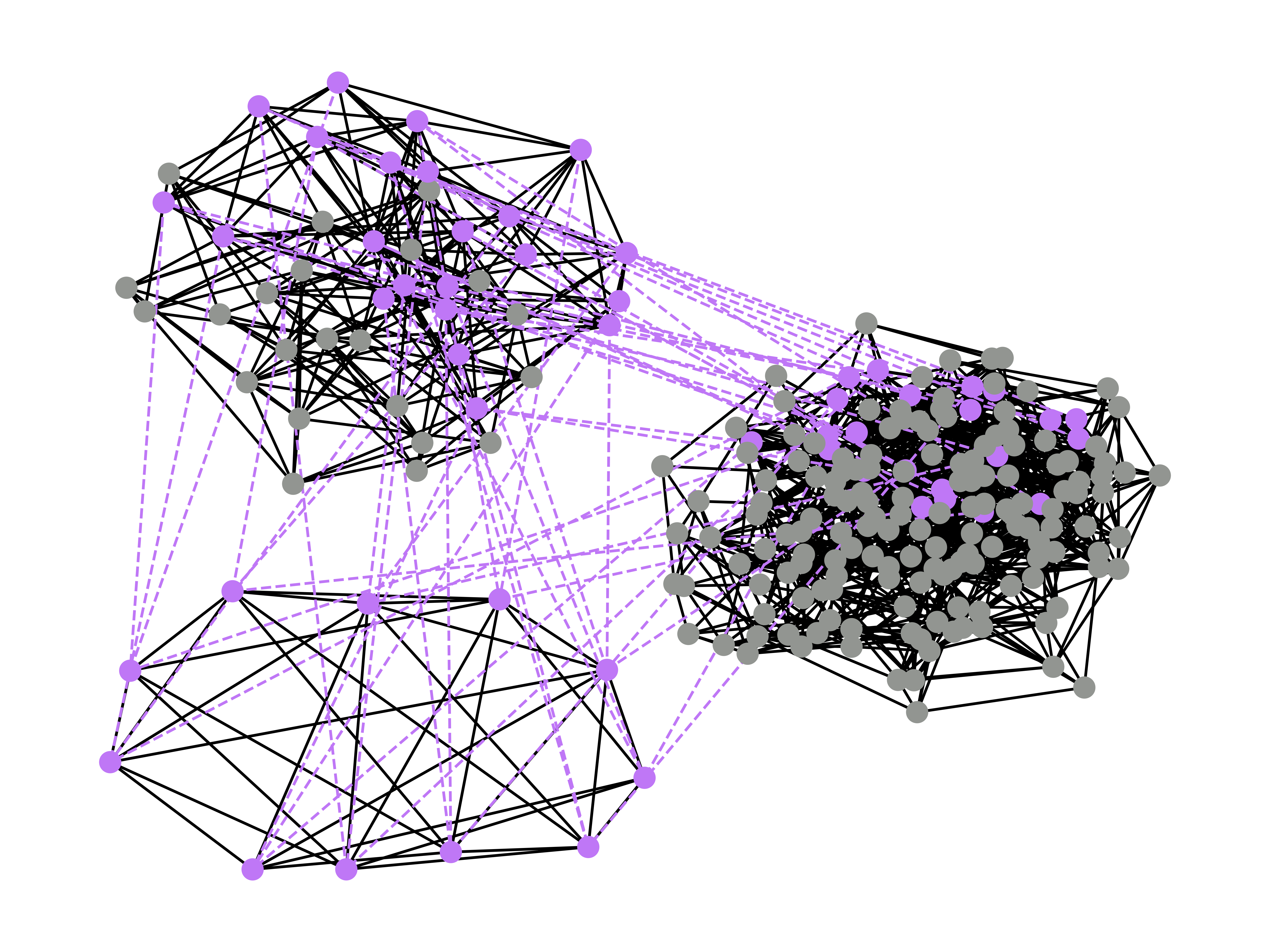}
  \caption{Connected components in the state space of the reduced chain with
    $k=2$ for a particular block in Alabama (Madison County, tract 2501, block
    2053). Dashed purple lines show additional crossing edges present with
  $k=3$, making the resulting graph connected.}
  \label{fig:components}
\end{figure}

\paragraph*{Disconnected graphs.}

Finally, we consider instances where the reduced chain is not irreducible for
$k=2$. \Cref{fig:components} shows one such example, where we draw the
(undirected) graph given by one-step transitions using $P_{k=2}$. Each node in
the graph is some $\bx \in \mcX$. The graph has 3 connected components, and
crucially, increasing to $k=3$ adds many crossing edges (shown with purple
dashes). Other examples we examined are qualitatively similar, though the
density of additional connections with $k=3$ varies.

\section{Hardness Results}
\label{app:hardness}
Existing hardness results for multidimensional knapsack and subset sum do not
strictly apply here since (1) we are interested in the decision version, and (2)
our solutions can be multisets. However, standard 3SAT reductions to subset sum
can be easily adapted to our multidimensional multiset case.
\begin{claim}
  \label{clm:hardness}
  \probname\ is NP-hard.
\end{claim}
\begin{proof}
  Given a 3SAT instance with variables $y_1, \dots, y_n$ and clauses $c_1, \dots,
  c_m$, we can construct a \probname\ instance with $2n + 3m$ elements and
  dimension $n + m$ using a slight variant of the standard 3SAT reduction to
  subset sum. We use the following encoding. For $i \in [n]$,
  \begin{align*}
    \bv_{2(i-1)}[j] &=
    \begin{cases}
      1 & j \le n \wedge i = j \\
      1 & j > n \wedge c_{j-n} \text{ contains } y_i \\
      0 & \text{otherwise}
    \end{cases} \\
    \bv_{2i-1}[j] &=
    \begin{cases}
      1 & j \le n \wedge i = j \\
      1 & j > n \wedge c_{j-n} \text{ contains } \overline{y_i} \\
      0 & \text{otherwise}
    \end{cases} \\
    \bv_{2i}[j] &=
    \begin{cases}
      4 & j = 2n -1 + i \\
      0 & \text{otherwise}
    \end{cases} \\
    \bv_{2i+1}[j] &=
    \begin{cases}
      5 & j = 2n + i \\
      0 & \text{otherwise}
    \end{cases} \\
    \bv_{2i+2}[j] &=
    \begin{cases}
      6 & j = 2n + 1 + i \\
      0 & \text{otherwise}
    \end{cases}
  \end{align*}
  Intuitively, the $\bv_i$ for $i \le 2n$ represent the variables, with one
  variable for $y_i$ and one for $\overline{y}_i$. The remaining $3m$ vectors
  are slack vectors for the clauses, corresponding to the fact that clause $j$
  can be satisfied by 1, 2, or 3 of its variables. Let the counts vector be
  \begin{align*}
    \bc[j] =
    \begin{cases}
      1 & j \le n \\
      7 & j > n
    \end{cases}.
  \end{align*}
  We will show that this \probname\ instance is feasible if and only if the
  original 3SAT instance was feasible.

  \textbf{3SAT solution $\Longrightarrow$ \probname\ solution.} Let $Y_i^*$ be a
  3SAT solution. Let $z_j$ be the number of ways in which $c_j$ is satisfied by
  $Y_i^*$. Note that $z_j \in \{1, 2, 3\}$.
  \begin{align*}
    \bx[j] =
    \begin{cases}
      1 & j \le 2n \wedge j \text{ is odd} \wedge Y_{(j+1)/2}^* = \text{ true}
      \\
      1 & j \le 2n \wedge j \text{ is even} \wedge Y_{j/2}^* = \text{ false} \\
      1 & j > 2n \wedge j \mod 3 = 1 \wedge z_{(j+2)/3} = 3 \\
      1 & j > 2n \wedge j \mod 3 = 2 \wedge z_{(j+1)/3} = 2 \\
      1 & j > 2n \wedge j \mod 3 = 3 \wedge z_{j/3} = 1 \\
      0 & \text{otherwise}
    \end{cases}
  \end{align*}
  Observe that $\bV \bx = \bc$ because:
  \begin{itemize}
    \item For $j \le n$, $(\bV \bx)[j] = \bc[j]$. This is because exactly one
      of $\bx[2i-1], \bx[2i]$ can be 1 for $i \le n$ by construction.
    \item For $j > n$, $(\bV \bx)[j] = \bc[j]$. This is because for each
      clause $\ell$, the first $2n$ columns of $\bV$ contribute exactly $z_\ell$
      to $\bc[\ell + 2n]$. The remaining $3m$ columns contribute $7 - z_\ell$.
  \end{itemize}

  \textbf{\probname\ solution $\Longrightarrow$ 3SAT solution.} Given a solution
  $\bx$, we assign variables as follows. $y_i$ is set to true if $\bx[2i-1] = 1$
  and false if not. We will show that for $j > n$, $\bV\bx[j] = \bc[j]
  \Longrightarrow$ clause $c_{j - n}$ is satisfied. To see this, note that no
  multiset consisting of only the last $3m$ columns of $\bV$ can sum to $7$,
  since they only contain the values 4, 5, and 6. Thus, $\bx$ must select at
  least one column from the first $2n$ that has a nonzero entry in position $j$.
  The only such columns correspond to variables satisfying $c_{j - n}$ (either
  $\bv_{2j-1}$ or $\bv_{2j}$).
\end{proof}

\section{Pathological Examples for the Reduced Chain}
\label{app:bad-examples}

\begin{example}[Disconnected reduced chain]
  \label{ex:disconnected}
  In the following example, the reduced chain is disconnected for $k=2$.
\end{example}

\begin{align*}
  \bV =
  \begin{pmatrix}
    3 & 0 & 0 & 1 \\
    0 & 3 & 0 & 1 \\
    0 & 0 & 3 & 1 \\
    1 & 1 & 1 & 1 \\
  \end{pmatrix}
  &&
  \bc =
  \begin{pmatrix}
    3 \\
    3 \\
    3 \\
    3
  \end{pmatrix}
\end{align*}
Solutions:
\begin{align*}
  \mcX = \left\{
    \begin{pmatrix}
      1 \\
      1 \\
      1 \\
      0
    \end{pmatrix},
    \begin{pmatrix}
      0 \\
      0 \\
      0 \\
      3
    \end{pmatrix}
  \right\}
\end{align*}
To move between the two solutions in $\mcX$, we must remove and replace 3
elements at a time, meaning the reduced chain is disconnected for $k=2$. This
example can be generalized to be disconnected for any $k < m$.

\begin{example}[Connected reduced chain with high mixing time]
  \label{ex:exp-mixing}
  \label{ex:connected-high-mt}
\end{example}

Here, we provide an instance $(\bV, \bc, \pi)$ such that the reduced chain with
$k=2$ is connected but has exponentially large mixing time.
Let $\pi$ assign equal probability to all solutions, i.e., $\pi(\bx) =
1/|\mcX|$. Let $\mc I_\ell$ be the set of $\ell$-dimensional unit basis vectors,
i.e., each $\be_i \in \mc I_\ell$ has a 1 in the $i$th position and 0 elsewhere.
Let $\bM_\ell \in \Zplus^{\ell \times \p{\ell + \binom{\ell}{2}}}$ be the matrix
formed by taking as columns all possible sums of two vectors (with replacement)
from $\mc I_\ell$ and multiplying them by $\ell$. For example,
\begin{align*}
  \bM_3 =
  \begin{pmatrix}
    6 & 0 & 0 & 3 & 3 & 0 \\
    0 & 6 & 0 & 3 & 0 & 3 \\
    0 & 0 & 6 & 0 & 3 & 3
  \end{pmatrix}.
\end{align*}
\newcommand{\onemat}[2]{\mathbf{1}_{#1 \times #2}}
\newcommand{\zeromat}[2]{\mathbf{0}_{#1 \times #2}}
Let $\bS_\ell \in \Zplus^{\ell \times (\ell-2)}$ be the matrix where column $i
\in \{1, \dots, \ell-2\}$ has as its first $i+1$ entries $2\ell$ and 0
elsewhere. For example,
\begin{align*}
  \bS_3 =
  \begin{pmatrix}
    6 \\
    6 \\
    0 
  \end{pmatrix}
  &&
  \bS_4 =
  \begin{pmatrix}
    8 & 8 \\
    8 & 8 \\
    0 & 8 \\
    0 & 0
  \end{pmatrix}.
\end{align*}
Let $\bT_\ell \in \Zplus^{\ell \times (2\ell-1)}$ be the matrix where the $i$th
column has $i+1$ everywhere. For example,
\begin{align*}
  \bT_3 =
  \begin{pmatrix}
    2 & 3 & 4 & 5 & 6 \\
    2 & 3 & 4 & 5 & 6 \\
    2 & 3 & 4 & 5 & 6
  \end{pmatrix}.
\end{align*}
Let $\onemat{i}{j}$ be the $i \times j$ matrix consisting of all ones.
Similarly, let $\zeromat{i}{j}$ be the matrix consisting of all zeros. Then,
define
\begin{align*}
  \bV_\ell =
  \begin{pmatrix}
    \\
    \bM_\ell & \zeromat{\ell}{1} & \bS_\ell & \bT_\ell & &
    \onemat{\ell}{\p{\ell + \binom{\ell}{2}}} & \onemat{\ell}{1} &
    \onemat{\ell}{(\ell-2)} & \onemat{\ell}{(2\ell-1)} \\
           & & & & \onemat{2\ell}{1} \\
    \onemat{\ell}{\p{\ell + \binom{\ell}{2}}} & \onemat{\ell}{1} &
    \onemat{\ell}{(\ell-2)} & \onemat{\ell}{(2\ell-1)} & & \bM_\ell &
    \zeromat{\ell}{1} & \bS_\ell & \bT_\ell \\
    \\
  \end{pmatrix}
\end{align*}

\setcounter{MaxMatrixCols}{30}
For example,
\begin{align*}
  \bV_3 =
  \begin{pmatrix}
    6 & 0 & 0 & 3 & 3 & 0 & 0 & 6 & 2 & 3 & 4 & 5 & 6 & 1 & 1 & 1 & 1 & 1 & 1 & 1 & 1 & 1 & 1 & 1 & 1 & 1 & 1 \\
    0 & 6 & 0 & 3 & 0 & 3 & 0 & 6 & 2 & 3 & 4 & 5 & 6 & 1 & 1 & 1 & 1 & 1 & 1 & 1 & 1 & 1 & 1 & 1 & 1 & 1 & 1 \\
    0 & 0 & 6 & 0 & 3 & 3 & 0 & 0 & 2 & 3 & 4 & 5 & 6 & 1 & 1 & 1 & 1 & 1 & 1 & 1 & 1 & 1 & 1 & 1 & 1 & 1 & 1 \\
    1 & 1 & 1 & 1 & 1 & 1 & 1 & 1 & 1 & 1 & 1 & 1 & 1 & 1 & 6 & 0 & 0 & 3 & 3 & 0 & 0 & 6 & 2 & 3 & 4 & 5 & 6  \\
    1 & 1 & 1 & 1 & 1 & 1 & 1 & 1 & 1 & 1 & 1 & 1 & 1 & 1 & 0 & 6 & 0 & 3 & 0 & 3 & 0 & 6 & 2 & 3 & 4 & 5 & 6  \\
    1 & 1 & 1 & 1 & 1 & 1 & 1 & 1 & 1 & 1 & 1 & 1 & 1 & 1 & 0 & 0 & 6 & 0 & 3 & 3 & 0 & 0 & 2 & 3 & 4 & 5 & 6  \\
  \end{pmatrix}
\end{align*}

Let $\bc_\ell \in \Zplus^{2\ell}$ have $2\ell$ in each entry. We will show that
the reduced chain on the problem defined by $\p{\bV_\ell, \bc_\ell, \pi}$ is
connected for $k=2$ but has relaxation time exponential in $\ell$. To do so, we
upper-bound conductance and use Cheeger's inequality.

\begin{lemma}
  The relaxation time of the chain from \Cref{ex:connected-high-mt} with $k=2$
  is $\Omega\p{\p{\frac{\ell}{e}}^\ell}$.
\end{lemma}

\begin{proof}
  Let $\mc V_L$ be the set of columns of $\bV_\ell$ shown in
  \textcolor{red}{\textbf{bold red}}.
  \begin{align*}
    \bV_\ell =
    \begin{pmatrix}
      \\
      \incolor{\bM_\ell} & \incolor{\zeromat{\ell}{1}} & \incolor{\bS_\ell} &
      \bT_\ell & & \onemat{\ell}{\p{\ell + \binom{\ell}{2}}} & \onemat{\ell}{1}
               & \onemat{\ell}{(\ell-2)} & \onemat{\ell}{(2\ell-1)} \\
               & & & & \onemat{2\ell}{1} \\
      \incolor{\onemat{\ell}{\p{\ell + \binom{\ell}{2}}}} &
      \incolor{\onemat{\ell}{1}} & \incolor{\onemat{\ell}{(\ell-2)}} &
      \onemat{\ell}{(2\ell-1)} & & \bM_\ell & \zeromat{\ell}{1} & \bS_\ell &
      \bT_\ell \\
      \\
    \end{pmatrix}
  \end{align*}
  For example,
  \begin{align*}
    \bV_3 =
    \begin{pmatrix}
      \incolor{6} & \incolor{0} & \incolor{0} & \incolor{3} & \incolor{3} & \incolor{0} & \incolor{0} & \incolor{6} & 2 & 3 & 4 & 5 & 6 & 1 & 1 & 1 & 1 & 1 & 1 & 1 & 1 & 1 & 1 & 1 & 1 & 1 & 1 \\
      \incolor{0} & \incolor{6} & \incolor{0} & \incolor{3} & \incolor{0} & \incolor{3} & \incolor{0} & \incolor{6} & 2 & 3 & 4 & 5 & 6 & 1 & 1 & 1 & 1 & 1 & 1 & 1 & 1 & 1 & 1 & 1 & 1 & 1 & 1 \\
      \incolor{0} & \incolor{0} & \incolor{6} & \incolor{0} & \incolor{3} & \incolor{3} & \incolor{0} & \incolor{0} & 2 & 3 & 4 & 5 & 6 & 1 & 1 & 1 & 1 & 1 & 1 & 1 & 1 & 1 & 1 & 1 & 1 & 1 & 1 \\
      \incolor{1} & \incolor{1} & \incolor{1} & \incolor{1} & \incolor{1} & \incolor{1} & \incolor{1} & \incolor{1} & 1 & 1 & 1 & 1 & 1 & 1 & 6 & 0 & 0 & 3 & 3 & 0 & 0 & 6 & 2 & 3 & 4 & 5 & 6  \\
      \incolor{1} & \incolor{1} & \incolor{1} & \incolor{1} & \incolor{1} & \incolor{1} & \incolor{1} & \incolor{1} & 1 & 1 & 1 & 1 & 1 & 1 & 0 & 6 & 0 & 3 & 0 & 3 & 0 & 6 & 2 & 3 & 4 & 5 & 6  \\
      \incolor{1} & \incolor{1} & \incolor{1} & \incolor{1} & \incolor{1} & \incolor{1} & \incolor{1} & \incolor{1} & 1 & 1 & 1 & 1 & 1 & 1 & 0 & 0 & 6 & 0 & 3 & 3 & 0 & 0 & 2 & 3 & 4 & 5 & 6  \\
    \end{pmatrix}.
  \end{align*}
  Formally, this is
  \begin{align*}
    \mc V_L
    &\triangleq
    \left\{\bv_i : i \le 2 \ell + \binom{\ell}{2} - 1\right\}.
  \end{align*}
  Let $\mcX_L \subset \mcX$ be the set of all solutions that only uses elements
  from $\mc V_L$. Formally, this is
  \begin{align*}
    \mcX_L \triangleq \left\{\bx \in \mcX : \bx[i] > 0 \Longrightarrow i \le
    2\ell + \binom{\ell}{2} - 1 \right\}.
  \end{align*}
  Define $\mcX_R$ analogously.

  \paragraph*{The chain is irreducible.}

  Observe that $\mcX_L$ is connected to $\mcX_R$ by swaps of size $k=2$:
  intuitively, starting at a state in $\mcX_L$, we can successively replace
  elements of $\mc V_L$ with elements from $\bS_\ell$ and the vector $[0  \dots
  0 ~ 1 \dots 1]$ until we reach the state with 1 copy of $[2\ell \dots 2\ell ~
  1 \dots 1]$ and $2\ell-1$ copies of $[0 \dots 0 ~ 1 \dots 1]$. From there, we
  can use elements from $\bT_\ell$ to reach the solution that consists of
  $2\ell$ copies of $[1 \cdots 1]$. Symmetrically, this state is reachable from
  $\mcX_R$, meaning the two are connected. More generally, from any solution, by
  construction the solution consisting of $2\ell$ copies of $[1 \cdots 1]$ is
  reachable.

  To bound conductance, we will show that:
  \begin{enumerate}
    \item only one state $\bx^* \in \mcX_L$ has edges to states in $\mcX
      \backslash \mcX_L$,
    \item $\pi(\bx^*) / \pi(\mcX_L) = 1/|\mcX_L|$ is small.
  \end{enumerate}

  \paragraph*{Only one state has crossing edges.}

  Let $\bx^* \in \mcX_L$ have an edge to some $\bx' \notin \mcX_L$. Treating
  $\bx^*$ as a multiset, we will reason about the number of copies of each
  vector $\bv$ it contains. We will show that $\bx^*$ has 1 copy of $[0 \cdots 0
  ~ 2\ell ~ 1 \cdots 1]$, 1 copy of $[2\ell \cdots 2\ell ~ 0 ~ 1 \cdots 1]$, and
  $2\ell-2$ copies of $[0 \cdots 0 ~ 1 \cdots 1]$.

  Let $\bv_a$ and $\bv_b$ be the two elements that can be removed from $\bx^*$
  and replaced by $\bv_c$ and $\bv_d$ to yield some $\bx' \notin \mc X_L$. Let
  $\bs = \bv_a + \bv_b$. Observe that the last $\ell$ entries of $\bs$ are each
  2 because $\bx^*$ only uses vectors in $\mc V_L$, which all have 1's in their
  last $\ell$ entries. Further, observe that each $\bv \notin \mc V_L$ has its
  first $\ell$ entries nonzero and identical. Without loss of generality, let
  $\bv_c \notin \mc V_L$. In order for $\bv_a + \bv_b = \bs = \bv_c + \bv_d$, it
  must be the case that each of the first $\ell$ entries of $\bs$ is nonzero,
  since each of the first $\ell$ entries of $\bv_c$ is nonzero. For sufficiently
  large $\ell$ (i.e., $\ell > 4$), only three combinations of elements $\bv_a,
  \bv_b \in \mc V_L$ have this property. Denoting the first $\ell$ entries of a
  vector $\bv$ as $\bv[:\ell]$, these three possible combinations are:
  \begin{align*}
    \bv_a[:\ell] &=
    \begin{pmatrix}
      0 \\
      \vdots \\
      0 \\
      \ell \\
      \ell
    \end{pmatrix}, ~~
    \bv_b[:\ell] =
    \begin{pmatrix}
      2 \ell \\
      \vdots \\
      2\ell \\
      0 \\
      0
    \end{pmatrix} \\
    \bv_a[:\ell] &=
    \begin{pmatrix}
      0 \\
      \vdots \\
      0 \\
      \ell \\
      \ell
    \end{pmatrix}, ~~
    \bv_b[:\ell] =
    \begin{pmatrix}
      2 \ell \\
      \vdots \\
      2\ell \\
      2\ell \\
      0
    \end{pmatrix} \\
    \bv_a[:\ell] &=
    \begin{pmatrix}
      0 \\
      \vdots \\
      0 \\
      0 \\
      2\ell \\
    \end{pmatrix}, ~~
    \bv_b[:\ell] =
    \begin{pmatrix}
      2 \ell \\
      \vdots \\
      2\ell \\
      2\ell \\
      0
    \end{pmatrix}
  \end{align*}
  This yields possible sums for the first $\ell$ entries:
  \begin{align*}
    (\bv_a + \bv_b)[:\ell] = \bs[:\ell] \in \left\{
      \begin{pmatrix}
        2\ell \\
        \vdots \\
        2\ell \\
        \ell \\
        \ell
      \end{pmatrix},
      \begin{pmatrix}
        2\ell \\
        \vdots \\
        2\ell \\
        3\ell \\
        \ell
      \end{pmatrix},
      \begin{pmatrix}
        2\ell \\
        \vdots \\
        2\ell \\
        2\ell \\
        2\ell
      \end{pmatrix}
    \right\}
  \end{align*}
  For the first 2 possibilities, it must be the case that the first $\ell$
  entries of $\bv_c$ are identical (because $\bv_c \notin \mc V_L$) and at most
  $\ell$ (because otherwise $\bv_d = \bs - \bv_c$ would contain negative
  entries). But we can verify that by construction, no $\bv_d$ exists that
  yields the desired sum. Thus, the only possible sum is $\bs = [2\ell \dots
  2\ell ~ 2 \dots 2]$. Given this choice of $\bs$, the only solution that
  includes the required $\bv_a$ and $\bv_b$ must include these two elements and
  $2\ell-2$ copies of $[0 \dots 0 ~ 1 \dots 1]$. This is $\bx^*$ as claimed.

  \paragraph*{$\mcX_L$ contains many solutions.}

  Next, we lower-bound $|\mcX_L|$. A fairly straightforward argument shows
  $|\mcX_L| \ge 2^{\lfloor \ell/2 \rfloor}$. To see this, note that for each pair
  of indices $i \ne j \le \ell$, there are two ways to make those entries
  $2\ell$ and the rest of the first $\ell$ entries 0:
  \begin{align*}
    \begin{pmatrix}
      2\ell \\
      0
    \end{pmatrix}
    +
    \begin{pmatrix}
      0 \\
      2\ell
    \end{pmatrix}
    &&
    \begin{pmatrix}
      \ell \\
      \ell
    \end{pmatrix}
    +
    \begin{pmatrix}
      \ell \\
      \ell
    \end{pmatrix}
  \end{align*}
  Because there are $\lfloor\ell/2\rfloor$ disjoint pairs of entries in the
  first $\ell$, there are at least $2^{\lfloor\ell/2\rfloor}$ distinct
  solutions.

  A more sophisticated argument shows that $|\mcX_L| =
  \Omega\p{\p{\frac{\ell}{e}}^\ell}$. Let $\mcX_L'$ be the set of solutions that
  only use the first $\ell + \binom{\ell}{2}$ columns of $\bV_\ell$. Since
  $\mcX_L' \subseteq \mcX_L$, $|\mcX_L'| \le |\mcX_L|$. Observe that each of the
  first $\ell + \binom{\ell}{2}$ columns of $\bV_\ell$ can be indexed by $(i,
  j)$ for $i \le j$, where $\bv_{(i,j)}$ is $\ell \times (\be_i + \be_j)$ for
  unit basis vectors $\be_i$ and $\be_j$. (A unit basis vector $\be_i$ has a 1
  in position $i$ and 0 elsewhere.) Any $\bx \in \mcX_L'$ can thus be expressed
  as a sequence of these ordered pairs.

  Further, any $\bx \in \mcX_L'$ must have exactly two copies of $\ell \times
  \be_i$ for each $i \in [\ell]$. For some solution $\bx \in \mcX_L'$, observe
  that beginning at an arbitrary $\bv_{(i,j)}$ in the solution, we can trace out
  a cycle $i \to j$, $j \to k$, \ldots, $z \to i$ simply by finding the item
  that provides the matching pair for each $\be_j$. Thus, each $\bx \in \mcX_L'$
  uniquely corresponds to a set of (undirected) cycles over $\ell$ elements.
  (For example, the appearance of $\bv_{(i, i)}$ corresponds to a self-loop, and
  $\bv_{(i, j)}$ appearing twice corresponds to a cycle of length 2.) Thus,
  $|\mcX_L'|$ is exactly the number of distinct such collections of cycles.

  Defining $a(\ell)$ to be the number of distinct collections of cycles over
  $\ell$ elements, \citet[Example 5.2.9]{stanley1999enumerative} shows that
  $a(\ell)$ satisfies
  \begin{align*}
    \sum_{\ell \ge 0} a(\ell) \frac{x^\ell}{\ell!} = (1-x)^{-1/2} \exp\p{\frac{x}{2} +
    \frac{x^2}{4}}.
  \end{align*}
  The sequence $a(0), a(1), \dots$ appears as OEIS sequence
  \href{https://oeis.org/A002135}{A002135} \citep{oeis}, and satisfies the
  recurrence $a(\ell) = \ell \cdot a(\ell-1) - \binom{\ell-1}{2} a(\ell-3)$.
  Asymptotically (c.f. Pietro Majer \citep{oeis}),
  \begin{align*}
    a(\ell) \sim \sqrt{2} \exp\p{\frac{3}{4}} \p{\frac{\ell}{e}}^{\ell} \p{1 +
    O\p{\frac{1}{\ell}}}.
  \end{align*}
  Thus, $|\mcX_L| \ge |\mcX_L'| = a(\ell) = \Omega\p{\p{\frac{\ell}{e}}^\ell}$.

  \paragraph*{Conductance and Cheeger's inequality.}

  With this, we are ready to bound the chain's conductance and thus its
  relaxation time. We need the following definitions. For our Markov chain
  $M_\ell$ with transition matrix $P_\ell$,
  \begin{align*}
    Q_\ell(\bx, \bx')
    &\triangleq \pi(\bx') P_\ell(\bx, \bx') \\
    Q_\ell(\mc S, \overline{\mc S})
    &\triangleq \sum_{\bx \in \mc S, \bx' \notin \mc S} Q_\ell(\bx,
    \bx') \\
    \Phi(\mc S, \overline{\mc S})
    &\triangleq \frac{Q_\ell(\mc S, \overline{\mc S})}{\pi(\mc S)} \\
    \Phi(M_\ell)
    &\triangleq \min_{\mc S : 0 < \pi(\mc S) \le 1/2} \Phi(\mc S, \overline{\mc
    S}) \\
    &\le \Phi(\mcX_L, \overline{\mcX_L}) \\
    &= \frac{\sum_{\bx \in \mcX_L, \bx' \notin \mcX_L} \pi(\bx) P_\ell(\bx,
    \bx')}{\pi(\mcX_L)} \\
    &\le \frac{\sum_{\bx \in \mcX_L, \bx' \notin \mcX_L}
    \pi(\bx)}{\pi(\mcX_L)} \tag{$P_\ell(\cdot, \cdot) \le 1$} \\
    &= \frac{\pi(\bx^*)}{\pi(\mcX_L)} \tag{$\bx^*$ is the only state with a
    crossing edge} \\
    &= \frac{1}{|\mcX_L|} \tag{$\pi(\bx) \propto 1$} \\
    &= O\p{\p{\frac{e}{\ell}}^\ell}
  \end{align*}
  By Cheeger's inequality~\citep{jerrum1988conductance,lawler1988bounds},
  \begin{align*}
    \lambda_2(P_\ell) \ge 1 - 2\Phi(M_\ell).
  \end{align*}
  This means that
  \begin{align*}
    \relaxation(P_\ell) = \frac{1}{1-\lambda_2(P_\ell)} \ge
    \frac{1}{2\Phi(M_\ell)} =
    \Omega\p{\p{\frac{\ell}{e}}^\ell}.
  \end{align*}
\end{proof}

\section{Final Algorithm Description}
\label{app:final-alg}

\begin{algorithm}
  \caption{\textproc{ILPandMCMC}($\bV, \bc, f, N, k, t$)}
  \label{alg:overall}
  \begin{algorithmic}
    \State $\mc S \gets$ subset of $\mcX$ of size at most $N$
    \Comment{Using ILP, with objective $L(\cdot)$ given by \Cref{eq:pi-approx}}
    \State $\bx \gets$ a random $\bx \in \mc S$ according to $\pi \given \bx \in
    \mc S$
    \If{$|\mc S| < N$}
    \Comment{In this case, $(\pi \given \bx \in \mc S) = \pi$}
    \State \Return $\bx$
    \Else
    \State \Return \Call{Reduced}{$\bV, \bc, f, k, t, \bx$}
    \EndIf
  \end{algorithmic}
\end{algorithm}

Recall that our upper bounds on $\numsampk$ relied on the fact that $\pi(\bx_0)
\ge 1/|\mcX|$. This may no longer be true when $\bx_0$ is sampled from $\pi
\given \bx \in \mc S$. Here, we provide a bound that only depends on
$\pi(\mc S)$, not $|\mc X|$. As long as $\pi(\mc S) \ge 1/|\mcX|$, our new bound
is stronger.
\overall*
\begin{proof}
  Let $\bX \sim \pi_0$ where $\pi_0 = \pi \given \bx \in \mc S$.
  Let $\Preduced^t(\bX, \cdot)$ be the distribution of over the reduced chain
  beginning with $\bX$ after $t$ steps. Following \citet[Theorem
  3.4]{sousi2020mixing}, for a reversible Markov chain with transition matrix
  $\Preduced$ and relaxation time $\relaxation$,
  \begin{align*}
    \numberthis \label{eq:l2-bound}
    \|\Preduced^t(\bX, \cdot) - \pi\|_{2,\pi} \le \exp(-t/\relaxation) \|\pi_0 -
  \pi\|_{2,\pi},
  \end{align*}
  where given a stationary distribution $\pi$, we define
  \begin{align*}
    \|\sigma - \pi\|_{2,\pi} \triangleq \p{\sum_{\bx \in
    \mcX}\p{\frac{\sigma(\bx)}{\pi(\bx)} - 1}^2 \pi(\bx)}^{1/2}.
  \end{align*}
  For initial distribution $\pi_0 = \pi \given \bx \in \mc S$,
  \begin{align*}
    \|\pi_0 - \pi\|_{2,\pi}^2
    &= \sum_{\bx \in \mcX} \p{\frac{\pi_0(\bx)}{\pi(\bx)} - 1}^2 \pi(\bx) \\
    &= \sum_{\bx \in \mcX} \p{\frac{\pi_0(\bx)^2}{\pi(\bx)^2} -
    \frac{2\pi_0(\bx)}{\pi(\bx)} + 1} \pi(\bx) \\
    &= \sum_{\bx \in \mcX} \frac{\pi_0(\bx)^2}{\pi(\bx)} - 2\pi_0(\bx) +
    \pi(\bx) \\
    &= \p{\sum_{\bx \in \mcX} \frac{\pi_0(\bx)^2}{\pi(\bx)}} - 1 \\
    &= \p{\sum_{\bx \in \mc S} \frac{\pi_0(\bx)^2}{\pi(\bx)} + \sum_{\bx \notin
    \mc S} \frac{\pi_0(\bx)^2}{\pi(\bx)}} - 1 \\
    &= \p{\sum_{\bx \in \mc S} \frac{(\pi(\bx)/\pi(\mc S))^2}{\pi(\bx)}} - 1 \\
    &= \p{\sum_{\bx \in \mc S} \frac{\pi(\bx)}{\pi(\mc S)^2}} - 1 \\
    &= \frac{1}{\pi(\mc S)} - 1 \\
    &\le \frac{1}{\pi(\mc S)}.
  \end{align*}
  Thus, by~\cref{eq:l2-bound},
  \begin{align*}
    \numberthis \label{eq:2bound}
    \|\Preduced^t(\bX, \cdot) - \pi\|_{2,\pi} \le \exp(-t/\relaxation)
    \frac{1}{\sqrt{\pi(\mc S)}}.
  \end{align*}
  Next, observe that for any $\sigma$,
  \begin{align*}
    \|\sigma - \pi\|_1^2
    &= \p{\sum_{\bx \in \mcX} |\sigma(\bx) - \pi(\bx)|}^2 \\
    &= \p{\sum_{\bx \in \mcX} \left|\frac{\sigma(\bx)}{\pi(\bx)} - 1\right|
    \pi(\bx)}^2 \\
    &\le \sum_{\bx \in \mcX} \p{\frac{\sigma(\bx)}{\pi(\bx)} - 1}^2 \pi(\bx)
    \tag{Jensen's inequality} \\
    &= \|\sigma - \pi\|_{2,\pi}^2.
  \end{align*}
  Combining this with \cref{eq:2bound},
  \begin{align*}
    2\dTV(\Preduced^t(\bX, \cdot), \pi))
    = \|\Preduced^t(\bX, \cdot) - \pi\|_1
    \le \|\Preduced^t(\bX, \cdot) - \pi\|_{2,\pi}
    \le \exp(-t/\relaxation) \frac{1}{\sqrt{\pi(\mc S)}}.
  \end{align*}
  To achieve the desired bound on $\dTV(\Preduced^t(\bX, \cdot), \pi)$, it
  suffices to set
  \begin{align*}
    \frac{1}{2} \exp(-t/\relaxation) \frac{1}{\sqrt{\pi(\mc S)}} &\le \varepsilon \\
    \exp(-t/\relaxation) &\le 2\varepsilon \sqrt{\pi(\mc S)} \\
    t &\ge \relaxation \log \p{\frac{1}{2\varepsilon \sqrt{\pi(\mc
    S)}}}.
  \end{align*}
\end{proof}

\end{document}